\documentclass[12pt]{article} 

\usepackage{natbib}

\RequirePackage[colorlinks,citecolor=blue,linkcolor=blue,urlcolor=blue,pagebackref]{hyperref}

\usepackage{amsmath,amsfonts,esint}
\usepackage{amsmath,amstext,rotating}
\usepackage{amsfonts,amssymb,graphics,xspace,endnotes}
\usepackage{lineno}
\usepackage{epsfig}
\usepackage{srcltx}
\usepackage{amsthm}
\usepackage{graphicx}
\usepackage{xcolor}	
\usepackage{verbatim}
\usepackage{natbib}
\usepackage{caption}
\usepackage{subcaption}
\usepackage{setspace}

\def \PP {\mathbb{P}}

\DeclareMathOperator*{\relint}{\mathrm{relint}}

\theoremstyle{plain}
\newtheorem{prop}{Proposition}

\newtheorem{assumption}{ASSUMPTION}
\newtheorem{theorem}{THEOREM}

\newtheorem{corollary}{COROLLARY}
\newtheorem{remark}{Remark}
\theoremstyle{remark}  
\newtheorem{definition}{DEFINITION}
\newtheorem{example}{EXAMPLE}

\usepackage{natbib}
\setcitestyle{sort&compress}

\AtBeginDocument{}

\hypersetup{
    colorlinks,
    linkcolor={blue!75!black},
    urlcolor={blue!75!black},
    citecolor={blue!75!black},
}

\begin{document} 

\title{Quantile Restrictions, Revealed Rankings, and the Limits of Multinomial Choice} 
\author{Tatiana Komarova\thanks{Faculty of Economics, University of Cambridge.  tk670@cam.ac.uk.}} 

\date{October 13, 2026}
\maketitle

\begin{abstract}
This paper analyzes when choice probabilities reveal rankings of deterministic utility indices in semiparametric discrete choice models. It begins with binary choice, showing that quantile thresholds can arise from behavioral departures from utility maximization under exchangeable unobservables, or from non-exchangeable unobservables under utility maximization. These behavioral and distributional routes are then extended to multinomial choice. Under limited attention, balance restrictions on attention probabilities give linear separators whose strict signs reveal rankings uniformly over admissible unobservable distributions and parameter values. Beyond these restrictions, opposite rankings can generate overlapping probability images. Under non-exchangeable unobservables, a common ranking partition across distributions need not exist. Under a common location-scale structure, both extensions admit an injective, generally nonlinear probability map that permits ranking recovery on its attainable image. The results distinguish distribution-robust ranking rules from recovery under specific structural restrictions, clarifying the limits of extending binary quantile restrictions to multinomial choice.

\vspace{0.05in}

\begin{description}
\item[Keywords:] Binary choice; Multinomial choice; Quantile restrictions; Probability simplex; Ranking; Ranking ambiguity; Consideration sets; Non-exchangeability
 
\end{description}
\end{abstract}
\setstretch{1.35}

\newpage

\section{Introduction}

Semiparametric discrete choice models often obtain identifying content from weak restrictions on the conditional distribution of unobservables rather than from a fully parametric specification. A classic example is the binary response model studied by \citet{Manski1985,Manski1988}, $Y=1(x\beta+\varepsilon\geq 0)$, together with the median restriction $Med(\varepsilon| x)=0$ and a conditional c.d.f. that is continuous and strictly increasing around zero. Under these conditions, observed choice probabilities reveal the sign of the deterministic utility index as in 
\begin{equation}
\label{medianECON}
P(Y=1| x)\geq (\leq)\frac12
\quad\Longleftrightarrow\quad
x\beta\geq (\leq)0.
\end{equation}
More generally,  under 
$Q_\tau(\varepsilon| x)=0$ and the
analogous c.d.f. condition, 
\begin{equation}
\label{quantilenECON}
P(Y=1| x)\geq (\leq)1-\tau
\quad\Longleftrightarrow\quad
x\beta\geq (\leq)0.
\end{equation}
The econometric content of this representation is well known. Geometrically,  in \eqref{quantilenECON} a single threshold $1-\tau$ divides the one-dimensional probability simplex  into two regions corresponding to the two possible rankings of the deterministic utilities. That is, once the location of $P(Y=1| x)$ relative to the threshold is known, the ranking is known as well.

This paper analyzes whether an analogous idea can be extended to multinomial choice. With $J>2$ alternatives, the vector of choice probabilities lies in a $(J-1)$-dimensional simplex and there is no longer a single probability threshold. I use the term “quantile” multinomial choice to describe a common ranking partition of a specified probability domain. The domain and its partition are common across admissible values of the index parameter. On the attainable part of this domain, the region containing a probability vector determines the ranking of the deterministic utility indices, including any ties. The ``quantile'' terminology is not meant to refer to any conventional multivariate quantile notion.

It is useful to distinguish the geometry of this partition from ranking recovery on the attainable probability set. Ranking recovery asks whether the same probability vector can arise under different index rankings. The geometric results additionally describe the boundaries between ranking regions on a specified common domain. A particular parameter value may attain only part of that domain, and whether the entire domain is attained as the parameter varies is a separate question.

The extension from two to several alternatives is not automatic. With three or more alternatives, the probability of choosing $j$ rather than $k$ can also be affected by what happens to other alternatives. As illustrated in the paper, generally probability vectors generated by different rankings of the deterministic indices can  overlap. In other words, the same observed probability vector obtained under some $x$ and $\widetilde{x}$ may then be compatible with both $x_j\beta_j>x_k\beta_k$ and $\widetilde{x}_j\beta_j<\widetilde{x}_k\beta_k$. This failure of ordinal recovery can be understood as \emph{ranking ambiguity}. One of the central questions of this  paper is what restrictions on the decision making or on the distribution of utility unobservables prevent such ambiguity.

I approach this question by first returning to the binary model and asking what economic structures can generate the quantile representation \eqref{quantilenECON}. In the classical random utility interpretation underlying the median case \eqref{medianECON}, agents maximize latent utility and the two utility unobservables are conditionally exchangeable.\footnote{Original \citet{Manski1975} restriction was  conditional i.i.d.  but \citet{GoereeHoltPalfrey2005} and later \citet{Fox2007} note that conditional exchangeability suffices for the rank order  result. In binary choice, exchangeability also implies  a zero median of the error difference.} Most semiparametric work takes the median or quantile restriction as a primitive. I  look into rationalizing  quantile restrictions in discrete choice models along two different routes. Either of them can generate representation  \eqref{quantilenECON}.  A \emph{behavioral route}  maintains conditional exchangeability of the latent utility unobservables but allows decision making to differ from  utility maximization. The decision rule may reflect limited consideration, noisy perception, or other stochastic behavior. A \emph{distributional route} maintains standard utility maximization but allows the joint distribution of the utility unobservables to be non-exchangeable. 

I then carry these  two routes into multinomial choice. The natural benchmark is the classical random utility model of \citet{Manski1975}. Under utility maximization with conditionally i.i.d. unobservables and some further standard \citet{Manski1975} conditions,
\begin{equation} 
P(Y=j|\mathbf x)\geq P(Y=k|\mathbf x)
\quad\Longleftrightarrow\quad
x_j\beta_j\geq x_k\beta_k\footnote{Again, \citet{GoereeHoltPalfrey2005} and \citet{Fox2007} show that conditional exchangeability is enough to ensure this ranking relationship.}, 
\label{Manskiranorder} 
\end{equation} 
which I will refer to as Manski's rank order property. Pairwise equalities of choice probabilities,  thus, describe  hyperplanes that split  the probability simplex into regions corresponding to unique complete rankings of the deterministic utility indices.  I use this classical case as the multinomial ``median'' benchmark and ask how much of its ranking geometry survives under limited attention and under non-exchangeable unobservables in latent utilities.

The paper makes three main substantive contributions and develops a framework for distinguishing the robustness of the resulting ranking claims. 

The first main contribution is to  give economic foundations for binary non-median (and median as a special case) probability thresholds. In line with the behavioral route mentioned above, 
I describe general stochastic decision rules that give \eqref{quantilenECON} under conditionally exchangeable unobservables and go beyond utility maximization. These  include  limited attention, noisy perception, and random choice, with exogenous consideration sets providing a particularly transparent foundation for both median and non-median cases. For the second distributional route, I show how non-median thresholds can instead arise from asymmetry in the joint distribution of latent utility unobservables with agents still maximizing utility. 

The second main contribution is to study ranking recovery and stronger ``quantile'' partitionings in multinomial models and to do so at different levels capturing different robustness properties. A \emph{Behavioral environment} level specifies a decision rule and a class of conditional distributions of unobservables. A \emph{Structural environment} additionally fixes one such conditional distribution.\footnote{The main text also mentions an \emph{empirical specification} level, which also fixes the marginal distribution of the covariates, but does not study findings at this level as it requires a case-by-case analysis.}  In addition to these two levels,  the paper also studies which results apply for a fixed $\boldsymbol\beta$ or uniformly over all admissible $\boldsymbol\beta$. Throughout, ranking recovery is understood relative to the stated environment. Empirical implementation additionally requires the objects determining the ranking rule to be known or separately identified. 

The sharpest results arise in the behavioral route when looking at the exogenous consideration sets. Limited attention changes the geometry of the probability simplex since choice probabilities now combine utility comparisons across different menus. Nevertheless, particular restrictions on attention  probabilities can restore linear pairwise ranking rules, generalizing \citet{Manski1975} partitioning. These are restrictive balancing conditions  on attention probabilities and they become increasingly more restrictive as $J$ increases. When these balance conditions are satisfied, the conclusions are strong as the pairwise boundaries are (i) linear, (ii) depend only on attention probabilities, and (iii) apply uniformly over admissible index parameters and distributions of unobservables. A probability vector lying strictly on either side of a boundary reveals the corresponding strict ranking. Under some additional conditions, the boundaries also distinguish index indifference from strict rankings. But when the balancing conditions on attention probabilities fail, the probability images of opposite pairwise rankings can overlap, producing ranking ambiguity.

To my knowledge, this paper provides the first characterization of exogenous attention structures supporting  two-coordinate linear ranking separators, which are distribution robust, together with sufficient conditions for separators involving additional probability coordinates.\footnote{These structures can be asymmetric and different from \citet{Manski1975} boundaries $\{p_k=p_{\ell}\}$.} 

The distributional route gives   a different conclusion. If the
conditional distribution of the utility unobservables is allowed to
vary over an admissible class of absolutely continuous, potentially
non-exchangeable distributions, there is in general no probability
partition that is common across that class.

The third main contribution is to show that ranking recovery and ``quantile'' partitionings can be restored under weaker robustness requirements of the \emph{Structural environment} level. In both the limited attention model and the model with non-exchangeable unobservables, a common location-scale structure induces a map from normalized deterministic utility index differences into choice probabilities. Under the conditions developed in the paper, this map is injective, so  the probabilities uniquely determine the
ranking of the deterministic utility indices within the given
\textit{Structural environment}. The map is generally nonlinear. Under limited attention it depends on both the attention probabilities
and the fixed distribution of the normalized unobservables, while under
utility maximization with non-exchangeable unobservables it depends on the fixed distribution of the unobservables. Importantly, it does not depend on the parameter vector $\boldsymbol{\beta}$. The fact that the map is common to the environment supports an interpretation of this case as \textit{Structural environment} ``quantile'' multinomial choice. Overall, as is evident from the paper, a ``quantile'' interpretation of multinomial choice  requires enough structure to generate a global ranking partition of the relevant probability domain. Whether that partition is  robust to the distribution of unobservables, or whether it is specific to a given  \emph{Structural environment} is a separate robustness question.

There is also a fundamental asymmetry between binary and multinomial cases. In binary choice, behavioral and distributional distortions act on a single margin and a quantile representation can survive under much weaker conditions than an analogous ``quantile'' representation in multinomial choice.

The paper contributes to several  literatures. It builds  on the semiparametric discrete choice framework of \citet{Manski1975,Manski1985,Manski1988}. It is also related to work on binary choice under quantile restrictions. \citet{Kordas2006} studies smoothed binary regression quantiles, while \citet{Volgushev2020} develops asymptotic theory for binary response quantile processes with linear quantile restrictions. That literature largely takes the quantile restriction as an econometric primitive. This paper instead asks which behavioral and distributional structures can generate the probability threshold and what remains of its ordinal interpretation when choice becomes multinomial. \citet{Matzkin1993} studies nonparametric identification of multinomial  choice models with flexible deterministic utility functions and weak distributional restrictions, but preserves the classical pairwise ranking geometry of  \citet{Manski1975}.

The paper is also related to the literature on consideration sets. \cite{Fox2007} also studies rank ordering in multinomial choice models with unobserved random consideration sets. His result preserves Manski's rank order property \eqref{Manskiranorder} by imposing restrictions on covariate-dependent consideration probabilities that essentially force the consideration mechanism to reinforce Manski's  \eqref{Manskiranorder}. In particular, in binary choice his assumptions force the model to preserve the   median threshold and does not permit other quantiles. In contrast, this paper holds attention probabilities fixed and characterizes the balance conditions under which ordinal recovery nevertheless emerges from the attention structure itself, allowing the relevant separators to differ from the standard Manski's comparison $\{p_k=p_{\ell}\}$. 

\citet{BarseghyanMolinariThirkettle2021}  develop a discrete choice model with heterogeneous risk preferences and unobserved consideration sets and provide conditions for its identification. In comparing limited consideration with standard random utility models, they introduce a conditional analogue of Manski’s rank order property \eqref{Manskiranorder} (conditional on latent  risk preference type) and show that it implies a generalized dominance restriction on observed choice probabilities, a restriction that limited consideration can violate. An earlier working paper \citet{BarseghyanMolinariThirkettle2019} also studies a Random Consideration Level (RCL)  model and shows that it preserves the conditional Manski's rank order property. RCL imposes a strong symmetry restriction, where  conditional on consideration set size (including singletons), all sets are equally likely, so swapping any two alternatives leaves the probability of the consideration set unchanged. This symmetry is a very special case of the balance restrictions in this paper. The present paper allows asymmetric attention, characterizes the restrictions under which distribution-robust ordinal recovery remains possible, and derives the resulting (potentially non-Manski) linear probability boundaries. 

\citet{Allen2026} shows that, with bounded observable utility shifters, exogenous limited attention can be observationally equivalent to full consideration, but that question is different from this paper's direct study on when limited attention preserves ordinal recovery. 

Other relevant consideration sets literature includes
\citet{AbaluckAdamsPrassl2021}, who identify consideration probabilities
from asymmetric demand responses,
\citet{MasatliogluNakajimaOzbay2012},
\citet{ManziniMariotti2014},
\citet{CattaneoMaMasatliogluSuleymanov2020}, and
\citet{AguiarBoccardiKashaevKim2023}. \citet{DardanoniManziniMariottiTyson2020} study recovery of cognitive heterogeneity from aggregate choice shares, showing that a single sufficiently large menu can suffice under homogeneous preferences. 
The question here is different. Rather than treating attention 
probabilities as the principal object to be identified, I
take the consideration structure as given and ask what ordinal
information about deterministic utilities is encoded in the resulting
choice probability geometry. In addition, the object of interest is not a single fixed preference ranking, but the ranking of deterministic utility indices as covariates vary. The relevant variation thus comes from the movement of the choice probabilities across covariate values and from whether the resulting probability images associated with different rankings remain separated.

Finally, the results at the \emph{Structural environment} level are somewhat related to single crossing arguments in
\citet{ApesteguiaBallesterLu2017} and
\citet{BarseghyanMolinariThirkettle2021}. Those papers use single crossing restrictions to identify preference heterogeneity and, in the latter case, consideration mechanisms. Here the \textit{Structural environment} fixes the unobservables distribution and any attention probabilities. I ask when choice probabilities reveal deterministic index rankings within that environment. Under the stated conditions, a common location-scale structure yields an injective probability map and hence ranking recovery, without imposing single crossing on preference heterogeneity.

The rest of the paper is organized as follows. Section~\ref{sec:binary} studies binary choice and develops behavioral and distributional foundations for median and more general quantile threshold representations. Section~\ref{sec:multinomial} turns to multinomial choice, introduces the ranking recovery framework and the levels of robustness, and then studies limited attention and non-exchangeable utility unobservables. Section~\ref{sec:conclusion} concludes.

\section{Binary choice}
\label{sec:binary}

A standard starting point, including in  \citet{Manski1985, Manski1988}, for a semiparametric binary choice model is the formulation $y=1(x\beta+\varepsilon \geq 0)$, $x \in \mathbb{R}^{1 \times k}$, $\beta \in \mathbb{R}^{k \times 1}$, with $x\beta$ describing the part of utility depending on the observed characteristics and potentially identifiable and $\varepsilon$ describing its unobservable component. The subsequent analysis of such a model relies on what is assumed about the relationship between  $\varepsilon$ and covariates $x$.  A canonical restriction of that relationship comes from \citet{Manski1985, Manski1988} and relies on imposing 
\begin{equation} \label{median} Med(\varepsilon|x)=0.
\end{equation}
It permits rich dependence of the conditional distribution of $\varepsilon$ on $x$ while linking utility indices $x\beta$ to observed choice probabilities $P(Y=1\mid x)$, as discussed below.

The motivation for the median restriction (\ref{median}) comes from   \citet{Manski1975}  random utility framework in which each option is associated with a latent utility and the agent chooses the option with the higher utility. In the binary case, let $j=0,1$ and define latent utilities as $U_j = x_j \beta_j + \varepsilon_j$, $x_j \in \mathbb{R}^{1 \times k_j}$, $\beta_j \in \mathbb{R}^{k_j \times 1}$ for $j=0,1$. 
In the standard utility optimization the agent chooses option 1 if and only if $ U_1 \geq U_0 \iff \varepsilon_1 - \varepsilon_0 \geq -x_1\beta_1 + x_0\beta_0.$ 
Under the assumption that $(\varepsilon_0,\varepsilon_1)$ are i.i.d. conditional on $(x_0,x_1)$, the difference $\varepsilon_1 - \varepsilon_0$ has zero conditional median. Support convexity additionally  guarantees that the c.d.f. of this difference conditional on $(x_0,x_1)$ is strictly increasing in the neighborhood of 0. Absolute continuity of $(\varepsilon_0,\varepsilon_1)$ guarantees that the c.d.f. of $\varepsilon_1 - \varepsilon_0$ is continuous.

In what follows, define \(x\beta = x_1\beta_1 - x_0\beta_0\), where \(x\) collects all distinct covariates appearing in \(x_1\) or \(x_0\). The coefficient vector \(\beta\) assigns to each covariate the difference between its coefficients in \(x_1\beta_1\) and \(x_0\beta_0\), using \(0\) for absent covariates. Thus, \(x\) consolidates \((x_1,x_0)\) without redundancy. In particular, if \(x_1\) and \(x_0\) share no covariates, then
\(
x=(x_1,x_0)\), $\beta=(\beta_1^{\top},-\beta_0^{\top})^{\top}$. 
Conditioning on \(x\) is informationally equivalent to conditioning on \((x_1,x_0)\). Denoting 
$\varepsilon=\varepsilon_1-\varepsilon_0,$ we end up with the model
\begin{align}
\label{modelBC}
   Y =1(x\beta+\varepsilon \geq 0), & \quad Med(\varepsilon|x)=0,\\
    \text{c.d.f.} F_{\varepsilon|x}(\cdot|x)  \text{ is continuous and strictly } & \text{increasing around 0}.
   \label{modelBC_part2}
\end{align}
(\ref{modelBC})-(\ref{modelBC_part2}) implies that observed choice probabilities \( P(Y=1 | x) \) are informative about the sign of the utility index in the following sense: $P(Y=1 | x) \;  \geq (\leq)  \; P(Y=0 | x) \iff x\beta \geq (\leq) \; 0$,
or equivalently, \eqref{medianECON}. 
Characterization (\ref{medianECON}) underlies both identification arguments and estimation methods in semiparametric binary choice models.

\citet{Manski1988} proposed a quantile-based semiparametric binary choice model:
\begin{equation}
\label{modelBCquantile}
Y = 1(x\beta + \varepsilon \geq 0), \quad Q_{\tau}(\varepsilon | x) = 0,
\end{equation}
together with the same continuity and strict increase condition on the
conditional c.d.f. at zero as above, but the quantile index \( \tau \in (0,1) \) need not equal \( 1/2 \). Choice probabilities are then informative as \eqref{quantilenECON} holds. Behavioral or distributional motivations for \eqref{modelBCquantile} or \eqref{quantilenECON} have received very little attention in the literature.

\begin{remark}\label{remark1} 
Unlike \citet{Manski1985,Manski1988}, who work directly with \(U_1-U_0\), I keep \(U_1\) and \(U_0\) separate, allowing asymmetry between \(\varepsilon_1\) and \(\varepsilon_0\) and more general pre-choice evaluation of alternatives.
\end{remark}

\begin{remark}\label{remark2} 
Since the  focus here is not identification of utility indices, the linear indices \(x_j\beta_j\), \(j=0,1\), could be replaced by general functions \(\phi_j(x_j)\), with \(x\beta\) replaced by \(\phi_1(x_1)-\phi_0(x_0)\). Linearity is mostly maintained for comparability with \citet{Manski1975,Manski1985,Manski1988}.
\end{remark} 

\begin{remark}\label{remark3}
By the standard random utility paradigm, this paper means a setting in which the agent chooses the option with the highest realized latent utility.
\end{remark}

\subsection{Behavioral route: General theorem (any $\tau$)} 
\label{sec:median} 

We start by imposing the following assumption, which permits more general conditional dependence between \(\varepsilon_1\) and \(\varepsilon_0\) than the conditional i.i.d. structure in \citet{Manski1975}.

\begin{assumption}[Distribution of $(\varepsilon_1,\varepsilon_0)|x$]
\label{assn:binary_distributionMED}
Conditional on $x$:
\begin{enumerate}
    \item[(a)] $\varepsilon_1$ and $\varepsilon_0$ have the same
    absolutely continuous marginal distribution.
    \item[(b)] Their copula $C_x$ is exchangeable and absolutely continuous.
    \item[(c)] Their joint support is convex with nonempty interior.
\end{enumerate}
\end{assumption}
Thus, the unobservables may be conditionally correlated, with dependence restricted to exchangeable copulas. This nests the conditional independence copula of \citet{Manski1975}.

\begin{remark}
Throughout, the aim is to derive conditions that are \emph{generic}, in the sense that they do not rely on particular features of the conditional distribution of the unobservables given $x$, nor on the marginal distribution of $x$ or on the value of $\beta$. To that end, this section abstracts  from specifications in which restrictions on the joint support of $(x,\varepsilon)$ limit the range of the index $x\beta$ relative to the support of $\varepsilon$. Such cases can be viewed as non-generic, as they can typically be resolved by enriching the model by, for instance, expanding the support of $x$ or adjusting $\beta$ so that $x\beta$ spans the relevant range of $\varepsilon$. 

The necessity or near-necessity results below are therefore stated for this enriched class, which we call \emph{generic models}.\footnote{In much of the classic literature, the model is effectively specified in this enriched form: \(\varepsilon| x\) has full support on \(\mathbb{R}\), and \(x\beta\) varies over \(\mathbb{R}\) with \(x\). These assumptions are analytically convenient but often can be relaxed.} In the multinomial models of Section~\ref{sec:multinomial}, this idea is formalized through different layers of analysis.
\end{remark}

The theorem below gives  behavioral conditions for 
\eqref{quantilenECON} for any \(\tau \in (0,1)\) under Assumption~\ref{assn:binary_distributionMED}, while allowing departures from utility maximization. It nests the median $\tau=1/2$ case ~\eqref{medianECON}, again allowing deviations from the standard random utility paradigm.

Throughout this section, restrictions on choice probabilities conditional on utility rankings (such as $U_1 \geq U_0$)  refer to fixed measurable versions. If a conditioning event has probability zero given \(x\), its conditional choice probability must admit a value in \([0,1]\) satisfying the stated restrictions.

\begin{theorem}[General theorem for any $\tau \in (0,1)$] \label{prop:median}
Let $U_j = x_j\beta_j+\varepsilon_j$, $j=0,1$, and let the distribution of $(\varepsilon_1, \varepsilon_0) $ satisfy Assumption \ref{assn:binary_distributionMED}. 

\textbf{A.} The choice behavior of an economic  agent can be represented as in (\ref{quantilenECON}) if  the following  conditions are satisfied: 
\begin{enumerate}
\item[(1)] $P(Y=1| U_1 \geq U_0,x) -P(Y=1|U_0 > U_1,x) >0$ a.e.; 
\item[(2)] $P(Y=1| U_1 \geq U_0,x) - P(Y=0|U_0 > U_1,x) = 1-2\tau$  a.e., or equivalently, $P(Y=1 | U_1 \ge U_0, x) + P(Y=1 | U_0 > U_1, x) = 2(1-\tau)$ a.e.
\end{enumerate}

\textbf{B.} If the choice behavior of an economic agent can be represented as in (\ref{quantilenECON}) and Condition (2) holds, then Condition (1) is necessary away from index indifference:  
\begin{equation}
\label{condition1_weaker}
P(Y=1 | U_1 \ge U_0, x) - P(Y=1 | U_0 > U_1, x) > 0
\quad \text{a.e. on  } 
\{x:x\beta\neq0\}
\end{equation}

\textbf{C.} Condition (2) cannot be omitted from the distribution-free sufficiency result in Part A. More precisely, for every \(\tau\in(0,1)\) and every \(\beta \neq 0\), there exists a generic model satisfying Assumption 1 and the strict version of Condition (1), but violating Condition (2) on a set of positive measure, for which \eqref{quantilenECON} fails.

\end{theorem}

Condition (1)   rules out choice rules so perverse that higher utility of an option
reduces the probability of choosing it. This guarantees the expected monotonic relationship since  as the index $x\beta$ crosses zero from below, $P(Y=1|x)$ crosses the threshold $1-\tau$ from below.\footnote{This monotonicity condition is distinct from the single crossing property in random utility models (\citet{ApesteguiaBallesterLu2017}). Single crossing imposes an ordering on preference heterogeneity such that rankings over alternatives change at most once across types, whereas the present condition is a much weaker reduced form restriction on choice probabilities. } Condition (2) relaxes the standard random utility paradigm requirement 
$P(Y=1 | U_1 \geq U_0, x)
=
P(Y=0 | U_0 > U_1, x)
=1$ and allows choice reversals while fixing the difference between the probabilities of choosing correctly under the two utility rankings at \(1-2\tau\).  For \(\tau\neq 1/2\), the theorem provides a behavioral foundation for the representation \eqref{quantilenECON}.\footnote{If Condition (2) holds but Condition (1) fails, monotonicity may reverse, yielding the unintended implication \(P(Y=1| x)\leq (\geq) 1-\tau \iff x\beta\geq (\leq) 0\).}

Part A gives sufficient conditions for \eqref{quantilenECON}. Conditional on Condition (2), Part B establishes necessity of Condition (1) away from index indifference. Part C shows that Condition (1) alone does not guarantee \eqref{quantilenECON}. 

Theorem \ref{prop:median} establishes the representation \eqref{quantilenECON}, but it is important to note that under behavioral departures from utility maximization, the error in a representation \eqref{modelBCquantile} need not coincide with the original utility difference \(\epsilon_1-\epsilon_0\), whose conditional median remains zero under Assumption \ref{assn:binary_distributionMED}. Also note that Theorem \ref{prop:median} does not require choices to depend only on \(U_1-U_0\). Choices may depend on the levels of \(U_1\) and \(U_0\) as well, provided the theorem's two conditions hold with  examples of such behavioral models  given below.

\begin{example} \label{ex:randchoice}  1. \textit{Mixture of random choice and utility maximization, $\tau=1/2$}
For an economic agent characterized by $x$, with probability $r(x) \in (0,1)$, the choice between 1 and 0 is made by tossing a fair coin, and with probability $1-r(x)$, the choice is made by maximizing utility. Then $P(Y=1| U_1\geq U_0,x) = r(x)/2 + (1-r(x)) = 1-r(x)/2$ and
$P(Y=1| U_0>U_1,x) = r(x)/2$.
Condition (1) holds since $1-r(x)/2 - r(x)/2 = 1-r(x)>0$. Condition (2) holds for $\tau=1/2$ since  $P(Y=1|U_1 \geq U_0, x)$ and $P(Y=1|U_1 < U_0, x)$ sum to 1.

2. \textit{Noisy Perception, $\tau=1/2$.} An agent does not know their $U_j$ and instead perceives their noisy versions $\xi\cdot U_j$ where $\xi$ is independent of
$(U_0,U_1)$ conditional on $x$ with $P(\xi\geq 0| x) > P(\xi<0| x)$, $P(\xi=0| x)=0$ \footnote{Equivalently, $P(\xi > 0 |  x)>1/2$. This is not a strong requirement especially taking into account that $\xi$ centered around 1 would give us the disturbed latent utilities centered around true $U_1,U_0$.} ($\xi$ and $x$ can be dependent). Agent's choice is made by maximizing these noisy versions. Hence,    $P(Y=1|U_1 \geq U_0, x)   = P( \xi \cdot U_1 \geq \xi \cdot U_0 |U_1 \geq U_0, x) = P(\xi \geq 0 |  x)$ and $P(Y=1|U_1 < U_0, x)   =  P(\xi < 0 |  x)$. Condition (1) in Theorem \ref{prop:median} is guaranteed by  $P(\xi \geq 0 |  x) - P(\xi < 0 | x)>0$. Condition (2) holds too a.e. for $\tau=1/2$ as   $P(Y=1|U_1 \geq U_0, x)$ and $P(Y=1|U_1 < U_0, x)$ sum to 1. $\blacksquare$
\end{example}

Neither of these two examples give  behavioral models that lead to \eqref{quantilenECON} with $\tau \neq 1/2$. Therefore, the next subsection  focuses  on more canonical behavioral models that maintain Assumption \ref{assn:binary_distributionMED} and deliver conditions (1) and (2) of Theorem \ref{prop:median} for $\tau \neq 1/2$. Before that, let's summarize some implications of model (\ref{quantilenECON}) on the choice probabilities. 

\begin{prop}\label{prop:implications}
 Consider the setting  in Theorem \ref{prop:median} and suppose conditions (1) and (2) there hold. Then 
 $P(Y=1|U_1<U_0,x) \leq P(Y=1|x) \leq P(Y=1|U_1\geq U_0,x)$ a.e., and (i) for $\tau>1/2$, 
$P(Y=1|x) \leq 2(1-\tau)$; 
(ii) for $\tau<1/2$, $P(Y=1|x) \geq 2(1-\tau)-1$.
\end{prop}
Proposition~\ref{prop:implications} yields a testable implication: for $\tau\neq 1/2$, choice probabilities are bounded strictly away from at least one of 0 and 1. In contrast, the median model permits probabilities arbitrarily close to both boundaries on sets of positive (but decreasing) measure. Thus, such behavior is incompatible with the behavioral route model of Theorem~\ref{prop:median} when $\tau\neq 1/2$.

\subsection{Consideration sets as the canonical behavioral model} 

As it turns out, models with consideration sets provide the canonical behavioral implementation of
Theorem~\ref{prop:median} for $\tau \neq 1/2$.

\vskip 0.05in 

\noindent \emph{Independent consideration set formation} Suppose the collection of possible consideration sets is 
$\{\{0\}, \{1\}, \{0,1\}\}$, with probabilities 
$\gamma_0$, $\gamma_1$, and 
$1-\gamma_0-\gamma_1 \in (0,1)$, respectively. These consideration sets are assumed to be 
formed independently of $(U_0,U_1,x)$. When the agent faces the menu 
$\{0,1\}$, they choose the alternative with the higher utility. The model 
can be written as $Y = z_1z_2 + (1-z_1)\mathbf{1}\{U_1 \geq U_0\}$, where $z_1$ and $z_2$ are independent binary random variables, also independent 
of $(U_0,U_1,x)$, with $P(z_1=1)=\gamma_0+\gamma_1$, 
$P(z_2=1)=\frac{\gamma_1}{\gamma_0+\gamma_1}$. Here, $z_1$ indicates that the consideration set is a singleton, while $z_2$ 
indicates that this singleton is $\{1\}$. Thus, independence of $z_1$ and 
$z_2$ from $(U_0,U_1,x)$ captures the exogeneity of the attention process.

Then 
$P(Y=1|x)=\gamma_1 +(1-\gamma_0 -\gamma_1) (1-F_{\varepsilon|x} (-x\beta)) $ 
and, under Assumption \ref{assn:binary_distributionMED}, 
$$x\beta \; \, \geq (\leq) \; \, 0 \iff P(Y=1|x) \; \, \geq (\leq) \; \, \frac{1}{2}+\frac{1}{2}(\gamma_1-\gamma_0).$$
If $\gamma_0=\gamma_1$ then we in fact obtain a model with the median characterization \eqref{medianECON} which is special case of the model in Example \ref{ex:randchoice}. If $\gamma_0\neq \gamma_1$ we obtain a quantile characterization (\ref{quantilenECON}) for $\tau=1/2(1-\gamma_1+\gamma_0)\neq 1/2$. The asymmetry in consideration probabilities $\gamma_1 - \gamma_0$ directly
controls the quantile index.

It is instructive to equivalently rewrite the above relation as 
$x\beta \; \, \geq (\leq) \; \, 0 \iff P(Y=1|x) - \left(\gamma_1 +\frac{1-\gamma_0-\gamma_1}{2}\right)\; \, \geq (\leq) \; \, 0,$ 
where $P(Y=1|x)$ is measured relative to its reference value $\gamma_1 +\frac{1-\gamma_0-\gamma_1}{2}$ obtained under index indifference $x_1\beta_1=x_0\beta_0$. This reference value determines the location of the probability partition. We use the same perspective in the multinomial consideration set model later on. 

The consideration sets model presented here  has testable implications beyond Proposition \ref{prop:implications}, including when $\gamma_1 = \gamma_0$. Namely, $P(Y=1|x) \in [\gamma_1,1-\gamma_0]$, so  choice probabilities are
bounded away from zero if $\gamma_1>0$, from one if $\gamma_0>0$, and
from both boundaries if both are positive. Since
$\gamma_0+\gamma_1>0$ under our maintained conditions, they are always
bounded away from at least one boundary. Hence, if $P(Y=1| x)$ lies
arbitrarily close to both zero and one with positive probability over
$x$, this model is ruled out. 

When $\gamma_0=0$, we have $z_2=1$ a.e., and the only consideration sets are $\{1\}$ and $\{0,1\}$. This case captures a predisposition toward option 1, with $\gamma_1$ measuring its strength. Equivalently, it can represent probabilistic conformity to a social norm being option 1.

 Even though this paper does not deal with identification of $\beta$ or attention probabilities and instead focuses on ranking recovery, 
 note that under standard large support conditions ensuring
that $P(U_1\geq U_0|x)$ can approach both 0 and 1, 
the limiting choice probabilities identify $\gamma_1$ and $\gamma_0$, and hence, $\tau$. With the usual scale normalization and support conditions,
identification of $\beta$ can then proceed using standard
\citet{Manski1985,Manski1988} arguments.

\vskip 0.05in

\noindent \textit{Heterogeneous consideration set formation} What if we weakened the assumption on $z_1, z_2$ by allowing them to be correlated with $x$ but maintained their independence from $(U_0,U_1)$ and from each other given $x$? Let \(\gamma_0(x),\gamma_1(x)\geq0\), with \(\gamma_0(x)+\gamma_1(x)<1\) a.e. Set $P(z_1=1|x)=\gamma_0(x)+\gamma_1(x)$, 
$P(z_2=1| x)=\frac{\gamma_1(x)}{\gamma_0(x)+\gamma_1(x)}$,
taking \(z_2=0\) whenever the denominator is zero. We obtain $P(Y=1|x)=\gamma_1(x) +(1-\gamma_0(x) -\gamma_1(x)) (1-F_{\varepsilon|x} (-x\beta)) $, 
and under Assumption \ref{assn:binary_distributionMED}, 
$x\beta \; \geq (\leq) \; 0 \iff P(Y=1|x) \; \geq  (\leq)  \; \frac{1}{2}+\frac{1}{2}(\gamma_1(x)-\gamma_0(x)).$ 
This equivalence uses \(1-\gamma_0(x)-\gamma_1(x)>0\) a.e..  
Thus, models with heterogeneous attention process, where $\gamma_0$ and
$\gamma_1$ depend on $x$, can generate \eqref{quantilenECON}, provided
$\gamma_1(x) - \gamma_0(x) = \Delta$ a.e.\ for some constant $\Delta \in (-1,1)$. If in this condition $\Delta>0$, then  $P(Y=1|x)$ potentially can approach  1 arbitrarily 
closely but not 0. If $\Delta<0$, then  $P(Y=1|x)$ potentially can approach  0 arbitrarily 
closely but not 1. Thus, when this heterogeneous choice formation  gives model (\ref{quantilenECON}) for $\tau \neq 1/2$, it weakens the testable implications of independent consideration set formation. Probabilities now need only be bounded away from one boundary, as in Proposition \ref{prop:implications}.

The question now is whether there are other foundations that lead to (\ref{quantilenECON}), particularly those  where the choice probabilities could approach both natural probability boundaries 0 and 1 arbitrarily closely for $\tau \neq 1/2$. This is addressed in the next section.

\subsection{Distributional route: Non-exchangeable distribution of $(\varepsilon_0,\varepsilon_1|x)$} 

We now show that \eqref{quantilenECON} can also be rationalized within the 
standard random utility paradigm by allowing the joint distribution of 
$(\varepsilon_0,\varepsilon_1)$ to be asymmetric. The following theorem covers 
both non-exchangeable copulas and unequal marginal distributions. It  
allows for behavior beyond the standard random utility paradigm, but our main interest 
remains in this paradigm where the role of distributional asymmetries in the 
unobservables is especially transparent.

\begin{theorem}\label{prop:copula} 
Let $U_j=x_j\beta_j+\varepsilon_j$, $j=0,1$,   and let the joint distribution of $(\varepsilon_0, \varepsilon_1)|x $ be absolutely continuous, have  convex support with nonempty interior and let the support of $\varepsilon_1- \varepsilon_0|x $ contain $0$ in its interior. Then: 

\textbf{A.} 
The choice behavior  can be represented as in  (\ref{quantilenECON}) if the following conditions hold: 
\begin{enumerate}
    \item[(1')] Condition (1) in Theorem \ref{prop:median}  holds; 
    \item[(2')] $P(\varepsilon_1 \geq \varepsilon_0|x)   P(Y=1|U_1\geq U_0,x) + P(\varepsilon_1 < \varepsilon_0|x)   P(Y=1|U_1 < U_0,x)=1-\tau$  a.e.  
\end{enumerate}

\textbf{B.} If the choice behavior of an economic agent can be represented as in \eqref{quantilenECON} and condition (2') holds, then Condition (1') is necessary away from index indifference, as in
\eqref{condition1_weaker}. 

\textbf{C.}
Condition (2') cannot be omitted from the distribution-free sufficiency
result in Part \textbf{A}. More precisely, for every fixed
$\beta\neq 0$, there exists a generic model satisfying the maintained
assumptions and the strict version of Condition (1'), but violating
Condition (2') on a set of positive $P_X$-measure, for which
\eqref{quantilenECON} fails.
\end{theorem} 

(2') suggests that the covariates $x$  and the error dependence  interact in a way that stabilizes the threshold for the choice probability. Just like in Theorem \ref{prop:median}, Part A gives sufficient conditions, Part B establishes necessity of  (1') away from index indifference conditional on (2'), and Part C shows that (1') alone does not guarantee \eqref{quantilenECON}. 

 Condition (2') reduces to Condition (2) when the joint distribution of
$(\varepsilon_0,\varepsilon_1)$ is exchangeable, since
$P(\varepsilon_1\geq\varepsilon_0| x) = 1/2$ in that case.
Under the standard random utility paradigm
($P(Y=1|U_1\geq U_0,x) = 1$, $P(Y=1| U_1<U_0,x) = 0$),
Condition (2') becomes $P(\varepsilon_1\geq\varepsilon_0 | x) = 1-\tau$
a.e., which is the constancy across $x$ of the asymmetry of the joint unobservable distribution. When the marginals of $\varepsilon_0$ and $\varepsilon_1$ conditional on $x$
are identical (Assumption \ref{assn:binary_distributionMED}(a) holds) but
the copula $C_x$ is not exchangeable, this becomes
$P_{C_x}(v\geq u| x) = 1-\tau$ a.e., that is, the copula must have a fixed
degree of asymmetry across $x$. Appendix B gives an explicit construction showing that any
$\tau\in(0,1)$ is attainable with identical marginals under the maintained support conditions.

The key difference of the described distributional route from the behavioral route is that no boundary restriction
on $P(Y=1| x)$ emerges as choice probabilities can approach both zero and
one, as the copula non-exchangeability can accommodate any
probability level.

\section{Multinomial choice: More than two options} 
\label{sec:multinomial}

Binary case characterization \eqref{quantilenECON} effectively described  a geometric partitioning of the one-dimensional simplex $\Delta_1=\{p_0+p_1 =1, p_1 \in [0,1]\}$ (in the model of Theorem~\ref{prop:median},  it would more precisely be the partitioning of its feasible subset $\{p_0+p_1 =1, p_1 \in [\max\{0,2(1-\tau)-1\}, \min\{2(1-\tau),1\}]\}$).  With more than two options, the analysis inevitably becomes substantially richer and more complex. This section studies the extent to which the partition result of the binary case survives under more general behavioral and distributional assumptions.

\subsection{Review of \citet{Manski1975} multinomial model:    ``median'' multinomial choice}
\label{sec:reviewManski} 

I begin with the classical \citet{Manski1975} multinomial choice model. I consider it the ``median'' benchmark for multinomial choice, as the binary choice model \eqref{medianECON} under the median threshold arises as its  special case. 

The latent utility of an agent $i$ for the option $j$ is 
\begin{equation} 
\label{Umulti} U^*_{ij} = x_{ij}\beta_j+\varepsilon_{ij},  \quad x_{ij} \in \mathbb{R}^{1 \times k_j}, \quad \beta_j \in \mathbb{R}^{k_j \times 1}, \quad j=1, \ldots, J.
\end{equation}

On a general note, and not specifically in relation to \citet{Manski1975}, it is useful to distinguish the support of the covariates from the range of the deterministic utility indices.  Even if each $x_j$ has full support, the vector $(x_1\beta_1,\ldots,x_J\beta_J)^{\top}$ need not span $\mathbb R^J$, e.g. when covariates are shared across alternatives. Moreover, since choices and rankings depend only on relative utilities, the relevant richness concerns the $J-1$ index differences rather than the $J$ index levels. Sufficient ranking results later in this paper apply on the index configurations that are attainable. This includes the special case in which the $J-1$ index differences range over all of $\mathbb R^{J-1}$, but does not require it. Additional richness of the joint index difference  will be imposed only where it is needed for converse results or for conclusions to be  ``generic'' in some sense.

Under the standard random utility paradigm, conditionally i.i.d. $\boldsymbol{\varepsilon}=(\varepsilon_{i1}, \ldots, \varepsilon_{iJ})^{\top}$ conditional on $\mathbf{x}_i=(x_{i1},\ldots,x_{iJ})^{\top}$ with an absolutely continuous  joint distribution and support being $\mathbb{R}^J$,  
\begin{equation} 
\label{ECONmedianmulti} \forall \, j, k \qquad P(Y=j|\mathbf{x}_i) \geq P(Y=k|\mathbf{x}_i) \quad \iff \quad x_{ij}\beta_j \geq x_{ik}\beta_k.
\end{equation} 
Relations \eqref{ECONmedianmulti} imply that $P(Y=j|\mathbf{x}_i) = P(Y=k|\mathbf{x}_i)$ $  \iff$ $x_{ij}\beta_j =x_{ik}\beta_k$. Note that the full $\mathbb{R}^J$ support assumption on the conditional distribution of $\boldsymbol{\varepsilon}|\mathbf{x}$ guarantees that the vector $(P(Y=1|\mathbf{x}),\ldots,P(Y=J|\mathbf{x}))^{\top}$ never hits the boundary of the $(J-1)$-dimensional simplex  
$\Delta_{J-1}=\{\mathbf{p}=(p_1,\ldots,p_J): \sum_{j=1}^J p_j=1, p_j \geq 0, j=1, \ldots, J\}.$ 
Without that full support assumption on the distribution of $\boldsymbol{\varepsilon}|\mathbf{x}$, the index differences may reach or extend beyond the corresponding error difference supports,  in which case  inequality relationships between choice probabilities may become uninformative about the corresponding relationships between the associated utility indices.
Therefore, when working under the more general distributional assumptions, it would be natural to reformulate   \eqref{ECONmedianmulti} by 
excluding points in the simplex that are obtained when some index differences $x_j\beta_j-x_k\beta_k$ are either at the boundary or beyond the support of $\varepsilon_k-\varepsilon_j|\mathbf{x}$. This will be done in a technically proper way later in our discussion. For now, for expositional simplicity let us maintain that full support condition.

There are some observations  worth noting here. In Manski's setting, we can define preference relations over the choice set either in terms of probabilities,
    $ j \stackrel{(p)}{\succeq} k \text{ by } i \; \iff \; P(Y=j|\mathbf{x}_i) \geq P(Y=k|\mathbf{x}_i),$ 
or in terms of indices,  
    $ j \stackrel{(index)}{\succeq} k \text{ by } i \; \iff \;  x_{ij}\beta_j \geq x_{ik}\beta_k.$ 
These two preference relations are equivalent. Moreover, they satisfy the standard axioms of preference relations, including reflexivity, completeness and transitivity. Another implication of that is that we can create an ordering of the whole list of options by utilizing pairwise relations described above. Namely, take any permutation $(k_1,\ldots,k_J)$ of $(1,\ldots, J)$. Then $k_1 \stackrel{(index)}{\succeq} k_2 \ldots \stackrel{(index)}{\succeq} \ldots \stackrel{(index)}{\succeq} k_J $ by $i$ iff $k_{j-1} \stackrel{(p)}{\succeq} k_j $ by $i$ for any $j=2, \ldots, J.$

The following theorem essentially shows that \eqref{ECONmedianmulti} is consistent with
a wider class of decision rules than standard utility maximization. Before that, let's formulate an assumption on unobservables.

\begin{assumption}[Distribution of $\boldsymbol{\varepsilon}|\mathbf{x}$] \label{assn:multi_distributionMED} 
 The joint distribution $\boldsymbol{\varepsilon}|\mathbf{x}$ is exchangeable and absolutely continuous and has a convex support on $\mathbb{R}^J$ with a non-empty interior.  
\end{assumption} 
The absolute continuity of the joint distribution in Assumption \ref{assn:multi_distributionMED} guarantees that ties $U^*_j=U^*_{k}$ for $k \neq j$ happen with probability 0 and, thus, can  be ignored in our derivations.

In Theorem \ref{th:multimedian} below  the maximum of latent utilities $U_j$, $j=1, \ldots, J$, is denoted as $U^{J:J}$.

\begin{theorem}
\label{th:multimedian}  Let $U_j$ , $j=1, \ldots, J$,  be given as in (\ref{Umulti}), and let the distribution of $\boldsymbol{\varepsilon}|\mathbf{x}$ satisfy Assumption \ref{assn:multi_distributionMED}. Then the  behavior of an economic agent can be represented as
{\small$$x_{j}\beta_j \geq x_{k}\beta_k \Rightarrow P(Y=j|\mathbf{x}) \geq P(Y=k|\mathbf{x}), \quad P(Y=j|\mathbf{x}) > P(Y=k|\mathbf{x}) \Rightarrow x_{j}\beta_j > x_{k}\beta_k$$}if the following conditions hold a.e. in $\mathbf{x}$:
\begin{enumerate}
\item[(1)] For any $m$ such that $P(U_m=U^{J:J}|\mathbf{x})>0$ and $j\neq m$, $P (Y = m|U_m = U^{J:J}, \mathbf{x}) > P (Y = j|U_m =  U^{J:J}, \mathbf{x})$. 
\item[(2)] For any $m$ such that $P(U_m=U^{J:J}|\mathbf{x})>0$, $P(Y=m| U_m=U^{J:J},\mathbf{x})$ does not depend on
$m$ among alternatives satisfying
$P(U_k=U^{J:J}|\mathbf{x})>0$. 
\item[(3)] $P(Y=j| U_m=U^{J:J},\mathbf{x})$, $j\neq m$, does not
depend on $(j,m)$ among pairs for which
$P(U_m=U^{J:J}|\mathbf{x})>0$.   
\end{enumerate}
\end{theorem}
This theorem contains a modified version of \eqref{ECONmedianmulti}. The modification matters when some index differences reach or extend beyond the supports of the corresponding error differences, which can occur when the support of the unobservables is bounded in some directions. For brevity and also due to multiplicity of options,  I do not focus on near-necessity formulations (analogous to \textbf{B} and \textbf{C} in theorems \ref{prop:median} and \ref{prop:copula}). \citet{Manski1975} is, of course, a special case of Theorem \ref{th:multimedian} where the joint conditional copula $C_{\mathbf{x}}$ is the independence copula and $P (Y = j|U_j = U^{J:J}, \mathbf{x})=1$, $P (Y = j|U_m = U^{J:J}, \mathbf{x})=0$, $m\neq j$, a.e.

Example \ref{ex:simple_multi} below is that of multinomial choice models that satisfy the sufficient conditions of Theorem \ref{th:multimedian} but go beyond the standard random utility paradigm. 



\begin{example} 
\label{ex:simple_multi}
1. \textit{Mistakes in the implementation}.   With probability $\alpha(\mathbf{x})\in(1/J,1)$, the agent correctly
chooses the utility-maximizing option, otherwise she randomizes uniformly over the
nonmaximizing alternatives.   
This model allows for mistakes in implementation even when the agent correctly identifies the utility-maximizing alternative.  Conditions (1)-(3) of Theorem \ref{th:multimedian} are satisfied as 
 $P(Y=j\mid U_j=U^{J:J},\mathbf{x})=\alpha(\mathbf{x})$, while for any $m\neq j$,
we have $P(Y=j\mid U_m=U^{J:J},\mathbf{x})=\frac{1-\alpha(\mathbf{x})}{J-1}$.  Condition 1 holds as $\alpha(\mathbf{x})>\frac{1}{J}$ and, hence, $\alpha(\mathbf{x}) - (1-\alpha(\mathbf{x}))/(J-1) > 0$.

2. \textit{Random shortlist around the utility maximizer.} 
Suppose the agent identifies the utility-maximizing alternative and forms
a shortlist of fixed size $r\in\{2,\ldots,J-1\}$ that always contains
that alternative. The remaining $r-1$ members are drawn uniformly without
replacement from the other $J-1$ alternatives. The agent then chooses
uniformly from the shortlist. If $U_j=U^{J:J}$, then $P(Y=j\mid U_j=U^{J:J},\mathbf{x})=\frac{1}{r}$. If instead $U_m=U^{J:J}$ for some $m\neq j$, alternative $j$ enters the
shortlist with probability $(r-1)/(J-1)$, and hence $P(Y=j\mid U_m=U^{J:J},\mathbf{x})
=
\frac{r-1}{r(J-1)}$. Conditions (2) and (3) of Theorem~\ref{th:multimedian} therefore hold,
while Condition (1) follows from $\frac{1}{r}-\frac{r-1}{r(J-1)}
= \frac{J-r}{r(J-1)}>0.$

\end{example}

\subsection{Geometric structure of the ``median''  multinomial model}

To make progress on what we can consider a ``quantile'' multinomial choice model, we give a geometric interpretation of the \citet{Manski1975} model reviewed in Section \ref{sec:reviewManski} as we associate it with a ``median'' benchmark for multinomial choice. 

Since in \citet{Manski1975}  an individual always makes a choice,  the vector of choice probabilities $ (P(Y=1 |\mathbf{x}), \ldots, P(Y=J|\mathbf{x}))^{\top}$ lies in the $(J-1)$-dimensional simplex $\Delta_{J-1}$, as already mentioned above. The geometry of \eqref{ECONmedianmulti} is striking. First, under total indifference \begin{equation}
    \label{totalindiff}
x_1\beta_1=\ldots =x_J 
\beta_J 
\end{equation} in the index space,   such an index collection  is mapped by the model to the barycentre $\pi^*:=(1/J \ldots,1/J)^{\top} \in \Delta_{J-1}$, as implied by exchangeability in Assumption \ref{assn:multi_distributionMED}.

Under partial pairwise indifference 
{\small\begin{equation}
    \label{partialindiff}
x_{k_1}\beta_{k_1}=x_{k_2}\beta_{k_2}, \quad k_1 \neq k_2,
\end{equation}}model
\eqref{ECONmedianmulti} implies $P(Y=k_1|\mathbf{x}) = P(Y=k_2|\mathbf{x})$.
Pairwise indifference therefore maps into the hyperplane
$\{p_{k_1}=p_{k_2}\}$ in $\Delta_{J-1}$. 
Taken across all pairs, these hyperplanes divide the simplex into
$J!$ regions with disjoint interiors. The interior of each region
corresponds to one complete strict ordering of the indices, while its
boundary contains the corresponding cases with ties.  The attainable probability set inherits these regions (but need not fill them). The
closures are polyhedral cones with common vertex
$\boldsymbol{\pi}^*$, truncated at the simplex boundary (see Figure~\ref{fig:simplex1}).

Figure \ref{fig:simplex1} illustrates the six ranking regions for \(J=3\). In general, neighboring regions differ by swapping two adjacent alternatives in the ranking. This arrangement is the multinomial ``median'' benchmark. A ``quantile'' version may instead have a shifted reference point and different, possibly nonlinear, ranking boundaries.

\begin{figure}[htbp] 
\centering
\includegraphics[width=0.5\linewidth]{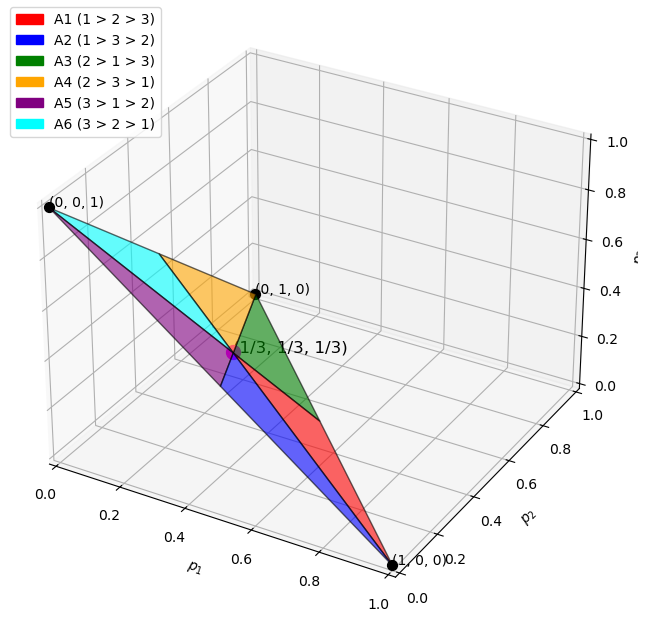}
\caption{Geometric illustration of \citet{Manski1975} for $J=3$.  Each region corresponds to one complete ordering with its interior capturing the respective 
strict ordering, and the region boundaries capturing ties.}
\label{fig:simplex1}
\end{figure}

\subsection{Two levels of analysis. Notation and preliminaries}
\label{sec:notation_preliminaries}

The analysis is conducted at two levels, which differ in what is held fixed
and what is allowed to vary. Let $\mathcal{B}_0=\{\|\boldsymbol{\beta}\|=1\}$ denote the admissible,
normalized parameter space for $\boldsymbol{\beta}
=
(\beta_1^{\top},\ldots,\beta_J^{\top})^{\top}$.

A \emph{Behavioral environment}, denoted by $B$, specifies a decision
rule, such as standard random utility paradigm or a consideration set model with
given attention  probabilities, and an admissible class
$\mathcal{P}$ of conditional distributions of
$\boldsymbol{\varepsilon}$ given $\mathbf{x}$. Results at this level are uniform over the 
conditional distribution of unobservables in $\mathcal{P}$. A \emph{Structural environment}, denoted by $M$, further fixes a conditional
distribution $\mathbb{P}_{\boldsymbol{\varepsilon}\mid\mathbf{x},\,
\mathbf{x}\in\mathbb{R}^{M_e}}
\in\mathcal{P}$. It can therefore be written as $M
=
\left(
B,
\mathbb{P}_{\boldsymbol{\varepsilon}\mid\mathbf{x},\,
\mathbf{x}\in\mathbb{R}^{M_e}}
\right)$. The conditional distribution of $\boldsymbol{\varepsilon}$ is specified for
every $\mathbf{x}\in\mathbb{R}^{M_e}$, where
$M_e$ is the effective dimension of $\mathbf{x}$ (thus, if e.g. the same covariate enters several indices, it appears only once in $\mathbf{x}$, and perfect linear combinations can be dealt with analogously). This allows conclusions that are not driven by restrictions on the covariate domain.\footnote{This echoes the paper's discussion of enriched support in the binary choice case. At the \emph{Structural environment} level, the support of $\mathbf{x}$ has already been maximally enriched.}
The class $\mathcal{P}$ may be defined by imposing suitable restrictions on
these conditional distributions (analogously to \cite{Manski1975}). 

In principle, one could consider an  \emph{empirical specification} level, which further fixes the marginal
distribution of $\mathbf{x}$ and is  characterized by $\left(
M,\boldsymbol{\beta},\mathbb{P}_{\mathbf{x}}
\right)$.\footnote{If the support of $\mathbb{P}_{\mathbf{x}}$ is smaller than
$\mathbb{R}^{M_e}$, some of the conditional distributions specified by $\mathbb{P}_{\boldsymbol{\varepsilon}\mid\mathbf{x},\,
\mathbf{x}\in\mathbb{R}^{M_e}}$
are not used by the empirical specification.} I don't consider this level theoretically as it needs to be approached on a case-by-case basis.

For any $\boldsymbol{\beta}\in\mathcal{B}_0$, I call $(B,\boldsymbol{\beta})$
a \emph{fixed-$\boldsymbol{\beta}$ behavioral specification}, and $(M,\boldsymbol{\beta})$
a \emph{fixed-$\boldsymbol{\beta}$ structural model}. The scope of results  with respect to $\boldsymbol{\beta}$
is a separate matter. A behavioral specification level or structural model 
result may hold e.g. for a fixed $\boldsymbol{\beta}$, or generically over
$\boldsymbol{\beta}$,  or uniformly over 
$\boldsymbol{\beta}\in\mathcal{B}_0$. Uniformity over $\boldsymbol{\beta}$ is particularly attractive as it means that the same probability based
ranking implication applies for every
$\boldsymbol{\beta}\in\mathcal{B}_0$. This is stronger than establishing,
separately for each $\boldsymbol{\beta}$, a ranking implication that may
depend on the particular parameter value. 

The result of \citet{Manski1975} is at the \emph{Behavioral environment}  level and
has the uniformity-in-$\beta$ scope.  Given that in empirical practice $\boldsymbol{\beta}\in\mathcal{B}_0$ is not known, the uniform-in-$\boldsymbol{\beta}$ results may therefore prove most valuable for a potential subsequent identification analysis.\footnote{If the deterministic component of latent utility were  by a more
general function $\phi_j(x_j)$ rather than $x_j\beta_j$, the analogous
distinction would be between a behavioral environment that leaves
$(\phi_1,\ldots,\phi_J)$ unspecified and a behavioral specification that
fixes those functions. Uniformity over the deterministic utility functions
would then require the same probability based ranking implication to hold
throughout the admissible class of such functions.} 
Section~\ref{sec:binary}  results for binary choice are also  at the \emph{Behavioral environment} level and  uniform in $\beta$.

\vskip 0.05in

\noindent
\underline{\textit{Simplex subsets of interest.}}
Fix a structural environment $M$ and a parameter value
$\boldsymbol{\beta}\in\mathcal{B}_0$. For the fixed-$\boldsymbol{\beta}$
structural model $(M,\boldsymbol{\beta})$, let $\mathbf{p}(\mathbf{x})
=
\left(
P(Y=1|\mathbf{x}),
\ldots,
P(Y=J|\mathbf{x})
\right)^{\top}
\in\Delta_{J-1}$
denote the conditional choice probability vector generated at
$\mathbf{x}$. To simplify notation, we suppress the dependence of
$\mathbf{p}(\mathbf{x})$ on $M$ and $\boldsymbol{\beta}$ whenever the
fixed structural model is clear from context.

Define the fixed-$\boldsymbol{\beta}$ feasible probability set at the \emph{Structural environment} level by
$\Delta^{M,\boldsymbol{\beta}}_{J-1}
=
\operatorname{Cl}
\left\{
\mathbf{p}(\mathbf{x})
:
\mathbf{x}\in\mathbb{R}^{M_e}
\right\}$. Thus, $\Delta^{M,\boldsymbol{\beta}}_{J-1}$ is the closure of the set of
probability vectors attainable under $(M,\boldsymbol{\beta})$ as the
covariates vary over the full covariate domain. The ultimate feasible probability set associated with the structural
environment $M$, allowing the preference parameter $\boldsymbol{\beta}$ to vary, is
$\Delta^{M,\cup_{\boldsymbol{\beta}}}_{J-1}
=
\operatorname{Cl}
\left(
\bigcup_{\boldsymbol{\beta}\in\mathcal{B}_0}
\Delta^{M,\boldsymbol{\beta}}_{J-1}
\right).$
This set is the closure of the probability vectors attainable for at least one admissible value of $\boldsymbol{\beta}$. 

We can define the corresponding objects at the behavioral specification level with 
$\Delta^{B,\boldsymbol{\beta}}_{J-1}
=
\operatorname{Cl}
\left(
\bigcup_{Q\in\mathcal{P}}
\Delta^{M_Q,\boldsymbol{\beta}}_{J-1}
\right)$
aggregating feasible probability vectors across all conditional 
distributions of unobservables admitted by $B$, while holding $\boldsymbol{\beta}$ fixed. The corresponding ultimate feasible set is
$\Delta^{B,\cup_{\boldsymbol{\beta}}}_{J-1}
=
\operatorname{Cl}
\left(
\bigcup_{\boldsymbol{\beta}\in\mathcal{B}_0}
\Delta^{B,\boldsymbol{\beta}}_{J-1}
\right)$.

Under some behavioral environments, the sets $\Delta^{M,\boldsymbol{\beta}}_{J-1}$,
$\Delta^{B,\boldsymbol{\beta}}_{J-1}$, $\Delta^{M,\cup_{\boldsymbol{\beta}}}_{J-1}$,
$\Delta^{B,\cup_{\boldsymbol{\beta}}}_{J-1}$
may have affine dimension strictly smaller than $J-1$, even when the
conditional support of
$\boldsymbol{\varepsilon}|\mathbf{x}$ is convex and has nonempty interior
in $\mathbb{R}^{J}$. It will be clarified later in the paper under what conditions the relevant
feasible set has full affine dimension.

\vskip 0.05in

\noindent
\underline{\textit{Images of pairwise ordered utility indices.}} I study whether an observed probability vector (unambiguously) reveals the ordering of the underlying deterministic utility indices. Since this question can be posed either conditionally on a fixed value of $\boldsymbol{\beta}$ or uniformly across values of $\boldsymbol{\beta}$,  the corresponding image sets need to be defined separately.

For any $k_1,k_2\in\mathcal{J}$ with $k_1\neq k_2$, and for any
 \emph{fixed-$\boldsymbol{\beta}$ structural model} $(M,\boldsymbol{\beta})$, define
\begin{equation}
\label{eq:defPk1k2_Model}
\mathcal{P}^{+,M,\boldsymbol{\beta}}_{k_1,k_2}
=
\left\{
\mathbf{p}\in\Delta_{J-1}:
\exists\,\mathbf{x}\in\mathbb{R}^{M_e}
\text{ s.t. }
\mathbf{p}(\mathbf{x})=\mathbf{p}, \; \; 
x_{k_1}\beta_{k_1}
>
x_{k_2}\beta_{k_2}
\right\},
\end{equation}
\begin{equation}
\label{eq:defIk1k2_Model}
\mathcal{I}^{M,\boldsymbol{\beta}}_{k_1,k_2}
=
\left\{
\mathbf{p}\in\Delta_{J-1}:
\exists\,\mathbf{x}\in\mathbb{R}^{M_e}
\text{ s.t. }
\mathbf{p}(\mathbf{x})=\mathbf{p}, \; \; 
x_{k_1}\beta_{k_1}
=
x_{k_2}\beta_{k_2}
\right\}.
\end{equation} 
Thus, $\mathcal{P}^{+,M,\boldsymbol{\beta}}_{k_1,k_2}$ collects the choice
probability vectors generated under the \emph{fixed-$\boldsymbol{\beta}$ structural
model} when 
$k_1\stackrel{\mathrm{index}}{\succ}k_2$,
whereas $\mathcal{I}^{M,\boldsymbol{\beta}}_{k_1,k_2}$ collects those
generated under pairwise index indifference. When $\boldsymbol{\beta}$ is allowed to vary within the structural
environment $M$, define the ultimate ordering and indifference images by $\mathcal{P}^{+,M,\cup_{\boldsymbol{\beta}}}_{k_1,k_2}
=
\bigcup_{\boldsymbol{\beta}\in\mathcal{B}_0}
\mathcal{P}^{+,M,\boldsymbol{\beta}}_{k_1,k_2}$ and $\mathcal{I}^{M,\cup_{\boldsymbol{\beta}}}_{k_1,k_2}
=
\bigcup_{\boldsymbol{\beta}\in\mathcal{B}_0}
\mathcal{I}^{M,\boldsymbol{\beta}}_{k_1,k_2}$, respectively. 

For a \emph{behavioral specification} level  for a fixed
$\boldsymbol{\beta}\in\mathcal{B}_0$, define $\mathcal{P}^{+,B,\boldsymbol{\beta}}_{k_1,k_2}
=
\bigcup_{Q\in\mathcal{P}}
\mathcal{P}^{+,M_{Q},\boldsymbol{\beta}}_{k_1,k_2}$, $\mathcal{I}^{B,\boldsymbol{\beta}}_{k_1,k_2}
=
\bigcup_{Q\in\mathcal{P}}
\mathcal{I}^{M_{Q},\boldsymbol{\beta}}_{k_1,k_2}$. 
These sets aggregate the fixed-$\boldsymbol{\beta}$ ordering and indifference
images across all structural models  admitted by $B$. When $\boldsymbol{\beta}$ is allowed to vary within $\mathcal{B}_0$, define the ultimate behavioral specification images by $\mathcal{P}^{+,B,\cup_{\boldsymbol{\beta}}}_{k_1,k_2}
=
\bigcup_{\boldsymbol{\beta}\in\mathcal{B}_0}
\mathcal{P}^{+,B,\boldsymbol{\beta}}_{k_1,k_2}$,
$\mathcal{I}^{B,\cup_{\boldsymbol{\beta}}}_{k_1,k_2}
=
\bigcup_{\boldsymbol{\beta}\in\mathcal{B}_0}
\mathcal{I}^{B,\boldsymbol{\beta}}_{k_1,k_2}$.

\vskip 0.05in 

\noindent \underline{\textit{Support of unobservables and relative interiors.}} 
\noindent
If the \citet{Manski1975} setting has $\mathbb{R}^J$ as the
support of $\boldsymbol{\varepsilon}|\mathbf{x}$, every generated choice probability vector lies in
$\relint(\Delta_{J-1})$. More generally, with bounded supports, choice probabilities on the
boundary of the relevant feasible probability set may be attained, and
distinct index configurations may generate the same boundary probability
vector. Accordingly, some results below restrict attention to the relative interior of the relevant feasible set. Such restrictions are imposed in the individual results, rather than in the definition of ranking recovery.

\vskip 0.05in 

\noindent
\underline{\textit{Ranking recovery value of choice probabilities.}} For a set $\mathcal Q$ of attainable probability vectors, say that the
ranking of alternatives $k_1$ and $k_2$ is recovered on $\mathcal Q$ if 
$\mathcal{P}^{+,\bullet}_{k_1,k_2}\cap\mathcal Q$, 
$\mathcal{I}^{\bullet}_{k_1,k_2}\cap\mathcal Q$, 
$\mathcal{P}^{+,\bullet}_{k_2,k_1}\cap\mathcal Q$
are pairwise disjoint, where $\bullet$ stands for a level of the analysis. Equivalently, every $\mathbf p\in\mathcal Q$
reveals the sign of
$x_{k_1}\beta_{k_1}-x_{k_2}\beta_{k_2}$ within the relevant environment
and parameter scope. Fixed-\(\beta\) recovery uses the ordering and indifference images indexed by \(\boldsymbol \beta\), and uniform recovery over \(\boldsymbol \beta\) uses their unions over admissible values of \(\boldsymbol \beta\). In either case, recovery is defined on a specified  \(Q\) consisting of probability vectors actually attained in the relevant environment and parameter scope.

Note that ranking recovery notion is given relative to the stated \textit{Behavioral} or \textit{Structural
environment}. Empirical implementation would additionally require the objects
determining the ranking rule to be known or separately identified.

Below are the implications of fixed-$\boldsymbol \beta$ ranking recovery on $Q$ at the \textit{Structural environment} level. 

\begin{enumerate}
\item[\textit{Impl 1.}]
The strict ordering images do not overlap: $\mathcal{P}^{+,M,\boldsymbol{\beta}}_{k_1,k_2}
\cap
\mathcal{P}^{+,M,\boldsymbol{\beta}}_{k_2,k_1}
\cap
Q
=
\varnothing$. Thus, for $\mathbf p=\mathbf p(\mathbf x)=\mathbf p(\widetilde{\mathbf x}) \in Q$, the same probability vector cannot imply both $x_{k_1}\beta_{k_1}>x_{k_2}\beta_{k_2}$ and $\widetilde{x}_{k_1}\beta_{k_1} <\widetilde{x}_{k_2}\beta_{k_2}$.

\item[\textit{Impl 2}.]
Suppose fixed-$\boldsymbol{\beta}$ recovery holds for every distinct pair of
alternatives. For any permutation
$(k_1,\ldots,k_J)$ of $\mathcal{J}$, we have $\mathbf{p}
\in
\bigcap_{j=1}^{J-1}
\mathcal{P}^{+,M,\boldsymbol{\beta}}_{k_j,k_{j+1}}
\cap
Q$
implies the complete strict ordering $x_{k_1}\beta_{k_1}
>
x_{k_2}\beta_{k_2}
>
\cdots
>
x_{k_J}\beta_{k_J}$ in the given \emph{Structural environment} $M$.  
Within $Q$, regions corresponding to different complete strict orderings are pairwise disjoint.

\item[\textit{Impl 3}.]
For $\mathbf p\in Q$, define $k_1
\succ_{\mathbf{p}}^{M,\boldsymbol{\beta}}
k_2
\iff
\mathbf{p}
\in
\mathcal{P}^{+,M,\boldsymbol{\beta}}_{k_1,k_2}.$ 
If fixed-$\boldsymbol{\beta}$ recovery holds for every pair, then 
$
\mathbf{p}
\in
\mathcal{P}^{+,M,\boldsymbol{\beta}}_{k_1,k_2}
\cap
\mathcal{P}^{+,M,\boldsymbol{\beta}}_{k_2,k_3}
\; \Rightarrow \; 
\mathbf{p}
\in
\mathcal{P}^{+,M,\boldsymbol{\beta}}_{k_1,k_3}.
$ 
The relation $\succ_{\mathbf{p}}^{M,\boldsymbol{\beta}}$ is therefore transitive. 

\end{enumerate}

\begin{table}[tbp]
\centering
\caption{Roadmap of multinomial ranking results}
\label{tab:multinomial_roadmap}
\footnotesize
\setlength{\tabcolsep}{3pt}
\renewcommand{\arraystretch}{1.08}
\begin{tabular}{@{}p{0.17\textwidth}p{0.21\textwidth}p{0.37\textwidth}p{0.19\textwidth}@{}}
\hline
\textbf{Robustness level}
&
\textbf{Restriction}
&
\textbf{Ranking result}
&
\textbf{$\boldsymbol\beta$ scope}
\\
\hline

\multicolumn{4}{@{}l}{\textit{Behavioral route: limited attention with exchangeable unobservables}}
\\[0.15em]

Behavioral environment
&
Rank-one balance
&
Common two-coordinate linear separators whose strict signs reveal pairwise rankings. Exact sign equivalence holds under the additional conditions in Theorem \ref{th:linear_generalJ}. For $J\geq4$, Theorem~\ref{th:linear_other_coordinates} gives sufficient conditions for more general linear separators using additional probability coordinates. 
&
The ranking rule is common to all admissible $\boldsymbol\beta$; the
attainable ranking regions need not be.
\\[0.45em]

Behavioral environment
&
Rank conditions fail
&
Opposite pairwise rankings may generate overlapping probability images. 
With sufficiently rich index difference variation, the corresponding two-coordinate and distribution-robust linear representation fails. 
&
The failure concerns robustness across admissible conditional unobservables 
distributions; uniformity over $\boldsymbol\beta$ is a separate requirement. 
\\[0.45em]

Structural environment
&
Common location-scale structure and connectivity
&
A common injective, generally nonlinear map $\Psi$ recovers the
normalized indices, and hence all pairwise rankings, from every
probability vector in $\Psi(\mathcal D_{\boldsymbol\beta})$.
&
$\Psi$ is common across $\boldsymbol\beta$; changing
$\boldsymbol\beta$ changes only $\mathcal D_{\boldsymbol\beta}$ and, hence, its probability space trace.
\\
\hline

\multicolumn{4}{@{}l}{\textit{Distributional route: utility maximization with non-exchangeable unobservables}}
\\[0.15em]

Behavioral environment
&
General admissible distribution class
&
No generic probability-space ranking rule is common across the admissible
unobservables distributions.
&
In general there is no distribution- and $\boldsymbol\beta$-uniform
ranking partition.
\\[0.45em]

Structural environment
&
Common location-scale structure
&
A common injective, generally nonlinear map $\Psi$ gives ranking
recovery on $\Psi(\mathcal D_{\boldsymbol\beta})$, although its geometry
may be asymmetric.
&
$\Psi$ is common across $\boldsymbol\beta$; the attainable ranking regions
and their geometry depend on $\mathcal D_{\boldsymbol\beta}$.
\\
\hline
\end{tabular}
\end{table}

Uniform recovery across $\boldsymbol{\beta}$  rules out cross-$\boldsymbol{\beta}$ ambiguity and ensures the same probability
vector cannot be generated with $k_1$ ranked above $k_2$ under one admissible
value of $\boldsymbol{\beta}$ and with $k_2$ ranked above $k_1$ under another.
Analogous implications apply at the \emph{Behavioral environment} level.

\vskip 0.05in 

 The \citet{Manski1975}  benchmark summarized in   \eqref{Manskiranorder} gives ranking recovery at the \textit{Behavioral environment} level, uniformly over \(\boldsymbol \beta\in B_0\). The following analysis examines how far this property survives the behavioral and distributional extensions. Table~\ref{tab:multinomial_roadmap} summarizes the main results across robustness levels in the subsequent analysis.

\subsection{ Consideration sets: the behavioral route. ``Quantile'' version?}

I now analyze a generalization of the consideration set behavioral route from binary to multinomial cases. I will first aim to derive results at the \textit{Behavioural environment} level, which itself is described in the next section.  

\subsubsection{Behavioral environment}

I start by formulating a set of basic assumptions on the consideration set process. 

\begin{assumption}
    \label{assn:indCS} The consideration set formation is independent of $(\mathbf{x}, \boldsymbol{\varepsilon})$.
\end{assumption} 
Recall that an analogous independence assumption was also our starting point in the binary choice setting. 
Under Assumption~\ref{assn:indCS}, denote by $\gamma_A$ the
probability of the consideration set $A\subseteq \mathcal{J}$  and write
$P(Y=j|\mathbf{x}) = \sum_{A\in\mathcal{A}:j\in A}\gamma_A P(Y=j| A,\mathbf{x})$, 
where $\mathcal{A}$ denotes the collection of all consideration sets and $P(Y = j | A, \mathbf{x})$ denotes the probability of choosing $j$ when considering a given menu $A$. Consideration sets  $A$ with  $\gamma_A>0$ are called \textit{active}. 

\begin{assumption}\label{assn:noemptyCS} Suppose that $\gamma_{\varnothing}=0$\footnote{Some existing work allows for empty consideration sets by interpreting this event as leading to a default or no-choice outcome (see e.g. \citet{ManziniMariotti2014}, \citet{Horan2019}). We could allow $\gamma_{\varnothing}>0$ and specify the corresponding choice rule, but doing so would not substantively affect our results. It would only introduce additional notation and case distinctions. For this reason, we set $\gamma_{\varnothing}=0$ throughout.} and every element $j \in \mathcal{J}$ belongs to some active consideration set.  
\end{assumption}

\begin{assumption}\label{assn:rationalCS} Given a menu $A$, an agent $i$ chooses an option $j$ from $A$ with the highest latent utility $U_{ij}^*=x_{ij}\beta_j+\varepsilon_{ij}$.   
\end{assumption}

Assumptions \ref{assn:noemptyCS} and \ref{assn:rationalCS} guarantee that an economic agent always makes a choice so $\sum_{j=1}^J P(Y=j|\mathbf{x})=1$ for all $\mathbf{x}$. They allow us to look, once again,  at the probability simplex $\Delta_{J-1}$ to characterize preference relations in terms of points from this simplex. The second part of Assumption \ref{assn:noemptyCS}  ensures that every alternative can, in principle, be chosen with a positive probability. 

By \textit{Behavioral environment}  throughout the consideration set discussion we will mean (i) the consideration set paradigm that satisfies  Assumptions \ref{assn:indCS}--\ref{assn:rationalCS} with a given collection $\{\gamma_A\}_{A \in \mathcal{A}}$\footnote{This collection may be unknown to an econometrician.};  (ii) family $\mathcal{P}$ of conditional distributions $\mathbb{P}_{\boldsymbol{\varepsilon}|\mathbf{x}, \mathbf{x} \in \mathbb{R}^{M_e}}$ that collects all joint distributions that satisfy Assumption \ref{assn:multi_distributionMED}.

\vskip 0.1in 

\noindent \textit{Consideration set structure connectivity.} Finally, I give a definition of connectedness of elements in a given consideration sets  structure. Its importance is discussed later.

\begin{definition} 
\label{def:cs_connected}
Consider  exogenous consideration set structure for $\mathcal{J}$ described by probabilities $\{\gamma_{A}\}_{A \in \mathcal{A}}$. Define the set of active non-singleton menus as 
$\mathcal{A}_{\geq 2} = \{ A \in \mathcal{A} : |A| \geq 2, \gamma_A > 0 \}.$ 

We will say that \textit{options $\ell$ and $h$ are connected} in the structure if either $\ell=h$, or there exists a sequence of menus $ A_1, A_2, \ldots, A_m \in \mathcal{A}_{\geq 2} $ such that:
$ \ell \in A_1 $,
$ h \in A_m $,
$ A_k \cap A_{k+1} \neq \emptyset $ for $ k = 1, 2, \ldots, m-1 $.
This sequence forms a path from $ \ell $ to $ h $. 

We will say that the \textit{consideration set structure for $\mathcal{J}$ is connected} if any two options in $\mathcal{J}$ are connected. 
\end{definition}

 The connectivity condition ensures that in our model all pairs of candidates $ (\ell, h) $ can be compared probabilistically, either directly (if they share an active  consideration set) or indirectly (through a chain of active consideration sets). Effectively, the connectedness will guarantee that utilities propagate through overlapping sets.

When $J=3$, connectedness of the whole structure is equivalent to requiring that each option $j$ belongs to at least one active non-singleton consideration set. For $J\geq 4$, failures of connectedness can take various forms. For example, if $J=4$ and the only active non-singleton consideration sets are $A_1=\{1,2\}$ and $A_2=\{3,4\}$, then the structure is not connected since options $1$ and $2$ are isolated from options $3$ and $4$ in the decision making process. If, however, there is also an active consideration set $A_3=\{1,3\}$, then connectedness is restored. In that case, for instance, options $2$ and $4$ are connected through the path   $\{1,2\} \to \{1,3\} \to \{3,4\}$. 

The knowledge of $\{\gamma_A\}_{A \in \mathcal{A}}$ would, of course, reveal the  nature of consideration set structure connectivity. But even  when these probabilities are unknown, lack of connectivity has observable implications for choice probabilities, as stated in Proposition~\ref{prop:notconnected}.

 \begin{prop}[Testable implications for the lack of connectivity] \label{prop:notconnected} 
Suppose Assumptions \ref{assn:indCS}--\ref{assn:rationalCS} hold. If the consideration set structure for $\mathcal{J}$ is not connected in the sense of Definition \ref{def:cs_connected}, then there exists a proper subset $B\subset \mathcal{J}$ such that $P(Y \in B|\mathbf{x}) $ is a.e. constant in $\mathbf{x}$ for any model with the given \textit{Behavioral environment}, with value in $(0,1)$.
 \end{prop}
 
The connectivity relation in Definition~\ref{def:cs_connected} partitions $\mathcal J$ into subsets $B_1,\ldots,B_M$ ($M\geq1$), where two elements are connected if and only if they belong to the same $B_m$. If the whole structure is connected, then $M=1$ and $B_1=\mathcal J$.

 \vskip 0.1in 

\noindent \textit{Probability simplex.} The attention structure itself places the following bounds on choice probabilities. Using the partition of $\mathcal J$ into $B_1,\ldots,B_M$, $M\geq1$, discussed after
Proposition~\ref{prop:notconnected}, define
{\small
\[
\widetilde{\Delta}^{B}_{J-1}
=
\bigcap_{m=1}^M
\left\{
\mathbf p\in\Delta_{J-1}:
\sum_{\widetilde A\in\mathcal A:\widetilde A\subseteq S}
\gamma_{\widetilde A}
\leq
\sum_{k\in S}p_k
\leq
\sum_{\widetilde A\in\mathcal A:\widetilde A\cap S\neq\emptyset}
\gamma_{\widetilde A}
\quad
\forall S\subseteq B_m
\right\}.
\]
}The form and affine and topological dimensions of $\widetilde{\Delta}^{B}_{J-1}$  depend on how elements are connected within the consideration set structure. If any two elements in $\mathcal{J}$ are connected in the sense of Definition \ref{def:cs_connected}, then 
$\widetilde{\Delta}^{B}_{J-1}$ has affine and topological dimensions equal to $J-1$.  Every generated probability vector satisfies these inequalities, so
$\Delta^{B,\cup_{\boldsymbol\beta}}_{J-1}
\subseteq
\widetilde{\Delta}^{B}_{J-1}$.

Given the \textit{Behavioral environment}, under Assumption \ref{assn:multi_distributionMED} the total indifference relation \eqref{totalindiff} is mapped into a single point $\boldsymbol{\pi}^*=(\pi_1^*, \ldots, \pi_J^*)$ such that 
\begin{equation}
\label{pistar}
\pi_j^* = \sum_{A: j \in A} \frac{\gamma_A}{|A|}, \quad j=1, \ldots, J.
\end{equation}
This does not depend on the connectivity of elements in the consideration set structure. 

\subsubsection{Partitioning by linear hyperplanes (Behavioral environment level)} 

Let's now ask when the model delivers a ``quantile'' partition of the ambient attention compatible probability set $\relint(\widetilde{\Delta}^{B}_{J-1})$ by linear hyperplanes at the \textit{Behavioral environment} level and uniformly over $\boldsymbol{\beta}$. When looking at the attainable choice probability vectors, this is effectively asking when  $\mathcal{I}^{B, \cup_{\boldsymbol{\beta}}}_{k_1,k_2}$ is in a linear hyperplane in $\relint(\Delta^{B,\cup_{\boldsymbol{\beta}} }_{J-1})$ and $\mathcal{P}^{+,B, \cup_{\boldsymbol{\beta}}}_{k_1,k_2} \cap \mathcal{P}^{+,B, \cup_{\boldsymbol{\beta}}}_{k_2,k_1} \cap \relint(\Delta^{B,\cup_{\boldsymbol{\beta}} }_{J-1}) = \varnothing$ for any $k_1 \neq k_2$.
For interpretability I first focus on $J=3$  and then give the extension to $J\geq 4$. For $J=3$, \eqref{pistar}  becomes $\pi_j^* = \gamma_{\{1,2,3\}}/3+\gamma_{\{j,k\}}/2+\gamma_{\{j,h\}}/2+\gamma_{\{j\}}$, $k \neq j, k\neq h, h\neq j.$ 

For the sufficient direction in Theorem \ref{th:linearJ3}  below, no additional assumptions on the range obtained by the index vector are required. For the converse, suppose that, as $\mathbf{x}$ and $\boldsymbol{\beta}$ vary over their stated domains, configurations satisfying $x_{k_1}\beta_{k_1}=x_{k_2}\beta_{k_2}$ allow the remaining index to vary arbitrarily relative to their common value.

\begin{theorem} \label{th:linearJ3} Suppose Assumptions \ref{assn:multi_distributionMED}--\ref{assn:rationalCS}  hold and $J=3$. Let elements $k_1, k_2$, $k_1 \neq k_2$, in $\mathcal{J}$ be connected in the sense of Definition \ref{def:cs_connected}. Then: 

\textbf{A.} If the rank of
\(
C_{k_1,k_2} = \begin{pmatrix}
\gamma_{\{1,2,3\}} & \gamma_{\{k_1,k_3\}} \\
\gamma_{\{1,2,3\}} & \gamma_{\{k_2,k_3\}}
\end{pmatrix}
\)
is strictly less than $2$, then $\mathcal{I}^{B,\cup_{\boldsymbol{\beta}}}_{k_1,k_2}\cap\relint(\Delta^{B,\cup_{\boldsymbol{\beta}}}_2)$
is contained in a linear hyperplane of dimension at most $J-2=1$. Under the additional index variation condition above,
the converse also holds.

\vskip 0.05in 

In parts B–-C, consider only attained probability vectors in \(\relint(\Delta^{B,\cup_{\boldsymbol \beta}}_2)\). 

\vskip 0.05in 

\textbf{B.} When $\mathrm{rank}(C_{k_1,k_2})=1$, the indifference set  $\mathcal{I}^{B,\cup_{\boldsymbol{\beta}}}_{k_1,k_2}
\cap
\relint\!\left(\Delta^{B,\cup_{\boldsymbol{\beta}}}_2\right)$ coincides with the intersection of
$\relint(\Delta^{B,\cup_{\boldsymbol{\beta}}}_2)$ and the linear hyperplane
$(\gamma_{\{1,2,3\}}+\gamma_{\{k_2,k_3\}})
(p_{k_1}-\pi^*_{k_1})
=
(\gamma_{\{1,2,3\}}+\gamma_{\{k_1,k_3\}})
(p_{k_2}-\pi^*_{k_2}).$ 
Set 
$\mathcal{P}^{+,B, \cup_{\boldsymbol{\beta}}}_{k_1,k_2}\cap \relint(\Delta^{B, \cup_{\boldsymbol{\beta}}}_{2})$ coincides with 
\begin{multline*}  \left\{\mathbf{p} \in \relint(\Delta^{B,\cup_{\boldsymbol{\beta}}}_{2}): (\gamma_{\{1,2,3\}}+\gamma_{\{k_2,k_3\}})(p_{k_1}-\pi^*_{k_1})  > (\gamma_{\{1,2,3\}}+\gamma_{\{k_1,k_3\}}) (p_{k_2}-\pi^*_{k_2})\right\}.
\end{multline*}

\textbf{C.} When $\mathrm{rank}(C_{k_1,k_2})=0$, then 
$k_3\in\mathcal{J}\setminus\{k_1,k_2\}$ is not connected to 
$k_1$ or $k_2$, and $\mathcal{I}^{B,\cup_{\boldsymbol{\beta}}}_{k_1,k_2}
=
\left\{
\mathbf{p}\in\Delta^{B,\cup_{\boldsymbol{\beta}}}_{2}:
p_{k_1}-\pi^*_{k_1}=0,\;
p_{k_2}-\pi^*_{k_2}=0,\;
p_{k_3}-\pi^*_{k_3}=0
\right\},$ and 
{\small$$\mathcal{P}^{+,B,\cup_{\boldsymbol{\beta}}}_{k_1,k_2}
\cap
\relint(\Delta^{B,\cup_{\boldsymbol{\beta}}}_{2})
=
\left\{
\mathbf{p}\in\relint(\Delta^{B,\cup_{\boldsymbol{\beta}}}_{2}):
p_{k_1}-\pi^*_{k_1}>0,\;
p_{k_2}-\pi^*_{k_2}<0,\;
p_{k_3}=\pi^*_{k_3}
\right\}.$$}
\end{theorem}
When the two index differences range over all of $\mathbb{R}^2$, the dimensions of the sets in the theorem can be stated more sharply. In part \textbf{B} (part \textbf{C}), 
$\mathcal{I}^{B,\cup_{\boldsymbol{\beta}}}_{k_1,k_2}
\cap\relint(\Delta^{B,\cup_{\boldsymbol{\beta}}}_{2})$
has dimension $1$ ($0$ in part \textbf{C}), while
$\mathcal{P}^{+,B,\cup_{\boldsymbol{\beta}}}_{k_1,k_2}
\cap\relint(\Delta^{B,\cup_{\boldsymbol{\beta}}}_{2})$
has dimension $2$ ($1$ in part \textbf{C}).

When $\gamma_{\{1,2,3\}}>0$,  condition $\mathrm{rank}(C_{k_1,k_2}) = 1$ is equivalent to
$\gamma_{\{k_1,k_3\}} = \gamma_{\{k_2,k_3\}}$.\footnote{$\gamma_{\{k_1\}}$, $\gamma_{\{k_2\}}$, $\gamma_{\{k_3\}}$ can all be different.} 
One can think of this as attention neutrality of $k_3$ with respect to the pair
$(k_1,k_2)$ and the attention process not ``taking sides'' between the competing
pair. When $\gamma_{\{1,2,3\}}=0$, the two pair menu probabilities may differ, and $C_{k_1,k_2}$
 has rank one provided at least one of them is positive. If both are zero, the matrix has rank zero.

Corollary \ref{cor:linearJ3allpairs} next  shows that the three pairwise linear separators induce six geometric regions in $\widetilde{\Delta}^{B}_2$ and  their attainable traces recover the corresponding strict utility rankings.

\begin{corollary}[``Quantile'' linear partitioning for $J=3$ and ranking recovery]\label{cor:linearJ3allpairs} 
Suppose Assumptions~\ref{assn:multi_distributionMED}-\ref{assn:rationalCS} hold and $J=3$.
Let $\mathcal J$ be connected in the sense of
Definition~\ref{def:cs_connected}, and suppose that
$C_{k_1,k_2}$ in Theorem~\ref{th:linearJ3} has rank one for every
pair $k_1\neq k_2$. For each pair $(k_1,k_2)$, let $k_3$ denote the remaining alternative.  Define the ambient indifference line
{\small$$H_{k_1,k_2}
:=
\left\{
\mathbf p\in\relint(\Delta_2):
\bigl(\gamma_{\{1,2,3\}}+\gamma_{\{k_2,k_3\}}\bigr)
(p_{k_1}-\pi^*_{k_1})
=
\bigl(\gamma_{\{1,2,3\}}+\gamma_{\{k_1,k_3\}}\bigr)
(p_{k_2}-\pi^*_{k_2})
\right\}.$$}The three pairwise lines are distinct, intersect at
$\boldsymbol\pi^*$, and divide $\relint(\widetilde{\Delta}^{B}_2)$
into six nonempty open regions. Among attained probability vectors in \(\operatorname{relint}(\Delta^{B,\cup_{\boldsymbol \beta}}_2)\), the lines identify pairwise ties and the six regions identify the corresponding distinct strict rankings. Some regions may contain no attained vectors.
\end{corollary}

 The corollary  delivers a ``quantile'' multinomial choice model, in the sense defined above, at the \emph{Behavioral environment} level uniformly over $\boldsymbol{\beta} \in \mathcal{B}_0$. The six region  ambient arrangement
depends only on the attention probabilities.  The  uniformity over $\boldsymbol\beta\in\mathcal B_0$ concerns the ranking rule, not the set of probability vectors attained, with different values of $\boldsymbol\beta$ potentially  generating different traces of the six ambient regions. For every pair $(k_1,k_2)$, the position of the corresponding probability vector relative to $H_{k_1,k_2}$ uniquely determines the sign of $x_{k_1}\beta_{k_1}-x_{k_2}\beta_{k_2}$,  automatically giving ranking recovery.

\vskip 0.05in 

\noindent \textbf{Graphical illustrations.} Figure \ref{fig:graphical_ill_simplex} illustrates two consideration structures satisfying Corollary \ref{cor:linearJ3allpairs}, showing how the shape of the attention compatible probability set $\widetilde{\Delta}^{B}_2$, $\boldsymbol\pi^*$, and the slopes of the indifference lines depend on $\{\gamma_A\}$. In both cases, any two elements are connected in the sense of Definition \ref{def:cs_connected}. The three ambient indifference lines intersect at $\boldsymbol\pi^*$ and divide the simplex into six cells, one for each complete strict ordering $j\succ h\succ\ell$. Intersecting these cells with $\widetilde{\Delta}^{B}_2$ gives the truncated regions shown in the figure.

\begin{figure}[htbp]
\centering
\begin{minipage}{0.45\textwidth}
\centering
\includegraphics[width=\linewidth]{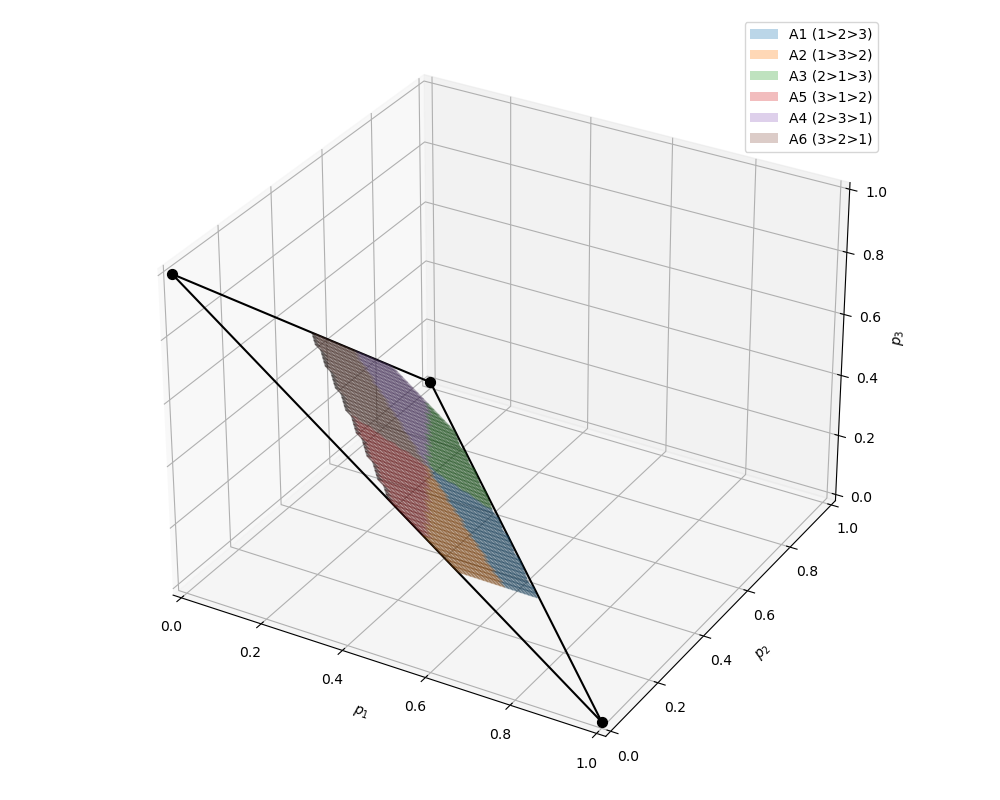}
\end{minipage}
\hfill
\begin{minipage}{0.54\textwidth}
\centering
\includegraphics[width=\linewidth]{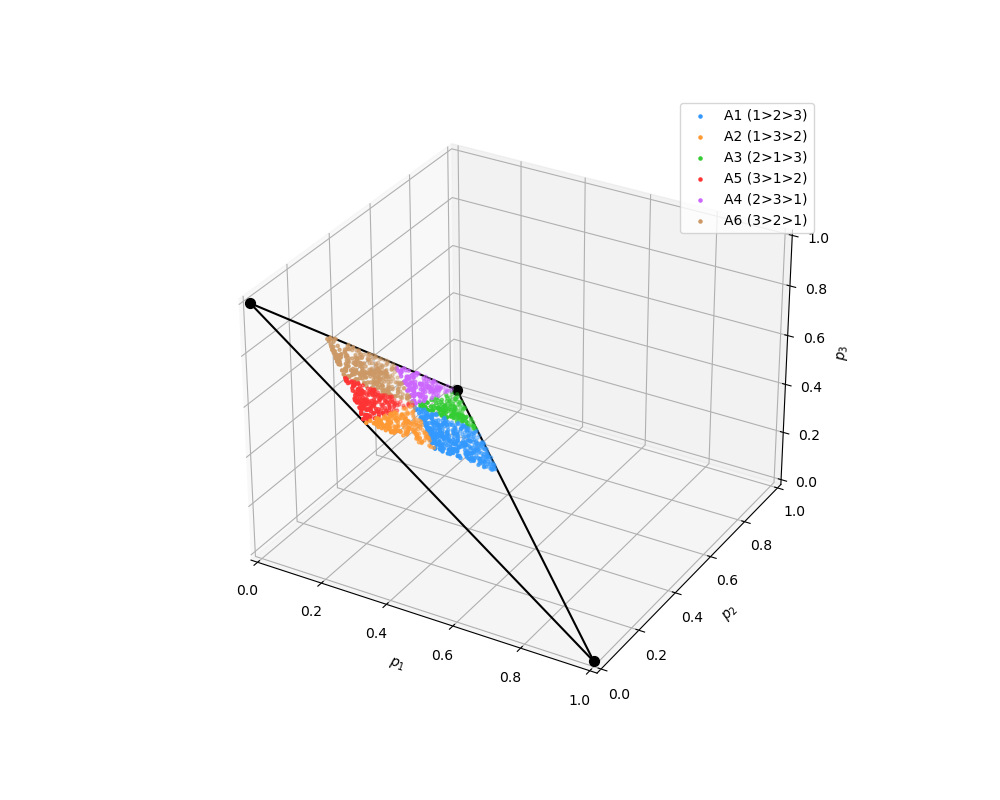}
\end{minipage}
\caption{LEFT: case $\gamma_{12}=1/2$, $\gamma_{23}=1/3$, $\gamma_{13}=1/6$. Here $\widetilde{\Delta}^{B}_2$ is bounded by $p_1\le \gamma_{12}+\gamma_{13}$, $p_2\le \gamma_{12}+\gamma_{23}$, $p_3\le \gamma_{13}+\gamma_{23}$. RIGHT: case $\gamma_{23}=2/3$, $\gamma_{12}=1/3$.  $\widetilde{\Delta}^{B}_2$ is bounded by $p_1\le \gamma_{12}$, $p_3\le \gamma_{23}$.}
\label{fig:graphical_ill_simplex}
\end{figure}

\vskip 0.05in 

\noindent \textbf{Case $J \geq 4$.} For $J \geq 4$, let's first extend Theorem \ref{th:linearJ3} and Corollary \ref{cor:linearJ3allpairs}. Such an extension will give common linear separators in the $(p_{k_1},p_{k_2})$-coordinates for a pair $(k_1,k_2)$ that satisfies a certain  two-coordinate rank-one condition. When taken together across pairs, these conditions generate common linear
pairwise ranking boundaries. Whether they give exact ranking partitioning
depends on the strictness conditions characterized below.

\begin{theorem}[Two-coordinate linear ranking revelation] 
\label{th:linear_generalJ}
Suppose Assumptions \ref{assn:multi_distributionMED}-\ref{assn:rationalCS}
hold and $J\geq4$.  Let $k_1$ and $k_2$ be connected in the sense of
Definition~\ref{def:cs_connected}. Denote  $O=\mathcal J\setminus\{k_1,k_2\}$.

\textbf{A.} For 
$
C^{(J)}_{k_1,k_2}
\equiv
\left(
\begin{array}{cc}
\bigl(\gamma_{S\cup\{k_1,k_2\}}\bigr)_{\varnothing\neq S\subseteq
\mathcal J\setminus\{k_1,k_2\}}
&
\bigl(\gamma_{S\cup\{k_1\}}\bigr)_{\varnothing\neq S\subseteq
\mathcal J\setminus\{k_1,k_2\}}
\\[1mm] 
\bigl(\gamma_{S\cup\{k_1,k_2\}}\bigr)_{\varnothing\neq S\subseteq
\mathcal J\setminus\{k_1,k_2\}}
&
\bigl(\gamma_{S\cup\{k_2\}}\bigr)_{\varnothing\neq S\subseteq
\mathcal J\setminus\{k_1,k_2\}},
\end{array}
\right). 
\quad \qquad \qquad 
$  if 
$\mathrm{rank}(C^{(J)}_{k_1,k_2})<2$, 
then the image of pairwise indifference between \(k_1\) and \(k_2\) is
in a linear hyperplane involving only \(p_{k_1}\) and \(p_{k_2}\). This rank condition is equivalent to
the following: 
for any nonempty 
$S,T\subseteq O$,
\begin{equation}
\gamma_{S\cup\{k_1\}}\gamma_{T\cup\{k_2\}}
=
\gamma_{S\cup\{k_2\}}\gamma_{T\cup\{k_1\}},
\qquad
\gamma_{S\cup\{k_1,k_2\}}
\bigl(\gamma_{T\cup\{k_1\}}-\gamma_{T\cup\{k_2\}}\bigr)=0.
\label{eq:primitive-rank-generalJ}
\end{equation} 
Conversely, if the $J-1$ deterministic index differences range over all of
$\mathbb R^{J-1}$, then such a two-coordinate linear hyperplane exists
only if $
\operatorname{rank}\!\left(C^{(J)}_{k_1,k_2}\right)<2.
$

\noindent
\textbf{B.}
Suppose $\operatorname{rank}(C^{(J)}_{k_1,k_2})=1$. Define
{\small\begin{equation} 
\label{eq:Hrank1J4}
H_{k_1,k_2}(p)
:=
(p_{k_1}-\pi^*_{k_1})
\sum_{\emptyset\neq S\subseteq O}
\left(
\gamma_{S\cup\{k_1,k_2\}}
+
\gamma_{S\cup\{k_2\}}
\right)
-
(p_{k_2}-\pi^*_{k_2})
\sum_{\emptyset\neq S\subseteq O}
\left(
\gamma_{S\cup\{k_1,k_2\}}
+
\gamma_{S\cup\{k_1\}}
\right).
\end{equation}}

\smallskip

\noindent
\emph{(Bi)} If $\gamma_{\{k_1,k_2\}}>0$, then 
for every generated $\mathbf p(\mathbf x)$, $
H_{k_1,k_2}(\mathbf p(\mathbf x))=(>)0  \iff x_{k_1}\beta_{k_1}= (>) x_{k_2}\beta_{k_2}.$

\emph{(Bii)} (One-sided linear ranking revelation) If $\gamma_{\{k_1,k_2\}}=0$, then for every generated $\mathbf p(\mathbf x)$, 
\[
x_{k_1}\beta_{k_1}=x_{k_2}\beta_{k_2}
\Rightarrow
H_{k_1,k_2}(\mathbf p(\mathbf x))=0, \qquad 
x_{k_1}\beta_{k_1}> (<) x_{k_2}\beta_{k_2}
\Rightarrow
H_{k_1,k_2}(\mathbf p(\mathbf x))\geq (\leq) 0, 
\] 
\[
H_{k_1,k_2}(\mathbf p(\mathbf x))>0
\; \Rightarrow\; 
x_{k_1}\beta_{k_1}>x_{k_2}\beta_{k_2},
\qquad 
H_{k_1,k_2}(\mathbf p(\mathbf x))<0
\;\Rightarrow\;
x_{k_1}\beta_{k_1}<x_{k_2}\beta_{k_2}.
\]
Thus, either open half-space reveals the corresponding strict pairwise
ranking, while no ranking conclusion follows from
\(H_{k_1,k_2}(\mathbf p(\mathbf x))=0\) without additional support
restrictions.

\smallskip

\noindent
\emph{(Biii) (Full support case).}
Suppose, in addition, that the admissible distribution class in the
\textit{Behavioral environment} is restricted to conditional distributions
of $\boldsymbol{\varepsilon}|\mathbf x$ with support $\mathbb R^J$.
Then, regardless of whether $\gamma_{\{k_1,k_2\}}$ is zero or strictly
positive, for every generated $\mathbf p(\mathbf x)$,
$$H_{k_1,k_2}(\mathbf p(\mathbf x))= (>)0 
\;\iff\;
x_{k_1}\beta_{k_1}= (>) x_{k_2}\beta_{k_2}.
$$
Thus, under full support, the hyperplane in \eqref{eq:Hrank1J4} and its two open
half-spaces exactly characterize pairwise indifference and the two strict pairwise rankings.

\noindent
\textbf{C.}
If $\operatorname{rank}(C^{(J)}_{k_1,k_2})=0$, then no active menu contains either \(k_1\) or \(k_2\) together with an
alternative outside \(\{k_1,k_2\}\). Since \(k_1\) and \(k_2\) are connected, $\gamma_{\{k_1,k_2\}}>0$. For every generated $\mathbf p(\mathbf x)\in
\relint(\Delta^{B,\cup_{\boldsymbol{\beta}}}_{J-1})$, it holds that $x_{k_1}\beta_{k_1}=x_{k_2}\beta_{k_2} \iff
p_{k_1}-\pi^*_{k_1}
=
p_{k_2}-\pi^*_{k_2}
=0$ and 
$x_{k_1}\beta_{k_1}>x_{k_2}\beta_{k_2}
\iff 
p_{k_1}-\pi^*_{k_1}>0
\iff 
p_{k_2}-\pi^*_{k_2}<0$.
\end{theorem}

The interpretation of \eqref{eq:primitive-rank-generalJ} is similar to the case $J=3$.  The first equality
in \eqref{eq:primitive-rank-generalJ} requires the probabilities of menus $S \cup \{k_1\}$ and $S \cup \{k_2\}$, for $S \neq \varnothing$, $S \cap \{k_1,k_2\} = \varnothing$  to be proportional across outside
sets $S$.  The second requires them to be equal whenever an active menu contains
both alternatives and at least one outside alternative. 

To illustrate that in Theorem~\ref{th:linear_generalJ}{(Bii)}, $H_{1,2}(\mathbf{p}(\mathbf x))=0$ can occur at $\mathbf p (\mathbf x)\in
\relint\!(\widetilde{\Delta}^{B}_{5})$ even when the
underlying pairwise index ranking is strict, take $J=6$ and consider the pair $(1,2)$. Suppose the active consideration sets are $\{1,3\},\ \{2,3\},\ \{1,4\},\ \{2,4\},\ \{3,5\},\ \{5,6\},\ \{6,4\}$, each with probability $1/7$. Then $\gamma_{\{1,2\}}=0$ and $\operatorname{rank} \! (C^{(6)}_{1,2})=1$. Let the unobservables be i.i.d.\ $\operatorname{Unif}[-1,1]$ and take
$x_1\beta_1=\frac14,  
x_2 \beta_2=0, 
x_3 \beta_3=\frac52, 
x_4 \beta_4=-\frac94, 
x_5 \beta_5=1, 
x_6 \beta_6=-\frac12$. Thus $x_1\beta_1>x_2\beta_2$. Since every unobservables  difference lies in $[-2,2]$, $P(Y=1|\mathbf x, \{1,3\})=P(Y=2|\mathbf x, \{2,3\})=0$, and $P(Y=1|\mathbf x, \{1,4\})=P(Y=2|\mathbf x, \{2,4\})=1$. Hence, $P(Y=1|\mathbf x)=P(Y=2|\mathbf x)=\frac17=\pi_1^*=\pi_2^*$ and  $H_{1,2}(\mathbf p(\mathbf x))=0$ despite the strict ranking $x_1\beta_1>x_2\beta_2$. A direct calculation gives
$\mathbf p(\mathbf x)
=
\left(
\frac17,\,
\frac17,\,
\frac{95}{224},\,
\frac{1}{896},\,
\frac17,\,
\frac{131}{896}
\right)$. Moreover, all nontrivial inequalities defining
$\widetilde{\Delta}^{B}_{5}$ are strict at this probability
vector, so  $\mathbf p (\mathbf x)\in
\relint\!(\widetilde{\Delta}^{B}_{5})$.

Corollary~\ref{cor:linear-allpairs-generalJ} is the $J\geq4$
counterpart of Corollary~\ref{cor:linearJ3allpairs}. 

\begin{corollary}
\label{cor:linear-allpairs-generalJ}
Suppose Assumptions~\ref{assn:multi_distributionMED}-
\ref{assn:rationalCS} hold, $J\geq4$, and the whole consideration set
structure is connected. Suppose $\operatorname{rank}(C^{(J)}_{k,\ell})=1$ for every distinct pair $k,\ell$. For each pair $(k,\ell)$, let
$H_{k,\ell}$ be defined as in \eqref{eq:Hrank1J4}.

The hyperplanes $H_{k,\ell}(\mathbf p)=0$, $k\neq\ell$,  all pass through $\boldsymbol\pi^*$ and divide $\relint(\widetilde{\Delta}^{B}_{J-1})$ into exactly $J!$ nonempty relatively open convex regions, one for each
permutation $\kappa=(k_1,\ldots,k_J)$ of $\mathcal J$. The region 
associated with $\kappa$ is $$\mathcal C_\kappa
=
\left\{
\mathbf p\in
\relint\!\left(\widetilde{\Delta}^{B}_{J-1}\right):
H_{k_r,k_s}(\mathbf p)>0
\text{ for all }1\leq r<s\leq J
\right\}.$$

Every generated probability vector $\mathbf p(\mathbf x)\in
\mathcal C_\kappa$ satisfies
$x_{k_1}\beta_{k_1}
>
x_{k_2}\beta_{k_2}
>
\cdots
>
x_{k_J}\beta_{k_J}$.

Conversely, for every generated $\mathbf p(\mathbf x)$, if the deterministic indices have this strict ordering, then $H_{k_r,k_s}(\mathbf p(\mathbf x))\geq0$ for all $r<s$, so $\mathbf p(\mathbf x)$ lies in the closure of $\mathcal C_\kappa$.

If, in addition, $\gamma_{\{k,\ell\}}>0$ for every distinct pair $k,\ell$, then, for every generated probability vector $\mathbf p(\mathbf x)\in \relint(\widetilde{\Delta}^{B}_{J-1})$, we have $\mathbf p(\mathbf x)\in\mathcal C_\kappa$ $\iff$ $x_{k_1}\beta_{k_1}> x_{k_2}\beta_{k_2}>\cdots> x_{k_J}\beta_{k_J}$. The same conclusion holds, regardless of the values of $\gamma_{\{k,\ell\}}$, when the admissible distribution class is restricted as in Theorem~\ref{th:linear_generalJ} (Biii). Under either condition, $H_{k,\ell}(\mathbf p(\mathbf x))=0$  $\iff$ $x_k\beta_k=x_\ell\beta_\ell$.

\end{corollary}

Corollary~\ref{cor:linear-allpairs-generalJ} gives the geometry of the
linear ranking arrangement and conditions for exact ranking recovery. Rank one
conditions for all pairs generate a common $J!$-region linear arrangement
in $\relint(\widetilde{\Delta}^B_{J-1})$, determined only by the attention
probabilities. Just like in $J=3$ case, generated probability vectors trace subsets of these
regions. Every $\mathbf{p}(\mathbf x)$ in an open region  reveals a complete strict
ranking for the underlying index vector. A strict ranking may nevertheless generate a probability vector on the boundary of the corresponding region. This can occur even when the rank-one conditions hold for every pair, as  Example~\ref{ex:star} in Appendix A (right after the proof or Corollary \ref{cor:linear-allpairs-generalJ}) illustrates this possibility.

The rank condition for pair $(k_1,k_2)$ in Theorem~\ref{th:linear_generalJ} concerns hyperplanes
that can be written using only $p_{k_1}$ and $p_{k_2}$.  Allowing other
choice probability coordinates gives a strictly larger class of separators.
The next result provides a direct primitive condition for such a separator. We naturally seek a separator in which an increase in \(p_{k_1}-\pi^*_{k_1}\) favors \(k_1\), while an increase in \(p_{k_2}-\pi^*_{k_2}\) favors \(k_2\). We therefore consider positive and negative coefficients on these terms, respectively, and normalize them to \(1\) and \(-1\).\footnote{For coefficients of opposite signs, such normalization is without loss of generality. Suppose the two coefficients are \(a>0\) and \(b<0\). Subtract \((a+b)/2\) from every coefficient. This leaves the expression unchanged because \(\sum_k(p_k-\pi_k^*)=0\), and makes the two coefficients \((a-b)/2\) and \(-(a-b)/2\). Dividing all coefficients by the positive number \((a-b)/2\) then gives \(1\) and \(-1\), without changing the sign of the expression.}

\begin{theorem}
\label{th:linear_other_coordinates}
Suppose Assumptions \ref{assn:multi_distributionMED}--\ref{assn:rationalCS}
hold and $J\geq4$.  Fix two distinct alternatives $k_1,k_2$.  Suppose numbers
$c_k$, $k\in\mathcal J\setminus\{k_1,k_2\}$, can be chosen so that:
\begin{itemize}
\item[(i)] If $\gamma_A>0$ for some
$A\subseteq\mathcal J\setminus\{k_1,k_2\}$, then $c_h=c_\ell$ for every
$h,\ell\in A$.

\item[(ii)] For every nonempty
$S\subseteq\mathcal J\setminus\{k_1,k_2\}$ such that
$\gamma_{S\cup\{k_1\}}+\gamma_{S\cup\{k_2\}}>0$,
\begin{equation}
 c_k=
 \frac{\gamma_{S\cup\{k_1\}}-\gamma_{S\cup\{k_2\}}}
 {\gamma_{S\cup\{k_1\}}+\gamma_{S\cup\{k_2\}}}
 \quad\text{for every }k\in S.
\label{eq:other-coordinate-exclusive}
\end{equation}

\item[(iii)] If
$\gamma_{S\cup\{k_1,k_2\}}>0$ for some nonempty
$S\subseteq\mathcal J\setminus\{k_1,k_2\}$, then $c_k=0$ for every $k\in S$.
\end{itemize}
If $k_1$ and $k_2$ are connected, then  for every generated
$\mathbf p(\mathbf x)$,   $x_{k_1}\beta_{k_1}=x_{k_2}\beta_{k_2}
\; \Rightarrow\; 
L(\mathbf p(\mathbf x))=0$, where \begin{equation}
L(\mathbf p)
:=
(p_{k_1}-\pi_{k_1}^*)-(p_{k_2}-\pi_{k_2}^*)
+\sum_{k\in\mathcal J\setminus\{k_1,k_2\}}
c_k(p_k-\pi_k^*).
\label{eq:other-coordinate-hyperplane}
\end{equation}
Moreover, $x_{k_1}\beta_{k_1}>(<)x_{k_2}\beta_{k_2}\; \Rightarrow\; L(\mathbf p(\mathbf x))\geq(\leq)0$. Hence,  a strict sign of $L$ reveals the corresponding strict ranking.

If $\gamma_{\{k_1,k_2\}}>0$, then 
$x_{k_1}\beta_{k_1}>(<) (=)x_{k_2}\beta_{k_2}
\; \iff \;
L(\mathbf p(\mathbf x)) >(<) (=)0$, and the same  characterization holds, regardless of
$\gamma_{\{k_1,k_2\}}$, when the admissible distribution class is
restricted to conditional distributions with support $\mathbb R^J$.

The hyperplane in \eqref{eq:other-coordinate-hyperplane} has dimension
$J-2$ in the simplex.
\end{theorem}

The pairwise result immediately yields complete ranking revelation on the attainable traces of the corresponding strict half-spaces.

\begin{corollary}
\label{cor:linear-allpairs-other-coordinates}
Suppose the whole consideration set structure is connected and the
assumptions and conditions (i)--(iii) of
Theorem~\ref{th:linear_other_coordinates} hold for every distinct pair
$k,\ell$. Let $L_{k,\ell}$ denote the corresponding linear functional.

For any permutation $(k_1,\ldots,k_J)$ of $\mathcal J$,  $L_{k_r,k_{r+1}}(\mathbf p(\mathbf x))>0$,
$r=1,\ldots,J-1$, implies $x_{k_1}\beta_{k_1}>
x_{k_2}\beta_{k_2}>
\cdots>
x_{k_J}\beta_{k_J}$.

Conversely, this strict ordering implies $L_{k_r,k_{r+1}}(\mathbf p(\mathbf x))\geq0$,
$ r=1,\ldots,J-1$. Hence the attainable strict half-space regions associated with distinct
permutations are disjoint. If $\gamma_{\{k,\ell\}}>0$ for every distinct pair $k,\ell$, the converse
holds with strict inequalities. The same conclusion holds, regardless
of the values of $\gamma_{\{k,\ell\}}$, under the full support
restriction in Theorem~\ref{th:linear_other_coordinates}.
\end{corollary}

Unlike Corollary~\ref{cor:linear-allpairs-generalJ}, this result allows the separator for a given pair to involve all probability coordinates. Thus, pairwise linear ranking revelation may hold even when one or more of the two-coordinate matrices $C^{(J)}_{k,\ell}$ have rank two. Theorem~\ref{th:linear_other_coordinates} provides pair-specific conditions and does not by itself determine the global geometry generated by all pairwise separators. In particular, the number and shape of the resulting regions must be evaluated for the particular consideration structure. Therefore, the paper does not give general $J!$-region partitioning result for this case. 

The sufficiency part of Theorem \ref{th:linear_generalJ} is the special case of Theorem \ref{th:linear_other_coordinates} in which the coefficients
$c_k$ can be chosen equal for all $k\notin\{k_1,k_2\}$.  Indeed, if their
common value is $c$, then \eqref{eq:other-coordinate-hyperplane} is equivalent,
using $\sum_k(p_k-\pi_k^*)=0$, to $(1-c)(p_{k_1}-\pi_{k_1}^*)
 -(1+c)(p_{k_2}-\pi_{k_2}^*)=0$. Condition \eqref{eq:other-coordinate-exclusive} then requires the same attention ratio across all outside sets, while any active menu containing both
$k_1$ and $k_2$  and at least one outside alternative forces $c=0$.  These are exactly the rank-one restrictions in Theorem~\ref{th:linear_generalJ}.

Example \ref{ex:other-coordinate-separator-J4} below shows that for $J \geq 4$,  a partitioning hyperplane involving other probability coordinates may exist even when conditions of Theorem \ref{th:linear_generalJ} fail. 

\begin{example}
\label{ex:other-coordinate-separator-J4}
Let $\mathcal J=\{1,2,3,4\}$ and suppose the only active non-singleton menus
are $\{1,3\}$, $\{2,3\}$, $\{1,4\}$, and $\{2,4\}$, with all four attention
probabilities positive.  Set $c_3=\frac{\gamma_{\{1,3\}}-\gamma_{\{2,3\}}}
 {\gamma_{\{1,3\}}+\gamma_{\{2,3\}}}$, $c_4=\frac{\gamma_{\{1,4\}}-\gamma_{\{2,4\}}}
 {\gamma_{\{1,4\}}+\gamma_{\{2,4\}}}$. Then Theorem~\ref{th:linear_other_coordinates} implies that $(p_1-\pi_1^*)-(p_2-\pi_2^*)+c_3(p_3-\pi_3^*)+c_4(p_4-\pi_4^*)$ has the weak sign of $x_1\beta_1-x_2\beta_2$, and a strict sign reveals the corresponding strict ranking. Under the full support restriction, the signs coincide exactly. Under the index variation condition for the converse in Theorem~\ref{th:linear_generalJ}, a separator robust to distributions and using only \(p_1\) and \(p_2\) requires \(c_3=c_4\). Hence, whenever $c_3\neq c_4$, the two-coordinate matrix has rank two while a four-coordinate linear separator still exists.\end{example}

\noindent \textbf{Lack of connectivity.}  When the consideration structure is disconnected, choice probabilities within each connected component have a fixed sum and depend only on within component index differences. The results from the connected case therefore apply separately within each component, subject to their stated assumptions. Rankings across different connected components generally require additional restrictions. In particular, an isolated alternative has a constant choice probability; if all alternatives are isolated, choice probabilities cannot distinguish between attainable rankings at all.

\subsubsection{When partitioning of $\widetilde{\Delta}^B_{J-1}$ by linear hyperplanes fails (Behavioral environment  level)} 
\label{sec:behavioralmultistructuralmodel}

When $\mathrm{rank}(C_{k_1,k_2})=2$ for $J=3$, a direct calculation
shows what goes wrong.  On the partial indifference locus \eqref{partialindiff}, 
{\small$$P(Y=k_1\mid\mathbf{x})-\pi_{k_1}^*
=
P(Y=k_2\mid\mathbf{x})-\pi_{k_2}^*
+
\bigl(\gamma_{\{k_1,k_3\}}-\gamma_{\{k_2,k_3\}}\bigr)
\bigl(G_{1,\mathbf{x}}(-x_{k_2}\beta_{k_2}+x_{k_3}\beta_{k_3})-1/2\bigr),$$}with $G_{1,\mathbf{x}}$ denoting the survival function of the difference between two errors conditional on $\mathbf{x}$ (Assumption~\ref{assn:multi_distributionMED} ensures that this is the same for all index pairs). Since
$\gamma_{\{k_1,k_3\}}-\gamma_{\{k_2,k_3\}}\neq0$, as implied by
$\mathrm{rank}(C_{k_1,k_2})=2$, the additional term
$\bigl(\gamma_{\{k_1,k_3\}}-\gamma_{\{k_2,k_3\}}\bigr)
\bigl(G_{1,\mathbf{x}}(-x_{k_2}\beta_{k_2}+x_{k_3}\beta_{k_3})-1/2\bigr)$ need not vanish and generally depends on the conditional distribution
$\boldsymbol{\varepsilon}|\mathbf{x}$ and on $\mathbf{x}$ in a nonlinear way. Varying $\mathbf{x}$ while preserving  partial indifference
\eqref{partialindiff}, together with variation in
$\PP_{\boldsymbol{\varepsilon}\mid\mathbf{x},\,\mathbf{x}\in\mathbb{R}^{M_e}}\in\mathcal{P}$, will generally result in 
$\mathcal{I}^{B,\cup_{\boldsymbol{\beta}}}_{k_1,k_2}$ that retains the full affine and topological dimension of 
$\Delta^{B,\cup_{\boldsymbol{\beta}}}_{2}$. Moreover, as illustrated in more detail below, the overlap of
$\mathcal{P}^{+,B,\cup_{\boldsymbol{\beta}}}_{k_1,k_2}\cap
\relint(\Delta^{B,\cup_{\boldsymbol{\beta}}}_{2})$
and
$\mathcal{P}^{+,B,\cup_{\boldsymbol{\beta}}}_{k_2,k_1}\cap
\relint(\Delta^{B,\cup_{\boldsymbol{\beta}}}_{2})$
may likewise have the full affine and topological dimension of
$\Delta^{B,\cup_{\boldsymbol{\beta}}}_{2}$.

Thus, once the rank conditions in Corollary \ref{cor:linearJ3allpairs} or, for a general $J$, conditions of  Corollary \ref{cor:linear-allpairs-other-coordinates} fail, the ``quantile'' multinomial choice model described above is generally no longer guaranteed at the \emph{Behavioral environment} level uniformly over
$\boldsymbol{\beta}\in\mathcal{B}_0$. In particular, ranking recovery at  that level may fail even though choice probabilities can continue to contain ranking information on more restricted sets.\footnote{On a restricted attainable set, a violation of the rank condition need not be exposed by the available index configurations.} Such a failure, however, does not rule out ranking recovery under a
more specific structural model. One can move to the
\emph{Structural environment} level and fix
$\mathbb P_{\boldsymbol\varepsilon\mid\mathbf x}$, thus fixing
$G_{1,\mathbf x}$ for each $\mathbf x$. This does not by itself produce
a common map from index differences into choice probabilities. But the
relevant question is whether the resulting relationship specific to this model 
still separates the rankings.

\subsubsection{A more modest perspective: the  \textit{Structural environment} level?} 
\label{sec:cs_structural}

For the fixed-$\boldsymbol{\beta}$ structural model level ranking recovery is still not automatic.
The fundamental difficulty is that Assumption~\ref{assn:multi_distributionMED}
imposes no restriction on how the joint distribution of $\boldsymbol{\varepsilon}|\mathbf{x}$
changes with $\mathbf{x}$.
If those changes are discontinuous in $\mathbf{x}$, the indifference set
$\mathcal{I}^{M,\boldsymbol{\beta}}_{k_1,k_2}$ need not be in a $(J-2)$-dimensional manifold, and even if it is, it may have 
self-intersections or other singularities that prevent separation of the ordering regions. In other words, the conditional distribution of unobservables may vary with covariates in ways that prevent different index rankings from generating distinct choice probabilities. To isolate the role of this distributional variation, I first consider the case in which $\boldsymbol{\varepsilon}$ is independent of $\mathbf{x}$, before allowing dependence.

\vskip 0.05in 
\noindent \textit{Stochastic independence of $\boldsymbol{\varepsilon}$ and $\mathbf{x}$.} 
Consider $J=3$ with $\boldsymbol{\varepsilon}$ independent of $\mathbf{x}$. This case produces a common indifference curve (more generally, a common $(J-2)$-dimensional manifold)\footnote{The model, of course,  may attain only a part of it.}  with desirable separability properties and then highlights how this separability can fail once independence is relaxed.

Suppose $\gamma_{123},\gamma_{12},\gamma_{13},\gamma_{23}>0$ with
$\gamma_{123}+\gamma_{12}+\gamma_{13}+\gamma_{23}=1$, and suppose
$\gamma_{12},\gamma_{13},\gamma_{23}$ are distinct.  Let $C_2$ denote the
bivariate copula of $(\varepsilon_{k_2}-\varepsilon_{k_1}, \varepsilon_{k_3}-\varepsilon_{k_1})$, and let $F$ denote the common marginal c.d.f.\ of
$\varepsilon_{k_2}-\varepsilon_{k_1}$ for every permutation
$(k_1,k_2,k_3)$ of $\mathcal{J}$.  In line with Assumption~\ref{assn:multi_distributionMED}, $C_2$ and $F$
are absolutely continuous.  For $k\neq j$, define
$a_{k,j}(\mathbf{x})=x_k\beta_k-x_j\beta_j$.  Then, for example,
{\small$$
P(Y=k_1|\mathbf{x})
 = \gamma_{123}
   C_2\!\left(
       F(a_{k_1,k_2}(\mathbf{x})),
       F(a_{k_1,k_3}(\mathbf{x}))
   \right)
 + \gamma_{k_1k_2}F(a_{k_1,k_2}(\mathbf{x}))
 + \gamma_{k_1k_3}F(a_{k_1,k_3}(\mathbf{x})),
$$}with analogous expressions for the other alternatives.

On $\mathcal{I}^{M,\cup_{\boldsymbol{\beta}}}_{k_1,k_2}$,
$a_{k_1,k_2}(\mathbf{x})=0$ and
$a_{k_1,k_3}(\mathbf{x})=a_{k_2,k_3}(\mathbf{x})$.  Write this common
difference as $a$ and set $u=F(a)$, $h(u)=C_2(1/2,u)$. The partial indifference locus can then be equivalently parametrized by $a$ or  $u\in(0,1)$:
\begin{equation}
\label{k1choiceoncurve}
P(Y=j|\mathbf{x})
=\gamma_{123}h(u)
  +\frac{\gamma_{k_1k_2}}{2}
  +\gamma_{j k_3}u, \quad j \in \{k_1,k_2\},
\end{equation}
and $P(Y=k_3|\mathbf{x})=1-P(Y=k_1|\mathbf{x})-P(Y=k_2|\mathbf{x})$.  Since $\gamma_{k_1k_3},\gamma_{k_2k_3}>0$, both coordinates
are strictly increasing in $u$.  Thus, under stochastic independence of
$\boldsymbol{\varepsilon}$ and $\mathbf{x}$, the image of pairwise index
indifference is a well defined one-dimensional curve in
$\relint(\Delta^{M,\cup_{ \boldsymbol{\beta}}}_2)$.  Its shape, however, depends on the difference copula through $h$ and is generally nonlinear but it does not depend on a particular $\boldsymbol{\beta} \in \mathcal{B}_0$ as it is fully indexed by $u$.

The left panel of Figure \ref{fig:ind_assumption} illustrates this nonlinearity
for i.i.d. standard normal unobservables.  It depicts the slope ratio $\frac{p_1'(a)}{p_2'(a)}$, where $p_k(a)$ is a special case of \eqref{k1choiceoncurve}: $p_k(a)
 =\gamma_{123}
   \Phi_2\!\left(0,\frac{a}{\sqrt{2}};\frac12\right)
 +\frac{\gamma_{12}}{2}
 +\gamma_{k3}\Phi\!\left(\frac{a}{\sqrt{2}}\right)$, $k=1,2$. For attention probabilities used in the figure,
the slope ratio varies smoothly from
$\approx1.20$
to $\approx 12$.  Hence the partial indifference locus cannot be
a straight line. 

The right panel illustrates that $\mathcal{I}_{k_1,k_2}^{M,\cup_{ \boldsymbol{\beta}}}$ also depends on the 
distribution of the unobservables, even when every candidate
distribution is fully exchangeable, absolutely continuous, and
independent of $\mathbf{x}$.  Three specifications are compared.  In the
first, $\boldsymbol{\varepsilon}
 =\boldsymbol{V}+0.25\boldsymbol{Z}$,
where $\boldsymbol{V}$ is uniformly distributed over the six coordinate
permutations of $(0,1,2)^{\top}$, independently of
$\boldsymbol{Z}\sim \mathcal{N}(\boldsymbol{0},I_3)$ (smooth permutation-Gaussian mixture).  The second specification is
the i.i.d. normal benchmark
$\boldsymbol{\varepsilon}=\boldsymbol{Z}$.  In the third, $\boldsymbol{\varepsilon}=S\boldsymbol{Z}$, where 
$P(S=1)=P(S=10^4)=\frac12$,
with $S$ independent of $\boldsymbol{Z}$ (common scale normal).    The display focuses on $u\in[0.08,0.30]$, where the
differences between curves are particularly  visible.\footnote{This comparison is not intended to
hold the marginal difference distribution fixed and each specification is
parameterized by its own percentile $u=F(a)$.}  Thus, at the \textit{Structural environment} level, $\mathcal{I}_{k_1,k_2}^{M,\cup_{\boldsymbol{\beta}}}$ is not determined by the attention probabilities alone but it does meaningfully separate $\relint(\Delta^{M,\cup_{\boldsymbol{\beta}}}_2)$
into disjoint ordering regions: for $x_{k_1}\beta_{k_1}>x_{k_2}\beta_{k_2}$ it holds that
$P(Y=k_1|\mathbf{x})>p_1(a)$ and
$P(Y=k_2|\mathbf{x})<p_2(a)$, placing the choice probability vector
strictly below the indifference curve, while the inequalities reverse
when $x_{k_1}\beta_{k_1}<x_{k_2}\beta_{k_2}$, placing it above the
curve.

\begin{figure}[htbp]
    \centering
    \begin{minipage}{0.48\textwidth}
        \centering
        \includegraphics[width=0.95\linewidth]
{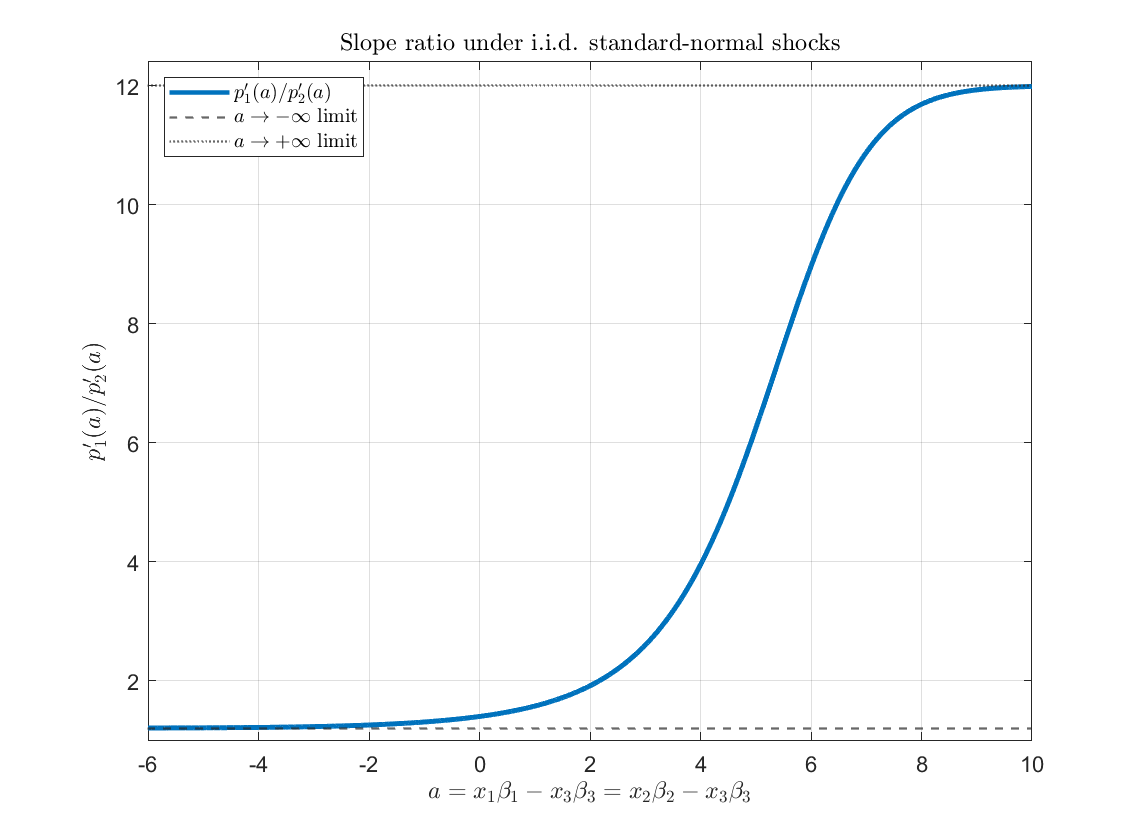}
    \end{minipage}
    \hfill
    \begin{minipage}{0.48\textwidth}
        \centering
        \includegraphics[width=0.95\linewidth]    {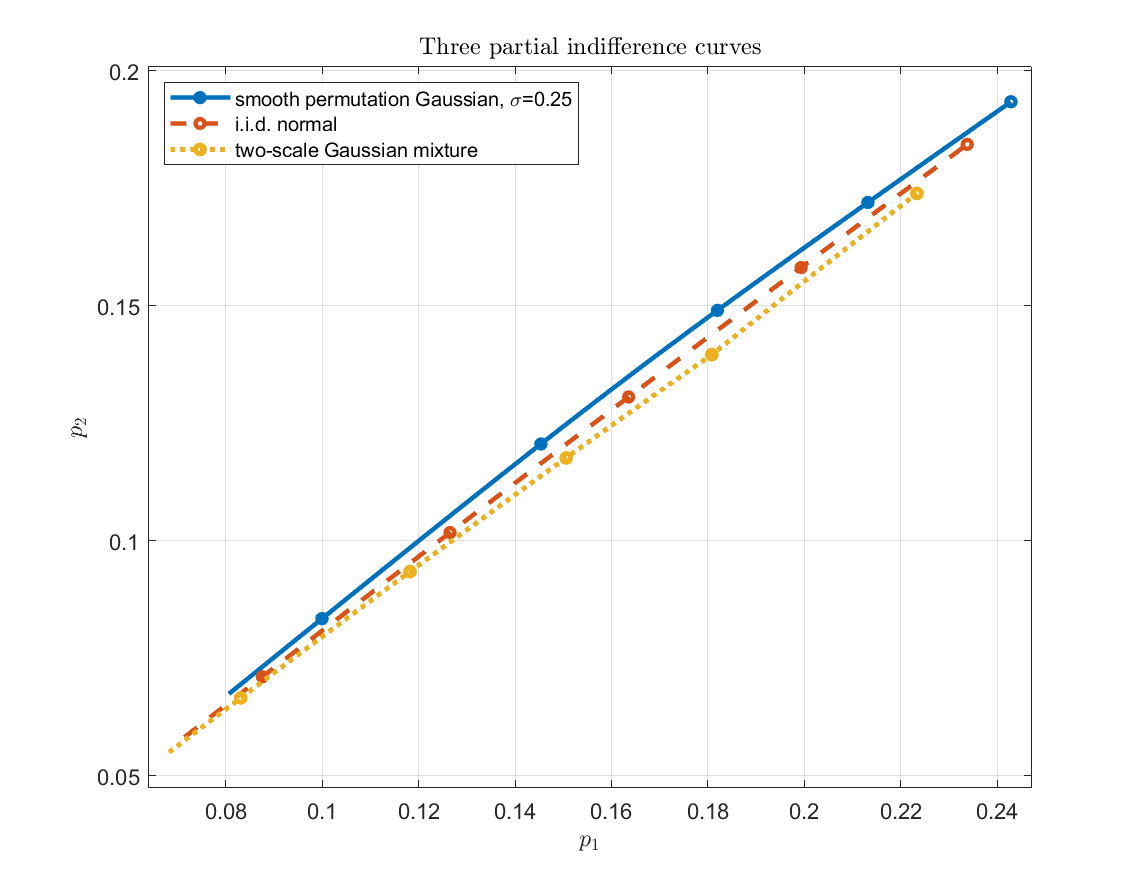}
    \end{minipage}
    \caption{
In both panels,
    $\gamma_{123}=0.80$, $\gamma_{13}=0.18$,
    $\gamma_{23}=0.015$, and $\gamma_{12}=0.005$.
    \\
    Left: the slope ratio $p_1'(a)/p_2'(a)$ under i.i.d.
    standard normal unobservables.  
    \\
    Right: a zoomed plot of
    $\mathcal{I}^{M, \cup_{\boldsymbol{\beta}}}_{1,2}$ in the original $(p_1,p_2)$ coordinates for
    the smooth permutation-Gaussian mixture (solid), the i.i.d. normal model (dashed),
    and the two-scale Gaussian mixture (dotted) described in the text.  The
    displayed range is $u\in[0.08,0.30]$.}
    \label{fig:ind_assumption} 
\end{figure}

\vskip 0.05in

\noindent \textit{Distribution of $\boldsymbol{\varepsilon}|\mathbf{x}$ depends on $\mathbf{x}$.} When the distribution of $\boldsymbol{\varepsilon}|\mathbf{x}$ {depends} on $\mathbf{x}$, the findings change. Continuing with  $J=3$ and now even fixing $\boldsymbol{\beta}$,  the set  $\mathcal{I}^{M, \boldsymbol{\beta}}_{k_1,k_2}$ is described as a collection of choice probabilities with 
$$P(Y=j|\mathbf{x})
 = \gamma_{123}\, C_{2,\mathbf{x}}(0.5,\,F_{\mathbf{x}}(a_{k_1,k_3}(\mathbf{x})))
 + \gamma_{\{k_1,k_2\}}/2 
 + \gamma_{\{j,k_3\}}\, F_{\mathbf{x}}(a_{k_1,k_3}(\mathbf{x})), \quad j \in \{k_1,k_2\},$$
where 
$C_{2,\mathbf{x}}$ denotes the bivariate copula of  $(\varepsilon_{k_2}-\varepsilon_{k_1},\varepsilon_{k_3}-\varepsilon_{k_1}) | \mathbf{x}$ and $F_{\mathbf{x}}$ is the c.d.f. of $\varepsilon_{k_1}-\varepsilon_{k_2}|\mathbf{x}$. Thus, the choice probabilities depend on what can be taken to be a parameter $F_\mathbf{x}(a_{k_1,k_3}(\mathbf{x}))$ but also depend on $C_{2,\mathbf{x}}$.

First, $\mathcal{I}^{M,\boldsymbol{\beta}}_{k_1,k_2}$ in
$\relint(\Delta^{M,\boldsymbol{\beta}}_2)$ need not have topological
dimension one less than that of $\relint(\Delta^{M,\boldsymbol{\beta}}_2)$. Even when it does, it need not be contained in a
codimension-one embedded manifold of 
$\relint(\Delta^{M,\boldsymbol{\beta}}_2)$. This is once again because
Assumption~\ref{assn:multi_distributionMED} imposes no restriction on how
the joint distribution of $\boldsymbol{\varepsilon}|\mathbf{x}$ varies
with $\mathbf{x}$. Even if, in a particular structural model,
$\mathcal{I}^{M,\boldsymbol{\beta}}_{k_1,k_2}$ is contained in a
codimension-one embedded manifold, this alone does not guarantee the required
separation property. The way the ordering images $\mathcal{P}^{+,M,\boldsymbol{\beta}}_{k_1,k_2}$ and
$\mathcal{P}^{+,M,\boldsymbol{\beta}}_{k_2,k_1}$ are generated may itself
vary discontinuously with $\mathbf{x}$, as illustrated below, so that
probability vectors associated with opposite rankings can lie on the same
side of, or overlap across, the partial indifference locus. Consequently,
$\mathcal{I}^{M,\boldsymbol{\beta}}_{k_1,k_2}$ may fail to separate
$\relint(\Delta^{M,\boldsymbol{\beta}}_2)$ into the two disjoint ordering
regions
$\mathcal{P}^{+,M,\boldsymbol{\beta}}_{k_1,k_2}
 \cap \relint(\Delta^{M,\boldsymbol{\beta}}_2)$
and
$\mathcal{P}^{+,M,\boldsymbol{\beta}}_{k_2,k_1}
 \cap \relint(\Delta^{M,\boldsymbol{\beta}}_2)$. Imposing suitable continuity in $\mathbf{x}$ of the 
distribution of $\boldsymbol{\varepsilon}|\mathbf{x}$ would not by
itself resolve the problem. Such continuity may produce a
codimension-one  indifference image, but it need not have separation properties as self-intersections or other singularities
may remain. Thus, additional structure is required for
$\mathcal{I}^{M,\boldsymbol{\beta}}_{k_1,k_2}$ to deliver the desired
separation property.

The right panel of Figure \ref{fig:ind_assumption}  can help  illustrate what can happen under such dependence. Suppose all $x_k$ are one-dimensional with $\beta_1=\beta_2=\beta_3=1/\sqrt{3}$. Using the three distributions described above, let the conditional unobservable distribution be the smooth permutation-Gaussian model when $x_1>x_2$, the two-scale Gaussian mixture when $x_2>x_1$, and the i.i.d. normal model when $x_1=x_2$.  The formal pairwise indifference locus is the i.i.d. normal curve,
since this is the distribution used when $x_1=x_2$. As $\mathbf x$
approaches $x_1=x_2$ from the two strict ordering sides, however,
the corresponding probability images approach the upper and lower
 curves specific to their regimes (see  the right panel). Under either fixed distribution, every probability vector in $\relint(\widetilde{\Delta}^B_2)$ is attainable. At \(p_1=0.16\), the curve corresponding to \(x_1=x_2\) lies at \(p_2\approx0.1322\) under the permutation-Gaussian distribution and at \(p_2\approx0.1245\) under the two-scale mixture. A point between them—for example, \((p_1,p_2)=(0.16,0.128)\)—arises with \(x_1>x_2\) under the first distribution and with \(x_1<x_2\) under the second. Both cases are allowed by the specified conditional model, so the same choice probabilities can arise under opposite index rankings. This establishes overlap. Thus, smooth indifference curves under each fixed distribution do not ensure ranking recovery when the distribution changes discontinuously with $\mathbf{x}$, even for fixed $\boldsymbol{\beta}$.

Fixing the trivariate copula, including at independence as in \citet{Manski1975}, does not by itself fix the induced bivariate copula of \((\varepsilon_{k_1}-\varepsilon_{k_2},\varepsilon_{k_1}-\varepsilon_{k_3})| \mathbf x\), which may vary with the common  $F_{\mathbf{x}}$. Thus, this restriction alone does not establish a common probability map. Whether such variation produces ranking ambiguity requires a separate argument.

A  useful assumption  on how $\boldsymbol{\varepsilon}|\mathbf{x}$ varies with $\mathbf{x}$ is  a  location-scale structure for the
unobservables: $\varepsilon_j=\mu(\mathbf{x})+s(\mathbf{x}) \widetilde{\varepsilon}_{j}$, $s(\mathbf{x})>0,$ 
where  $\widetilde{\boldsymbol\varepsilon}
=(\widetilde\varepsilon_1,\ldots,\widetilde\varepsilon_J)^{\top}$ is independent of $\mathbf x$.

Fix alternative 1 as a reference and define $d^{\boldsymbol\beta}(\mathbf x)
=
\left(0, 
\frac{x_2\beta_2-x_1\beta_1}{s(\mathbf x)},\ldots,
\frac{x_J\beta_J-x_1\beta_1}{s(\mathbf x)}
\right)^{\top}$. Let $\mathcal E^o_{\mathrm{diff}}$ be the relative interior of the support of
$(0, \widetilde\varepsilon_1-\widetilde\varepsilon_2,\ldots,
\widetilde\varepsilon_1-\widetilde\varepsilon_J)^{\top}$ in its $(J-1)$-dimensional  affine hull, and let $\mathcal D_{\boldsymbol\beta}
=
\{d^{\boldsymbol\beta}(\mathbf x):\mathbf x\in\mathbb R^{M_e}\}
\cap\mathcal E^o_{\mathrm{diff}}$ denote the attainable normalized index set restricted to
$\mathcal E^o_{\mathrm{diff}}$. Note that $(d^{\boldsymbol\beta})_j-(d^{\boldsymbol\beta})_k$ and $x_j\beta_j-x_k\beta_k$ have the same sign. Theorem \ref{th:cs_heteroskedasticity} proves a common structural partitioning at the \textit{Structural environment} uniformly over $\boldsymbol{\beta}$ and implies ranking recovery on $\Psi(\mathcal D_{\boldsymbol\beta})$ for each $\boldsymbol\beta$.

\begin{theorem}[Common structural partition and ranking recovery]
\label{th:cs_heteroskedasticity}
Suppose Assumptions~\ref{assn:indCS}--\ref{assn:rationalCS} hold,
every two alternatives are connected in the sense of
Definition~\ref{def:cs_connected}, and
$\varepsilon_j=\mu(\mathbf x)+s(\mathbf x)\widetilde\varepsilon_j$,
 $s(\mathbf x)>0$, 
where $\widetilde{\boldsymbol\varepsilon}$ is independent of $\mathbf x$
and has an absolutely continuous exchangeable distribution whose support
is convex and has nonempty interior.

Then there exists a map
$\Psi:\mathcal E^o_{\mathrm{diff}}
\longrightarrow
\relint(\Delta_{J-1})$
such that
$\mathbf p(\mathbf x)
=
\Psi(d^{\boldsymbol\beta}(\mathbf x))$
whenever
$d^{\boldsymbol\beta}(\mathbf x)\in\mathcal E^o_{\mathrm{diff}}$.
Moreover, $\Psi$ is a topological embedding (hence, in particular,
injective).

\smallskip
\noindent
A. (Ranking recovery) For every $\boldsymbol\beta\in\mathcal B_0$, ranking recovery holds on $\Psi(\mathcal D_{\boldsymbol\beta})$ as every probability vector in
$\Psi(\mathcal D_{\boldsymbol\beta})$ uniquely determines the normalized
index vector and hence all pairwise rankings of the deterministic utility
indices.

\smallskip
\noindent
B. (Common ranking regions.) 
For every $\boldsymbol\beta\in\mathcal B_0$, every $j\neq k$, and
$\bowtie\in\{<,=,>\}$, we have 
$\Psi\!\left(
\mathcal D_{\boldsymbol\beta}
\cap
\{d:d_j\bowtie d_k\}
\right)
=
\Psi(\mathcal D_{\boldsymbol\beta})
\cap
\Psi\!\left(
\mathcal E^o_{\mathrm{diff}}
\cap
\{d:d_j\bowtie d_k\}
\right)$. Thus, the pairwise ranking regions in
$\Psi(\mathcal E^o_{\mathrm{diff}})$ are common to all
$\boldsymbol\beta$, while their attained traces may depend on
$\boldsymbol\beta$.

Likewise, 
$$\Psi\!\left(
\mathcal D_{\boldsymbol\beta}
\cap
\{d:d_{k_1}>d_{k_2}>\cdots>d_{k_J}\}
\right)
=
\Psi(\mathcal D_{\boldsymbol\beta})
\cap
\Psi\!\left(
\mathcal E^o_{\mathrm{diff}}
\cap
\{d:d_{k_1}>d_{k_2}>\cdots>d_{k_J}\}
\right)$$ for every permutation $(k_1,\ldots,k_J)$ of $\mathcal J$.  The $J!$ complete order images (as well as their attained traces)  are pairwise disjoint. 

\smallskip
\noindent
C. 
For every $j\neq k$,
$\Psi\!\left(
\mathcal E^o_{\mathrm{diff}}
\cap
\{d:d_j=d_k\}
\right)$ is an embedded $(J-2)$-dimensional topological manifold in
$\Psi(\mathcal E^o_{\mathrm{diff}})$. Its complement has exactly two
connected components,
$\Psi\!\left(
\mathcal E^o_{\mathrm{diff}}\cap\{d:d_j>d_k\}
\right)$
 and 
$\Psi\!\left(
\mathcal E^o_{\mathrm{diff}}\cap\{d:d_j<d_k\}
\right)$.

Moreover, $0\in\mathcal E^o_{\mathrm{diff}}$, and all $J!$ strict order
regions are nonempty and connected, with $\Psi(0)$ in the closure of
each.
\end{theorem}

Thus, Theorem~\ref{th:cs_heteroskedasticity} gives a common structural
ranking partition of $\Psi(\mathcal E^o_{\mathrm{diff}})$. Each $\boldsymbol\beta$ traces its attainable part
$\Psi(\mathcal D_{\boldsymbol\beta})$. If
$\mathcal E^o_{\mathrm{diff}}
\subseteq
\bigcup_{\boldsymbol\beta\in\mathcal B_0}
\mathcal D_{\boldsymbol\beta}$,
the whole common partition is attained as $\boldsymbol\beta$ varies. Without this additional condition, the partition still describes the common structural geometry, while ranking recovery applies to its attained traces. Thus, under conditions of this theorem, 
the \emph{Structural environment} is ``quantile'' in the sense defined
above. Since $\Psi$ is a homeomorphism onto its image,  topological dimensions and the corresponding separation properties are preserved in probability space.

\subsection{Distributional route:  ``quantile'' version?}

\label{sec:distributionalmulti}

Let's now maintain the standard random utility paradigm and ask when restrictions on the distribution of $\boldsymbol\varepsilon|\mathbf{x}$ generate a ``quantile'' ranking partition or deliver only ranking recovery on the attainable probability set. Suppose generally that $\boldsymbol\varepsilon|\mathbf x$ has a convex support with nonempty interior and an absolutely continuous distribution, with copula $C_{\mathbf x}$ and marginal c.d.f.s $F_{j,\mathbf x}$.

Under complete index indifference \eqref{totalindiff}, the image in the probability simplex is $\pi_k^*(\mathbf x)
=
P\!\left(
\varepsilon_k\geq\varepsilon_j\ \text{for all }j
\mid\mathbf x
\right)$,
$ k=1,\ldots,J$. If the conditional distribution of $\boldsymbol\varepsilon$ varies freely
with $\mathbf x$, then $\boldsymbol\pi^*(\mathbf x)$ can also vary freely and
need not vary continuously. Thus even total indifference need not have a
single image at the \textit{Structural environment} level.\footnote{It need not have a
single image even at the \textit{empirical specification} level, but it is not analyzed in the paper.} A fixed copula $C_{\mathbf x}=C$ and common marginals $F_{j,\mathbf x}=F_{\mathbf x}$ for every $j$ stabilize this point: $\pi_k^{*,M}
=
\int_0^1
\frac{\partial C}{\partial u_k}(v,\ldots,v)\,dv$, $k=1,\ldots,J$, does not depend on $\mathbf x$.\footnote{Unless $C$ is exchangeable, this point need
not be the barycentre.} Constancy of the total indifference image point is not enough to make a partial indifference relationship invariant with respect to $\mathbf x$. As the shape of $F_{\mathbf x}$ changes with $\mathbf x$, the relationship between index differences and choice may also change with $\mathbf x$, even when the copula is fixed. Appendix C gives an explicit example with the same copula and different marginal c.d.f.s , producing partial indifference curves with different tangent directions at the same total indifference point.

Thus, a fixed copula and common marginals across alternatives do not by themselves deliver an $\mathbf x$-invariant pairwise indifference boundary.

More generally, when the conditional distribution of the unobservables is allowed to vary with $\mathbf x$, probability vectors generated by opposite index rankings may overlap. Appendix~C gives an explicit overlap
construction, and Proposition~\ref{prop:distributional_no_uniform_recovery}
there formalizes the resulting failure of ranking recovery even for a
fixed $\boldsymbol{\beta}$.

Just like in the limited attention models in Section \ref{sec:cs_structural}, a common location-scale structure removes this source of ambiguity by producing one fixed probability map  and allowing its indifference boundaries to be nonlinear. It is  generally asymmetric due to copula non-exchangeability (in Section \ref{sec:cs_structural} asymmetry of such a map was driven by non-symmetric attention probabilities).

\begin{theorem}[Common structural partition and ranking recovery]
\label{th:nonexch_specialcase}
Suppose the standard random utility paradigm holds and
$\varepsilon_j=\mu(\mathbf x)+s(\mathbf x)\widetilde\varepsilon_j$, $s(\mathbf x)>0$, 
where $\widetilde{\boldsymbol\varepsilon}$ is independent of $\mathbf x$.
Suppose its distribution has convex support with nonempty interior, its
marginals have the same absolutely continuous c.d.f. $\widetilde F$,
and its copula $\widetilde C$ is absolutely continuous but need not be
exchangeable. Let $\mathcal E^o_{\mathrm{diff}}$,
$d^{\boldsymbol\beta}(\mathbf x)$, and
$\mathcal D_{\boldsymbol\beta}$ be defined as in
Section~\ref{sec:cs_structural} for
$\widetilde{\boldsymbol\varepsilon}$.

Then there exists a map $\Psi:\mathcal E^o_{\mathrm{diff}}
\rightarrow
\relint(\Delta_{J-1})$ such that
$\mathbf p(\mathbf x)
=
\Psi(d^{\boldsymbol\beta}(\mathbf x))$
whenever
$d^{\boldsymbol\beta}(\mathbf x)\in\mathcal E^o_{\mathrm{diff}}$.
Moreover, $\Psi$ is a topological embedding (and, hence, injective).

\smallskip
\noindent
A. (Ranking recovery.) 
For every $\boldsymbol\beta\in\mathcal B_0$, ranking recovery holds on $\Psi(\mathcal D_{\boldsymbol\beta})$ as every probability vector in
$\Psi(\mathcal D_{\boldsymbol\beta})$ uniquely determines the normalized
index vector and hence all pairwise rankings of the deterministic utility
indices.

\smallskip
\noindent
B.(Common ranking regions.) 
For every $\boldsymbol\beta\in\mathcal B_0$, every $j\neq k$, and
$\bowtie\in\{<,=,>\}$, we have 
$\Psi\!\left(
\mathcal D_{\boldsymbol\beta}
\cap
\{d:d_j\bowtie d_k\}
\right)
=
\Psi(\mathcal D_{\boldsymbol\beta})
\cap
\Psi\!\left(
\mathcal E^o_{\mathrm{diff}}
\cap
\{d:d_j\bowtie d_k\}
\right)$. Thus, the pairwise ranking regions in
$\Psi(\mathcal E^o_{\mathrm{diff}})$ are common to all
$\boldsymbol\beta$, while their attained traces may depend on
$\boldsymbol\beta$.

Likewise, $$\Psi\!\left(
\mathcal D_{\boldsymbol\beta}
\cap
\{d:d_{k_1}>d_{k_2}>\cdots>d_{k_J}\}
\right)
=
\Psi(\mathcal D_{\boldsymbol\beta})
\cap
\Psi\!\left(
\mathcal E^o_{\mathrm{diff}}
\cap
\{d:d_{k_1}>d_{k_2}>\cdots>d_{k_J}\}
\right)$$ for every permutation $(k_1,\ldots,k_J)$ of $\mathcal J$. 
The $J!$ common strict order regions on the right-hand side are pairwise
disjoint whenever they are nonempty.

\smallskip
\noindent
C. For every $j\neq k$, $\Psi\!\left(
\mathcal E^o_{\mathrm{diff}}
\cap
\{d:d_j=d_k\}
\right)$ is an embedded $(J-2)$-dimensional topological manifold in
$\Psi(\mathcal E^o_{\mathrm{diff}})$. Its complement has exactly two
connected components, $\Psi\!\left(
\mathcal E^o_{\mathrm{diff}}\cap\{d:d_j>d_k\}
\right)$
 and 
$\Psi\!\left(
\mathcal E^o_{\mathrm{diff}}\cap\{d:d_j<d_k\}
\right)$.

If $0\in\mathcal E^o_{\mathrm{diff}}$, all $J!$ strict order regions
are nonempty and connected, with $\Psi(0)$ in the closure of each.
Moreover,
$\Psi(0)=\boldsymbol\pi^{*,M}$, where $\pi_j^{*,M}
=
\int_0^1
\frac{\partial\widetilde C}
{\partial u_j}(v,\ldots,v)\,dv$,
 $j=1,\ldots,J$. 
\end{theorem}

The same discussion as after Theorem \ref{th:cs_heteroskedasticity} applies here. The common map $\Psi$ gives a \textit{Structural environment} ``quantile'' partition of $\Psi(\mathcal E^o_{\mathrm{diff}})$  and, for each $\boldsymbol\beta$, ranking recovery on the attainable trace $\Psi(\mathcal D_{\boldsymbol\beta})$. If, in addition,$\mathcal E^o_{\mathrm{diff}} \subseteq \bigcup_{\boldsymbol\beta\in\mathcal B_0} \mathcal D_{\boldsymbol\beta}$, the whole common partition is attained as$\boldsymbol\beta$ varies.

Figure~\ref{fig:nonexch_gaussian} illustrates both the structural
indifference curve and its dependence on the unobservables  distribution, for simplicity taking $s(\mathbf{x})=1$. In each
case, $\widetilde{\boldsymbol\varepsilon}$ is a mean-zero trivariate normal
vector with unit variances. Let $\rho_{kj}$ denote the pairwise correlation between $\widetilde{\varepsilon}_k$ and $\widetilde{\varepsilon}_j$. We consider the following three cases of $(\rho_{12},\rho_{13},\rho_{23})$:  (a) case 1 is the case of all positive correlations 
$(0.15,\,0.35,\,0.55)$; (b) case 2 is the case of one negative and two positive correlations $ (-0.20,\,0.25,\,0.60)$; (c) case 3 is the case of two negative and one positive 
$(-0.3,\,-0.1,\,0.7)$. For these three cases, Figure \ref{fig:nonexch_gaussian} plots $(p_1(a),p_2(a))$ across $a$, where
\begin{equation*} 
p_k(a)  =  \int_{-\infty}^{+\infty} \frac{\partial C}{\partial u_k}( F(e_1), F(e_1), F(e_1 -a)) dF(e_1), \quad k=1,2,
\end{equation*}
which fully characterize partial indifference cases $d_2=0$ with $a$ standing for $d_3$. Figure \ref{fig:nonexch_gaussian} confirms that different correlation structures produce  different curvatures. The map $\Psi$ is common across values of
$\mathbf x$ only after the distribution of
$\widetilde{\boldsymbol\varepsilon}$ has been fixed. Within any one of these
structural environments, a particular $\boldsymbol\beta$ may still trace only
a segment or another proper subset of the corresponding curve.

\begin{figure}[htbp]
    \centering
    \includegraphics[width=0.7\linewidth]{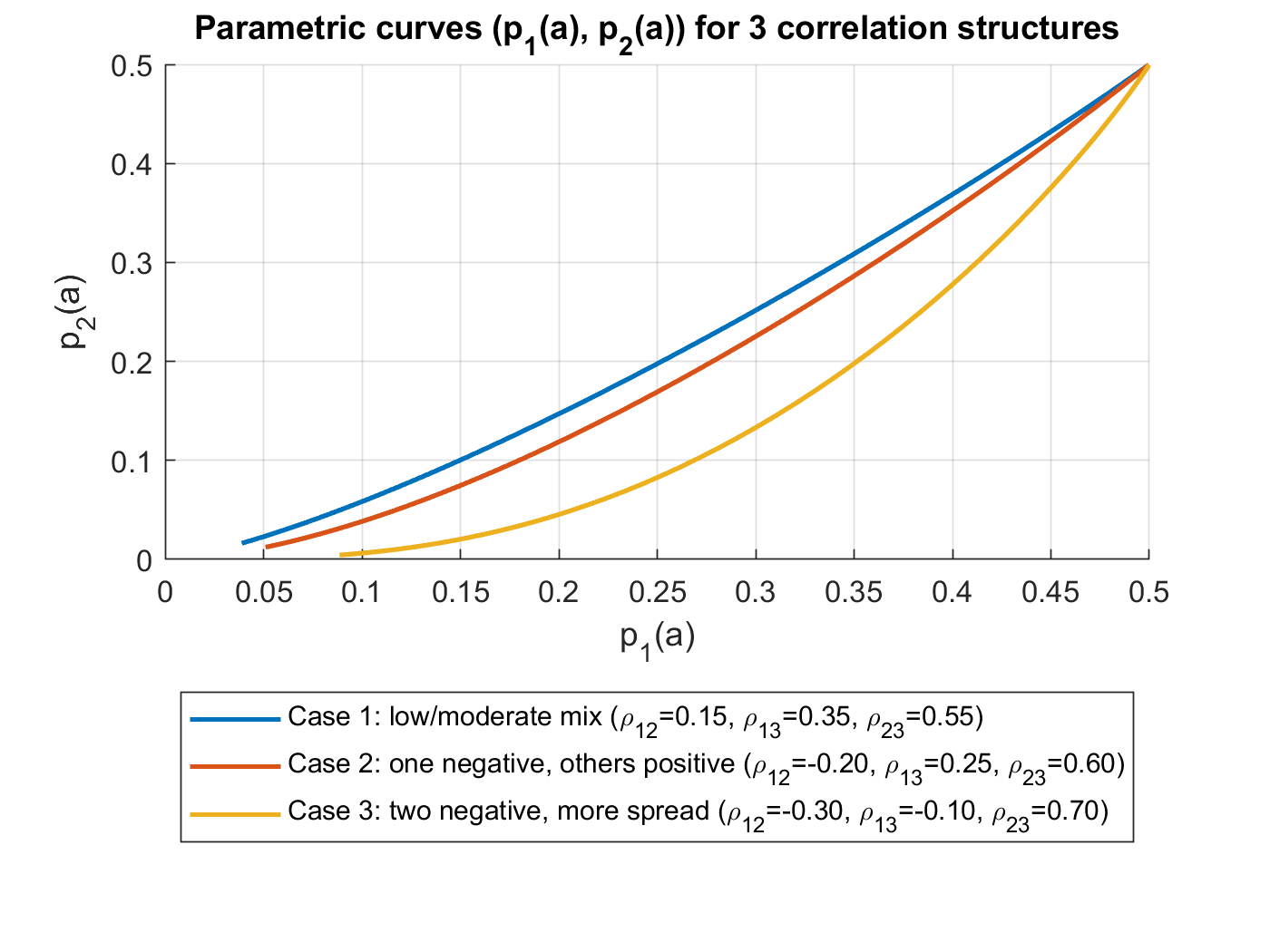}
    \caption{\label{fig:nonexch_gaussian}
    Indifference between alternatives $1$ and $2$ under three
    non-exchangeable trivariate normal unobservables  distributions. Each curve plots
    $(p_1,p_2)$ along $d_2=0$ as $d_3$ varies. The probability of alternative
    $3$ is $p_3=1-p_1-p_2$.}
\end{figure}

If the conditional distribution of $\boldsymbol\varepsilon$ varies with
$\mathbf x$ in a way that cannot be absorbed by a common location and scale, the challenges are analogous to those discussed in Section \ref{sec:cs_structural}. Namely,  there need not be a common index representation of choice probabilities across covariate values or a single map $\Psi$ that applies at every covariate value. Different
values of $\mathbf x$ may then be associated with different structural maps,
such as those generated by the three Gaussian specifications in
Figure~\ref{fig:nonexch_gaussian}. Even if each map separately generates a
smooth partial indifference codimension-one manifold, these manifolds may intersect or overlap when combined
across covariate values, and probability vectors associated with opposite
index rankings may coincide, producing ranking ambiguity.

\section{Conclusion}
\label{sec:conclusion}

A central insight of \citet{Manski1975} is that choice probabilities can
reveal ordinal information about deterministic utilities without a 
parametric model of the utility unobservables. That insight is classically obtained under utility maximization with
conditionally i.i.d. (more generally, exchangeable) unobservables. This paper  first studies the behavioral and distributional foundations of quantile (in particular, median) characterizations in binary choice, which guarantee ranking recovery in that setting.  It then extends these two approaches to multinomial choice, thus leaving \cite{Manski1975} benchmark behind, and asks when analogous ranking interpretations remain valid and at which level of robustness.

The extension is far from automatic.  With several alternatives, interactions across multiple comparison margins can lead to ranking ambiguity with no counterpart in binary choice. Nevertheless, under limited attention,  balance restrictions on attention probabilities give linear separators whose strict signs reveal index rankings uniformly over admissible unobservables distributions and parameter values. Under the stated conditions, a  common location-scale structure also allows ranking recovery within a fixed \textit{Structural environment}, even when the ranking boundaries are nonlinear. 

The results clarify how ordinal information survives beyond standard utility maximization with exchangeable unobservables and which restrictions are needed to recover it.

\bibliographystyle{ecta}
\bibliography{BehMS.bib}


\clearpage 
\appendix
\setcounter{page}{1}
\numberwithin{equation}{section}

\setcounter{section}{0}

\begin{center}
    {\LARGE \textbf{Appendix}}
\end{center}

\section{Proofs and further examples for sections \ref{sec:binary} and \ref{sec:multinomial}}

\noindent \textbf{Proof of Theorem \ref{prop:median}.} 
Exchangeability of $(\varepsilon_1,\varepsilon_0)| x$ in Assumption~\ref{assn:binary_distributionMED} implies that the distribution of $\varepsilon_1-\varepsilon_0| x
$ is symmetric, with absolute continuity guaranteeing that 
$F_{\varepsilon_1-\varepsilon_0|x}(0)=1/2$.
Convexity and nonempty interior of the joint support imply that
the support of $\varepsilon_1-\varepsilon_0| x$ is an interval with
zero in its interior. Thus its c.d.f. is strictly increasing on the
interior of its support, and
\begin{equation}
\label{eq:proof1_aux}
P(U_1\geq U_0| x)
=
1-F_{\varepsilon_1-\varepsilon_0\mid x}(-x\beta)
\geq(\leq)\frac12
\quad\Longleftrightarrow\quad
x\beta\geq(\leq)0.
\end{equation}

\textbf{A}. Use the fact  that condition (2) can be written as 
$P(Y=1|U_1 \geq  U_0,x) + P(Y=1|U_1 < U_0,x)= 2(1-\tau)$ 
a.e. (and, of course, it implies $P(Y=0|U_1 \geq  U_0,x) + P(Y=0|U_1 < U_0,x)=2 \tau$.) Using the law of total probability and some simple algebra, 
\begin{align} P(Y=1|x) & = P(Y=1|U_1 \geq  U_0,x)P(U_1 \geq  U_0 |x) + P(Y=1|U_1 <  U_0,x) P(U_1 <  U_0 |x) 
\notag \\
&= P(Y=1|U_1 <  U_0,x) \notag \\
& + \left(P(Y=1|U_1 \geq  U_0,x)-P(Y=1|U_1 <  U_0,x)\right)P(U_1 \geq  U_0 |x).  \label{Th1proof1}
\end{align} 

Using \eqref{eq:proof1_aux} 
 and the properties of the right-hand side of \eqref{eq:proof1_aux}  under Assumption \ref{assn:binary_distributionMED}, and also taking into account that 
 $P(Y=1|U_1 \geq  U_0,x)-P(Y=1|U_1 <  U_0,x)>0$ by condition (1), we obtain that
\begin{align*} x\beta \geq (\leq) 0 \iff &  P(Y=1|x) \geq (\leq) P(Y=1|U_1 <  U_0,x) \\
& + 0.5\left(P(Y=1|U_1 \geq  U_0,x)-P(Y=1|U_1 <  U_0,x)\right) \\ & = 0.5\left(P(Y=1|U_1 \geq  U_0,x)+P(Y=1|U_1 <  U_0,x)\right)=1-\tau. 
\end{align*}

\vskip 0.05in

\textbf{B}. By Condition (2) and \eqref{Th1proof1},
{\small$$P(Y=1|x)-(1-\tau)
=
\left(
P(Y=1|U_1\geq U_0,x)-P(Y=1|U_1<U_0,x)
\right)
(
P(U_1\geq U_0|x)-0.5
).$$}Under Assumption \ref{assn:binary_distributionMED}, $P(U_1\geq U_0|x)-0.5
>(<)0
\Longleftrightarrow
x\beta>(<)0$, while \eqref{quantilenECON} implies
$P(Y=1|x)-(1-\tau)
>(<)0
\Longleftrightarrow
x\beta>(<)0$.
Hence, $P(Y=1|U_1\geq U_0,x)-P(Y=1|U_1<U_0,x)>0$
a.e. on $\{x:x\beta\neq0\}$.

\vskip 0.05in

\textbf{C.} Let us show that 
for every fixed \(\beta\neq0\), there exists a generic model satisfying Assumption~\ref{assn:binary_distributionMED}  and the strict version of Condition (1), but violating Condition (2) on a set of positive $P_X$-measure, for which \eqref{quantilenECON} fails.

Fix $\beta\neq 0$. Let $x$ have full support on $\mathbb{R}^k$, and let $\varepsilon_0$ and $\varepsilon_1$ be independent standard normal random variables, independent of $x$. Set  $U_1=x\beta+\varepsilon_1$,  $U_0=\varepsilon_0$.  Then Assumption 1 holds and $P(U_1\geq U_0| X=x) = \Phi\!\left(\frac{x\beta}{\sqrt{2}}\right)$.  Choose $\eta,d>0$ sufficiently small that all the probabilities below belong to $(0,1)$. By continuity of $\Phi$, choose $\xi>0$ such that  $d\left[ \frac12-\Phi\!\left(-\frac{\xi}{\sqrt{2}}\right) \right] <\eta$.  Specify  
$$P(Y=1| U_1\geq U_0,x) = \begin{cases} 1-\tau+\eta+\dfrac d2, & -\xi<x\beta<0,\\[4pt] 1-\tau+\dfrac d2, & \text{otherwise}, \end{cases}$$  
$$P(Y=1\mid U_1<U_0,X=x) = \begin{cases} 1-\tau+\eta-\dfrac d2, & -\xi<x\beta<0,\\[4pt] 1-\tau-\dfrac d2, & \text{otherwise}. \end{cases}$$  
The difference between these two conditional probabilities is $d>0$ for every $x$, so the strict version of Condition (1) holds. Their sum, however, equals $2(1-\tau+\eta)$ whenever $-\xi<x\beta<0$, so Condition (2) fails on that set. Since $\beta\neq0$ and $x$ has full support, this set has positive $P_X$-measure. For every $x$ in this set, the law of total probability gives 
$$P(Y=1|X=x) = 1-\tau+\eta +d\left[ \Phi\!\left(\frac{x\beta}{\sqrt{2}}\right)-\frac12 \right] > 1-\tau,$$ 
where the inequality follows from the choice of $\xi$. Hence, on a set of positive $P_X$-measure, \[ x\beta<0 \qquad\text{but}\qquad P(Y=1| X=x)>1-\tau, \] contradicting \eqref{quantilenECON}. $\blacksquare$

\vskip 0.05in 

\textbf{Proof of Proposition \ref{prop:implications}.} 
Relation \eqref{eq:proof1_aux} and condition (1) of Theorem \ref{prop:median} imply 
$$P(Y=1|x) \geq P(Y=1|U_1<U_0, x) + 0 \cdot P(Y=1|U_1 \geq U_0, x)=P(Y=1|U_1<U_0, x),$$
\begin{multline*}P(Y=1|x) \leq P(Y=1|U_1<U_0, x) + (P(Y=1|U_1\geq U_0, x)-P(Y=1|U_1<U_0, x)) \cdot 1 \\ =P(Y=1|U_1 \geq U_0, x).
\end{multline*}
Condition (2) implies then that 
$P(Y=1|x) \leq 2(1-\tau).$  
This inequality only provides a meaningful upper bound (that is, less than 1) if $\tau>1/2$. Condition (2) also implies  that 
$$P(Y=1|x) \geq P(Y=1|U_1 < U_0, x) = 2(1-\tau) - P(Y=1|U_1 \geq U_0, x) \geq 2(1-\tau) - 1.$$ 
This provides a meaningful lower bound (that is, greater than 0) only when $\tau<1/2$.
$\blacksquare$ 

\vskip 0.05in 

\textbf{Proof of Theorem \ref{prop:copula}.} 

\textbf{A.} By the law of total probability,
\begin{align*}
P(Y=1|x)
&=
P(Y=1|U_1<U_0,x)\\
&\quad+
\left[
P(Y=1|U_1\geq U_0,x)
-
P(Y=1|U_1<U_0,x)
\right]
P(\varepsilon_1-\varepsilon_0\geq -x\beta|x).
\end{align*}
Condition (2') gives
\begin{align*}
1-\tau
&=
P(Y=1|U_1<U_0,x)+
\left[
P(Y=1|U_1\geq U_0,x)
-
P(Y=1|U_1<U_0,x)
\right]
P(\varepsilon_1-\varepsilon_0\geq0|x).
\end{align*}
Hence, 
\begin{align*}
P(Y=1|x)-(1-\tau)
&=
\left[
P(Y=1|U_1\geq U_0,x)
-
P(Y=1|U_1<U_0,x)
\right]\\
&\quad\times
\left[
P(\varepsilon_1-\varepsilon_0\geq -x\beta|x)
-
P(\varepsilon_1-\varepsilon_0\geq0|x)
\right].
\end{align*}
By the maintained support conditions, the second factor is strictly
positive, zero, or strictly negative according as
$x\beta>0$, $x\beta=0$, or $x\beta<0$. Condition (1') implies that the
first factor is strictly positive a.e. Therefore, 
$P(Y=1|x)-(1-\tau) \geq 0 (\leq 0) \; \iff \; x\beta \geq 0 (\leq 0),$ 
which proves \eqref{quantilenECON}.

\textbf{B} and \textbf{C}. The proofs are analogous to those of parts \textbf{B} and \textbf{C} of Theorem \ref{prop:median}. For part \textbf{B}, $0.5$ is replaced by $P(\varepsilon_1\geq\varepsilon_0|x)$, with the required strict inequalities following from the support conditions. For part \textbf{C}, the construction used there satisfies the present support conditions and Condition (2') reduces to Condition (2).

\vskip 0.1in

\noindent \textbf{Proof of Theorem \ref{th:multimedian}.} Fix $\mathbf x$. By absolute continuity, the utility maximizer is
unique a.s. Apply the law of total probability over events
$\{U_m=U^{J:J}\}$ of positive  probability conditional on $x$.

By Conditions (2)--(3), conditional on any such event, the probability
of choosing the maximizer is the same across maximizing alternatives,
and so is the probability of choosing any particular nonmaximizer.
Hence, for any $r$ such that
$P(U_r=U^{J:J}|\mathbf x)>0$ and any $s\neq r$,
\[
\begin{aligned}
P(Y=1|\mathbf x)-P(Y=2|\mathbf x)
={}&
\bigl[
P(Y=r| U_r=U^{J:J},\mathbf x)
-
P(Y=s| U_r=U^{J:J},\mathbf x)
\bigr] \\
&\times
\bigl[
P(U_1=U^{J:J}|\mathbf x)
-
P(U_2=U^{J:J}|\mathbf x)
\bigr].
\end{aligned}
\]
By Condition (1), the first factor is strictly positive. Therefore, $P(Y=1|\mathbf x)-P(Y=2|\mathbf x)$ has the same sign as
$P(U_1=U^{J:J}|\mathbf x)
-
P(U_2=U^{J:J}|\mathbf x)$. Note that 
{\small\begin{multline*}
P(U_1 = U^{J:J}|\mathbf{x})  = P(\cap_{j=2}^J ( x_1\beta_1+\varepsilon_1 \geq x_j\beta_j+\varepsilon_j)) =\\
 = \int_{-\infty}^{+\infty} \frac{\partial C_\mathbf{x}}{\partial u_1} \left( F_{\varepsilon_1|\mathbf{x}}(\varepsilon|\mathbf{x}), F_{\varepsilon_2|\mathbf{x}}(\varepsilon +x_1 \beta_1-x_2 \beta_2), \ldots, F_{\varepsilon_J|\mathbf{x}}(\varepsilon +x_1 \beta_1-x_J \beta_J)\right) dF_{\varepsilon_1|\mathbf{x}} (\varepsilon|\mathbf{x}),
\end{multline*}
\begin{multline*}
P(U_2 = U^{J:J}|\mathbf{x})  = P(\cap_{j\neq 2}^J ( x_2\beta_2+\varepsilon_2 \geq x_j\beta_j+\varepsilon_j)) =\\
 = \int_{-\infty}^{+\infty} \frac{\partial C_x}{\partial u_2} \left( F_{\varepsilon_1|x}(\varepsilon +x_2 \beta_2-x_1 \beta_1), F_{\varepsilon_2|x}(\varepsilon|x), \ldots, F_{\varepsilon_J|\mathbf{x}}(\varepsilon +x_2 \beta_2-x_J \beta_J)\right) dF_{\varepsilon_2|\mathbf{x}} (\varepsilon|\mathbf{x}).
\end{multline*}}
Since all the conditional marginal distributions $F_{\varepsilon_j|\mathbf{x}}$, $j=1,\ldots, J$, are the same and the copula $C_\mathbf{x}$ is exchangeable, we can equivalently write 
{\small\begin{multline*}
P(U_2 = U^{J:J}|\mathbf{x})  =\\ 
\int_{-\infty}^{+\infty} \frac{\partial C_\mathbf{x}}{\partial u_1} \left( F_{\varepsilon_1|x}(\varepsilon|\mathbf{x}), F_{\varepsilon_2|\mathbf{x}}(\varepsilon +x_2 \beta_2-x_1\beta_1), \ldots, F_{\varepsilon_J|\mathbf{x}}(\varepsilon +x_2 \beta_2-x_J \beta_J)\right) dF_{\varepsilon_1|\mathbf{x}} (\varepsilon|\mathbf{x})
\end{multline*}}By copula properties, $\frac{\partial C_\mathbf{x}}{\partial u_1}(u_1,u_2,\ldots,u_J|\mathbf{x})$ is increasing in $u_2, \ldots, u_J$. As $x_1\beta_1 \geq x_2 \beta_2$ implies 
{\small$$F_{\varepsilon_2|\mathbf{x}}(\varepsilon +x_1 \beta_1-x_2\beta_2) \geq  F_{\varepsilon_2|x}(\varepsilon +x_2 \beta_2-x_1\beta_1), \quad 
F_{\varepsilon_j|\mathbf{x}}(\varepsilon +x_1 \beta_1-x_j\beta_j) \geq  F_{\varepsilon_j|\mathbf{x}}(\varepsilon +x_2 \beta_2-x_j\beta_j), $$}$j=3, \ldots,J.$ Using this monotonicity, we obtain $P(U_1=U^{J:J}|\mathbf{x})
\geq
P(U_2=U^{J:J}|\mathbf{x})$, and, 
by the argument above, $P(Y=1|\mathbf{x})\geq P(Y=2|\mathbf{x})$. Analogous analysis lets us  conclude that for any $k_1 \neq k_2$, having $x_{k_1}\beta_{k_1} \geq x_{k_2}\beta_{k_2}$ implies $P(Y=k_1|\mathbf{x}) - P(Y=k_2|\mathbf{x}) \geq 0$.

Suppose now we have $P(Y=k_1|\mathbf{x}) - P(Y=k_2|\mathbf{x}) > 0$. Suppose, contrary to what we want to prove, we have $x_{k_1}\beta_{k_1} \leq x_{k_2}\beta_{k_2}$. Then, from what we have already shown,  we are guaranteed $P(Y=k_1|\mathbf{x}) - P(Y=k_2|\mathbf{x}) \leq 0$,  which is a contradiction to the supposition of $P(Y=k_1|\mathbf{x}) - P(Y=k_2|\mathbf{x}) > 0$. 
$\blacksquare$

\vskip 0.1in 

 \textbf{Proof of Proposition \ref{prop:notconnected}.} Suppose some $k_1$ and $k_2$ are not connected in the sense of Definition \ref{def:cs_connected}. Let 
 $A(k_i)$ denote the set of all the elements in $\mathcal{J}$ connected to $k_i$ in the sense of Definition \ref{def:cs_connected}, $i=1,2$. By the conditions of this proposition and the construction, we have $A(k_1) \neq \emptyset$, $A(k_2) \neq \emptyset$, and  $A(k_1) \cap A(k_2)=\emptyset$, and  $A(k_1)$ and $A(k_2)$ are disjoint connected
components. Moreover, no active consideration set can intersect both
$A(k_i)$ and its complement: for $i=1,2$, if $\gamma_A>0$ and
$A\cap A(k_i)\neq\varnothing$, then $A\subseteq A(k_i)$. In light of this,  
\begin{multline*}P(Y \in A(k_1)|\mathbf{x}) = \sum_{A \in \mathcal{A}} \gamma_A P(Y \in A(k_1)|\mathbf{x}, A) = \sum_{A \in \mathcal{A}: A \subseteq A(k_1) } \gamma_A P(Y \in A(k_1)|\mathbf{x}, A) \\
= \sum_{A \in \mathcal{A}: A \subseteq A(k_1) } \gamma_A \cdot 1 = \sum_{A \in \mathcal{A}: A \subseteq A(k_1) } \gamma_A,
\end{multline*}
where the last expression does not depend on $\mathbf{x}$. 

It remains to show that $A(k_1)$ is a proper subset of $\mathcal J$ and that the constant lies in $(0,1)$. By construction, $k_1 \in A(k_1)$ and $k_2 \notin A(k_1)$, so $\emptyset \neq A(k_1) \subsetneq \mathcal J$.

By Assumption~\ref{assn:noemptyCS}, there exists an active set $B$ with $k_1 \in B$, which must satisfy $B \subseteq A(k_1)$, implying $\sum_{A \subseteq A(k_1)} \gamma_A > 0$. Likewise, there exists an active set $C$ with $k_2 \in C$, where $C \subseteq A(k_2)$ and $A(k_2)\cap A(k_1)=\emptyset$, so $C \nsubseteq A(k_1)$ and hence 
$\sum_{A \subseteq A(k_1)} \gamma_A < 1$.

Therefore, $0 < \sum_{A \subseteq A(k_1)} \gamma_A < 1$, and $P(Y \in A(k_1)\mid \mathbf x)$ is  a constant in $(0,1)$ (a.e.). 
 $\blacksquare$

\vskip 0.1in

\textbf{Proof of Theorem  \ref{th:linearJ3}.}  
By Assumption \ref{assn:indCS}, conditioning on an active menu preserves exchangeability. 
The exchangeability of the joint distribution of unobservables  conditional on $\mathbf{x}$ implies that for any $c_1, c_2$ and any permutation $(k_1,k_2,k_3)$ of $(1,2,3)$, functions 
\[
G_{2,x}(c_1,c_2)
:=
P(\varepsilon_{k_1}-\varepsilon_{k_2}\ge c_1,\,
  \varepsilon_{k_1}-\varepsilon_{k_3}\ge c_2|x),
\]
\[
G_{1,x}(c)
:=
P(\varepsilon_{k_1}-\varepsilon_{k_2}\ge c| x).
\]
do not depend on specific alternatives $k_1, k_2, k_3$. 

We have 3 equations: for $k_1$ it is  
\begin{multline} 
\label{eq:Pk1} P(Y=k_1|\mathbf{x})= \gamma_{\{1,2,3\}}G_{2,\mathbf{x}}( -x_{k_1}\beta_{k_1}+x_{k_2}\beta_{k_2},-x_{k_1}\beta_{k_1}+x_{k_3}\beta_{k_3})   \\+ \gamma_{\{k_1,k_2\}} G_{1,\mathbf{x}}( -x_{k_1}\beta_{k_1}+x_{k_2}\beta_{k_2}) 
+ \gamma_{\{k_1,k_3\}} G_{1,\mathbf{x}}( -x_{k_1}\beta_{k_1}+x_{k_3}\beta_{k_3})+\gamma_{\{k_1\}}, 
\end{multline} 
and analogous for $k_2$, $k_3$. These equations do not depend on  the order of indices $x_{k_1}\beta_{k_1}$, $x_{k_2}\beta_{k_2}$, $x_{k_3}\beta_{k_3}$. Since the three choice probabilities sum to one, the third equation for $k_3$ follows from the first two. We therefore retain only the equations for \(P(Y=k_1| x)\) and \(P(Y=k_2| x)\).

To characterize $\mathcal{I}^{B, \cup_{\boldsymbol{\beta}}}_{k_1,k_2}$, we have to look at the collection of choice probabilities obtained under $x_{k_1}\beta_{k_1}= x_{k_2}\beta_{k_2}$.    With this equality enforced,  
let us denote $a(\mathbf{x})=-x_{k_1}\beta_{k_1}+x_{k_3}\beta_{k_3}=-x_{k_2}\beta_{k_2}+x_{k_3}\beta_{k_3}$. For sufficiency, no restriction on the attainable range of this value is needed. 
At the collection of such indices we have the following relations: 
\begin{align}P(Y=k_1|\mathbf{x}) & = \gamma_{\{1,2,3\}}G_{2,\mathbf{x}}( 0,a(\mathbf{x}))  + 0.5 \gamma_{\{k_1,k_2\}} + \gamma_{\{k_1,k_3\}} G_{1,\mathbf{x}}( a(\mathbf{x}))+\gamma_{\{k_1\}}, \label{eq:Pk1unknown}
\\P(Y=k_2|\mathbf{x}) & = \gamma_{\{1,2,3\}}G_{2,\mathbf{x}}( 0,a(\mathbf{x}))  + 0.5 \gamma_{\{k_1,k_2\}} + \gamma_{\{k_2,k_3\}} G_{1,\mathbf{x}}( a(\mathbf{x}))+\gamma_{\{k_2\}}, 
\label{eq:Pk2unknown}
\\P(Y=k_3|\mathbf{x}) & = \gamma_{\{1,2,3\}}(1-2G_{2,\mathbf{x}}( 0,a(\mathbf{x})))  + \gamma_{\{k_1,k_3\}}(1-G_{1,\mathbf{x}}( a(\mathbf{x}))) \notag \\
& \quad + \gamma_{\{k_2,k_3\}} (1-G_{1,\mathbf{x}}( a(\mathbf{x}))) +\gamma_{\{k_3\}}. 
\label{eq:Pk3unknown}
\end{align}
Equation \eqref{eq:Pk3unknown} is a linear combination of the first two equations, therefore we can just focus on those two equations. 
On the right-hand sides of \eqref{eq:Pk1unknown}-\eqref{eq:Pk2unknown} we have only two unknowns: $G_{2,\mathbf{x}}(0,a(\mathbf{x}))=G_{2,\mathbf{x}}(a(\mathbf{x}),0)$ and $G_{1,\mathbf{x}}( a(\mathbf{x}))$. With respect to these unknowns  our system is linear and has $C_{k_1,k_2}$ as the matrix of coefficients corresponding to $(G_{2,\mathbf{x}}( 0,a(\mathbf{x})), G_{1,\mathbf{x}}( a(\mathbf{x})))^\top$. The $j$-th row in $C_{k_1,k_2}$ is obtained from the equation for $P(Y=k_j|\mathbf{x})$, $j=1,2$, in (\ref{eq:Pk1unknown})-(\ref{eq:Pk2unknown}).

\textbf{Sufficiency}. First, suppose that the rank of \(C_{k_1,k_2}\) is one. Then $d_1\bigl(\gamma_{\{1,2,3\}},\gamma_{\{k_1,k_3\}}\bigr)
=
d_2\bigl(\gamma_{\{1,2,3\}},\gamma_{\{k_2,k_3\}}\bigr)$ for some \(d_1,d_2\geq 0\), with at least one of them strictly positive. Equations
(\ref{eq:Pk1unknown})--(\ref{eq:Pk2unknown}) therefore imply that, whenever
\(x_{k_1}\beta_{k_1}=x_{k_2}\beta_{k_2}\),
\begin{equation}
\label{diffJ3}
d_1(
P(Y=k_1\mid\mathbf{x})
-\frac12\gamma_{\{k_1,k_2\}}
-\gamma_{\{k_1\}}
)
-
d_2(
P(Y=k_2\mid\mathbf{x})
-\frac12\gamma_{\{k_1,k_2\}}
-\gamma_{\{k_2\}}
) =0.
\end{equation}

We next show that, for probability vectors in
\(\operatorname{relint}(\Delta^{B,\cup\beta}_2)\), the corresponding equality becomes a strict
inequality whenever \(x_{k_1}\beta_{k_1}\neq x_{k_2}\beta_{k_2}\). In particular, if
\(x_{k_1}\beta_{k_1}>x_{k_2}\beta_{k_2}\), exchangeability implies that the difference on the left-hand side of \eqref{diffJ3}  has the weak sign of
\(x_{k_1}\beta_{k_1}-x_{k_2}\beta_{k_2}\). The only way the resulting difference can remain
zero under a strict index ordering is through saturation of some choice
probabilities. As shown in the cases $\gamma_{\{1,2,3\}}=0$ and $\gamma_{\{1,2,3\}}>0$ below, when \(J=3\) such saturation places the generated
probability vector on the relative boundary of \(\Delta^{B,\cup\beta}_2\). Hence, on 
\(\operatorname{relint}(\Delta^{B,\cup\beta}_2)\), $d_1(
P(Y=k_1\mid\mathbf{x})
-\frac12\gamma_{\{k_1,k_2\}}
-\gamma_{\{k_1\}}
)
>
d_2(
P(Y=k_2\mid\mathbf{x})
-\frac12\gamma_{\{k_1,k_2\}}
-\gamma_{\{k_2\}}
)$ iff $x_{k_1}\beta_{k_1}>x_{k_2}\beta_{k_2}$. 

If \(\gamma_{\{1,2,3\}}>0\), rank one implies \(d_1=d_2>0\) and $\gamma_{\{k_1,k_3\}}=\gamma_{\{k_2,k_3\}}$. So \(\mathcal{I}^{B,\cup_{\boldsymbol{\beta}}}_{k_1,k_2}\) is characterized by $P(Y=k_1|\mathbf{x})-\pi^*_{k_1}
=P(Y=k_2\mid\mathbf{x})-\pi^*_{k_2}$.

As noted above, if \(x_{k_1}\beta_{k_1}>x_{k_2}\beta_{k_2}\), then exchangeability guarantees  $(P(Y=k_1|\mathbf{x})-\pi^*_{k_1})-
(P(Y=k_2|\mathbf{x})-\pi^*_{k_2}) \geq 0$. 
If this difference were zero, then, since
\(\gamma_{\{1,2,3\}}>0\), the probabilities of choosing \(k_1\) and \(k_2\) from the menu \(\{1,2,3\}\) would have to be equal. Suppose \(k_2\) wins with positive probability. Swapping its error with \(k_1\)’s produces an equally probable event on which \(k_1\) wins. Since  \(x_{k_1}\beta_{k_1}>x_{k_2}\beta_{k_2}\), both inequalities needed for \(k_1\) to win are strictly weaker. Convexity of the error support ensures that this relaxation adds positive probability as along a segment joining an interior point of the swapped event to its original point, \(k_1\) still wins just after leaving the swapped event. Thus, \(k_1\) wins strictly more often than \(k_2\). Their winning probabilities can therefore be equal only if both are zero.
$k_3$ then also beats \(k_1\) and \(k_2\) with probability one in the corresponding pair menus, implying that \(P(Y=k_3|\mathbf{x})\) is at its upper feasible bound. This upper bound is not constant on the feasible set as at total indifference, \(\pi^*_{k_3}\) is strictly below it because \(\gamma_{\{1,2,3\}}>0\). The generated probability vector therefore lies on the relative boundary of \(\Delta^{B,\cup_{\boldsymbol{\beta}}}_2\). Consequently, on
\(\relint(\Delta^{B,\cup_{\boldsymbol{\beta}}}_2)\),
$x_{k_1}\beta_{k_1}>x_{k_2}\beta_{k_2}$
iff 
$P(Y=k_1\mid\mathbf{x})-\pi^*_{k_1}
>
P(Y=k_2\mid\mathbf{x})-\pi^*_{k_2}$.

If \(\gamma_{\{1,2,3\}}=0\), at least one of
\(\gamma_{\{k_1,k_3\}}\) and \(\gamma_{\{k_2,k_3\}}\) must be strictly positive. First suppose that both are strictly positive. Then $\frac{d_1}{d_2}
=
\frac{\gamma_{\{k_2,k_3\}}}{\gamma_{\{k_1,k_3\}}}$, and we may take $d_1=\gamma_{\{k_2,k_3\}}$, 
$d_2=\gamma_{\{k_1,k_3\}}$. The indifference condition can then be written as $\gamma_{\{k_2,k_3\}}
\bigl(P(Y=k_1|\mathbf{x})-\pi^*_{k_1}\bigr)
=
\gamma_{\{k_1,k_3\}}
\bigl(P(Y=k_2\mid\mathbf{x})-\pi^*_{k_2}\bigr)$.
Let  \(x_{k_1}\beta_{k_1}>x_{k_2}\beta_{k_2}\). If \(\gamma_{\{k_1,k_2\}}>0\), then $\gamma_{\{k_2,k_3\}}
\bigl(P(Y=k_1|\mathbf{x})-\pi^*_{k_1}\bigr)
-
\gamma_{\{k_1,k_3\}}
\bigl(P(Y=k_2|\mathbf{x})-\pi^*_{k_2}\bigr)>0$
as the binary menu \(\{k_1,k_2\}\) strictly favors \(k_1\). If \(\gamma_{\{k_1,k_2\}}=0\), equality  $\gamma_{\{k_2,k_3\}}
\bigl(P(Y=k_1|\mathbf{x})-\pi^*_{k_1}\bigr)
-
\gamma_{\{k_1,k_3\}}
\bigl(P(Y=k_2|\mathbf{x})-\pi^*_{k_2}\bigr)=0$ can occur only if the probabilities of choosing \(k_1\) from \(\{k_1,k_3\}\) and \(k_2\) from \(\{k_2,k_3\}\) are both zero or both one. In the first case \(P(Y=k_3|\mathbf{x})\) is at its upper feasible bound, and in the second it is at its lower feasible bound. Thus equality under a strict index ordering is again impossible in
\(\relint(\Delta^{B,\cup_{\boldsymbol{\beta}}}_2)\), and, hence, $x_{k_1}\beta_{k_1}>x_{k_2}\beta_{k_2}$
iff 
$\gamma_{\{k_2,k_3\}}
\bigl(P(Y=k_1|\mathbf{x})-\pi^*_{k_1}\bigr)
>
\gamma_{\{k_1,k_3\}}
\bigl(P(Y=k_2|\mathbf{x})-\pi^*_{k_2}\bigr)$. Finally, suppose that exactly one of
\(\gamma_{\{k_1,k_3\}}\) and \(\gamma_{\{k_2,k_3\}}\) is zero. For example, if
\(\gamma_{\{k_1,k_3\}}=0\) and
\(\gamma_{\{k_2,k_3\}}>0\), then \(d_2=0\), while connectedness of \(k_1\) and \(k_2\) implies
\(\gamma_{\{k_1,k_2\}}>0\). The indifference set is therefore characterized by $P(Y=k_1|\mathbf{x})-\pi^*_{k_1}=\gamma_{\{k_1,k_2\}}(P(Y=k_1|\{k_1,k_2\},\mathbf{x})-1/2) $. By exchangeability and absolute continuity, the binary comparison in the active menu \(\{k_1,k_2\}\) gives $x_{k_1}\beta_{k_1}>x_{k_2}\beta_{k_2}$ $\Longleftrightarrow$ 
$P(Y=k_1|\mathbf{x})-\pi^*_{k_1}>0$.

The case
\(\gamma_{\{k_1,k_3\}}>0\) and
\(\gamma_{\{k_2,k_3\}}=0\) is symmetric.

Thus, when \(\operatorname{rank}(C_{k_1,k_2})=1\), the indifference set and the two strict ranking regions are characterized by the hyperplane and its two strict half-spaces stated in the theorem. The dimension statements follow under the full range condition on the two index differences stated in the main text.

Now suppose the rank of $C_{k_1,k_2}$ is 0. This means that $\gamma_{\{1,2,3\}} = \gamma_{\{k_1,k_3\}}=\gamma_{\{k_2,k_3\}}=0$, and, hence, $k_3$ is not connected to either $k_1$ or $k_2$. At the same time, $\gamma_{\{k_1,k_2\}}>0$ as the theorem states that $k_1$ and $k_2$ are connected. Again, 
\[
P(Y=k_j|\mathbf{x})-\pi^*_{k_j}
=
\gamma_{\{k_1,k_2\}}
\left(
P(Y=k_j|\{k_1,k_2\},\mathbf{x})-1/2
\right), \quad j=1,2.
\]
By exchangeability and absolute continuity, again, 
$P(Y=k_1|\{k_1,k_2\},\mathbf{x})>\frac12$
$\Longleftrightarrow$ 
$x_{k_1}\beta_{k_1}>x_{k_2}\beta_{k_2}$,
with equality to \(1/2\) if and only if
\(x_{k_1}\beta_{k_1}=x_{k_2}\beta_{k_2}\). Hence, 
\[
x_{k_1}\beta_{k_1}=x_{k_2}\beta_{k_2}
\quad\Longleftrightarrow\quad
P(Y=k_1|\mathbf{x})-\pi^*_{k_1}
=
P(Y=k_2|\mathbf{x})-\pi^*_{k_2}
=0,
\]
\[
x_{k_1}\beta_{k_1}>x_{k_2}\beta_{k_2}
\quad\Longleftrightarrow\quad
P(Y=k_1|\mathbf{x})-\pi^*_{k_1}>0
\quad\Longleftrightarrow\quad
P(Y=k_2|\mathbf{x})-\pi^*_{k_2}<0.
\]
The reverse inequalities characterize the opposite strict ranking.

\textbf{Necessity}. Let $\mathcal{I}^{B, \cup_{\boldsymbol{\beta}}}_{k_1,k_2}$  be linear in $\mathbf{p}$. Contrary to the statement of the proposition, suppose the rank of $C_{k_1,k_2}$ is 2. This implies that $\gamma_{\{1,2,3\}}>0$ and $\gamma_{\{k_1,k_3\}} \neq \gamma_{\{k_2,k_3\}}$.  Linear $\mathcal{I}^{B, \cup_{\boldsymbol{\beta}}}_{k_1,k_2}$  has to pass through $(\pi_{1}^*,\pi_{2}^*, \pi_{3}^*)$ point obtained for the case of the overall indifference. Taking into account that the sum of all  choice probabilities is 1,  we can thus  characterize this linear $\mathcal{I}^{B, \cup_{\boldsymbol{\beta}}}_{k_1,k_2}$  as the set of vectors $(p_{k_1}, p_{k_2}, 1-p_{k_1}-p_{k_2}) \in \Delta^{B, \cup_{\boldsymbol{\beta}}}_{2}$ that satisfy 
$d_1(p_{k_1}-\pi_{k_1}^*)-d_2 (p_{k_2}-\pi_{k_2}^*)=0,$ 
where at least one $d_j\neq0$, When $p_{k_1}$ and $p_{k_2}$ are replaced by $P(Y=k_1|\mathbf{x})$  and $P(Y=k_2|\mathbf{x})$, respectively, this relation has an implication with the first two-equations (\ref{eq:Pk1unknown})-(\ref{eq:Pk2unknown}) giving us 
\begin{equation} 
\label{eq:implication}
    (d_1 -d_2)\gamma_{\{1,2,3\}}(G_{2, \mathbf{x}}(0, a(\mathbf{x}))-1/3)+ (d_1 \gamma_{\{k_1,k_3\}}-d_2 \gamma_{\{k_2,k_3\}}) (G_{1, \mathbf{x}}(a(\mathbf{x}))-1/2) =0.
\end{equation}
Since a \emph{Behavioral environment}  separator must be valid for every
admissible conditional distribution of the unobservables, it suffices to restrict attention to one admissible
$\mathbf{x}$-invariant exchangeable distribution. We do so throughout
this necessity argument.

If $d_1=d_2$, then the fact that $d_1,d_2\neq 0$ and  $\gamma_{\{k_1,k_3\}} \neq \gamma_{\{k_2,k_3\}} $  implies that $d_1 \gamma_{\{k_1,k_3\}}-d_2 \gamma_{\{k_2,k_3\}} \neq 0$. Equation (\ref{eq:implication}) is possible only if $G_{1, \mathbf{x}}(a(\mathbf{x}))-1/2=0$ which gives us a contradiction for $a(\mathbf{x}) \neq 0$. Thus $d_1 \neq d_2$. 

If $d_1\gamma_{\{k_1,k_3\}}-d_2\gamma_{\{k_2,k_3\}} = 0$, then  $\gamma_{\{k_1,k_3\}} \neq \gamma_{\{k_2,k_3\}} $  implies that $d_1-d_2 \neq 0$. Equation (\ref{eq:implication}) is then possible only if $G_{2,\mathbf{x}}(0, a(\mathbf x))-1/3=0$ which gives us a contradiction for $a(x) \neq 0$. Thus, $d_1\gamma_{\{k_1,k_3\}}-d_2\gamma_{\{k_2,k_3\}} \neq 0.$ 

Having $d_1\gamma_{\{k_1,k_3\}}-d_2\gamma_{\{k_2,k_3\}} \neq 0$ and $d_1 \neq d_2$, we obtain that for any $a(\mathbf{x}) \neq 0$ we can write 
$$ \frac{G_{2, \mathbf{x}}(0,a(\mathbf{x}))-1/3}{G_{1, \mathbf{x}}(a(\mathbf{x}))-1/2} = \frac{d_2 \gamma_{\{k_2,k_3\}}-d_1 \gamma_{\{k_1,k_3\}}}{(d_1 -d_2)\gamma_{\{1,2,3\}}} \quad \iff$$
$$\frac{ G_{2, \mathbf{x}}(0,a(\mathbf{x}))-G_{2, \mathbf{x}}(0,0)}
{G_{1, \mathbf{x}}(a(\mathbf{x}))-G_{1, \mathbf{x}}(0)} = \frac{d_2 \gamma_{\{k_2,k_3\}}-d_1 \gamma_{\{k_1,k_3\}}}{(d_1 -d_2)\gamma_{\{1,2,3\}}}.$$

Under the additional index variation condition used for the converse and mentioned before the theorem in the main text, $a(\mathbf{x})$  can be taken arbitrarily far in either direction. If  $a(\mathbf{x})$ is taken arbitrarily close to the lower support point of the distribution of $\varepsilon_{k_1}-\varepsilon_{k_2}|\mathbf{x}$, it gives us 
$ \frac{d_2 \gamma_{\{k_2,k_3\}}-d_1 \gamma_{\{k_1,k_3\}}}{(d_1 -d_2)\gamma_{\{1,2,3\}}}=\frac{1/2-1/3}{1-1/2}=\frac{1}{3}.$ 
If $a(\mathbf{x})$ is taken arbitrarily close to the upper support point of the distribution of $\varepsilon_{k_1}-\varepsilon_{k_2}|\mathbf{x}$, it gives us 
$\frac{d_2 \gamma_{\{k_2,k_3\}}-d_1 \gamma_{\{k_1,k_3\}}}{(d_1 -d_2)\gamma_{\{1,2,3\}}} =\frac{0-1/3}{0-1/2} = \frac{2}{3}$. This is clearly a contradiction.  
$\blacksquare$

\vskip 0.05in 

\textbf{Proof of Corollary  \ref{cor:linearJ3allpairs}.}  
If $\gamma_{\{1,2,3\}}>0$, then rank one for every pair implies 
$\gamma_{\{1,3\}}=\gamma_{\{2,3\}}=
\gamma_{\{1,2\}}$.  The three separator equations in part \textbf{B}
of Theorem~\ref{th:linearJ3} therefore reduce, up to the same strictly
positive multiplicative constant, to 
 $(p_1-\pi_1^*)-(p_2-\pi_2^*)=0$, 
 $(p_1-\pi_1^*)-(p_3-\pi_3^*)=0$, 
$ (p_2-\pi_2^*)-(p_3-\pi_3^*)=0$.
They are three distinct
lines through $\boldsymbol\pi^*$. If $\gamma_{\{1,2,3\}}=0$, then rank one for every pair implies
that no two of
$\gamma_{\{1,2\}}$, $\gamma_{\{1,3\}}$, and
$\gamma_{\{2,3\}}$ can be zero.  The three separator equations are then $\gamma_{\{2,3\}}(p_1-\pi_1^*)-\gamma_{\{1,3\}}(p_2-\pi_2^*)=0$, 
$ \gamma_{\{2,3\}}(p_1-\pi_1^*)-\gamma_{\{1,2\}}(p_3-\pi_3^*)=0$, $\gamma_{\{1,3\}}(p_2-\pi_2^*)-\gamma_{\{1,2\}}(p_3-\pi_3^*)=0$.
Since at least two of the three coefficients are strictly positive, these
are again three distinct lines through $\boldsymbol\pi^*$.

Three distinct concurrent lines divide the two-dimensional affine space into six open regions. Their signs correspond to consistent rankings. If \(\gamma_{\{1,2,3\}}>0\), they compare the three centered probabilities \(p_j-\pi_j^*\). If \(\gamma_{\{1,2,3\}}=0\) and all three pair-menu probabilities are positive, they compare
$\frac{p_1-\pi_1^*}{\gamma_{\{1,2\}}\gamma_{\{1,3\}}}$, $\frac{p_2-\pi_2^*}{\gamma_{\{1,2\}}\gamma_{\{2,3\}}}$, $\frac{p_3-\pi_3^*}{\gamma_{\{1,3\}}\gamma_{\{2,3\}}}$.
If exactly one pair-menu probability is zero, say \(\gamma_{\{1,3\}}=0\), they instead compare
$\frac{p_1-\pi_1^*}{\gamma_{\{1,2\}}}$,\qquad
$0$, 
$\frac{p_3-\pi_3^*}{\gamma_{\{2,3\}}}$. Thus, in every case, the six regions correspond to the six strict orderings of three real numbers.

 Part \textbf{B} of Theorem \ref{th:linearJ3}  implies that an attainable probability vector in \(\relint(\Delta^{B,\cup\beta}_2)\) lies in the cell corresponding to \(k_1\succ k_2\succ k_3\) if and only if every index configuration generating it satisfies \(x_{k_1}\beta_{k_1}>x_{k_2}\beta_{k_2}>x_{k_3}\beta_{k_3}\). Thus, within this relative interior, each strict ranking region consists exactly of the attained probability vectors in its corresponding cell. These regions are pairwise disjoint because different strict rankings disagree on at least one pairwise comparison. Attained vectors in this relative interior generated by an index tie lie on the corresponding indifference line. $\blacksquare$

\vskip 0.05in

\noindent\textbf{Proof of Theorem \ref{th:linear_generalJ}.}
For every nonempty consideration set $A$ and
$h\in A$, let 
 $q_h^A(\mathbf x)
 :=P(Y=h| A,\mathbf x)
 $. Then 
$P(Y=h|\mathbf x)=\sum_{A:h\in A}\gamma_Aq_h^A(\mathbf x)$. 

Let $O=\mathcal J\setminus\{k_1,k_2\}$ and impose
$x_{k_1} \beta_{k_1}=x_{k_2} \beta_{k_2}$.  For each nonempty $S\subseteq O$, exchangeability under
the transposition of $k_1$ and $k_2$ gives $g_S^C(\mathbf x)
 =q_{k_1}^{S\cup\{k_1,k_2\}}(\mathbf x)
   =q_{k_2}^{S\cup\{k_1,k_2\}}(\mathbf x)$, $g_S^E(\mathbf x)
 =q_{k_1}^{S\cup\{k_1\}}(\mathbf x)
   =q_{k_2}^{S\cup\{k_2\}}(\mathbf x)$, and also  
$q_{k_1}^{\{k_1,k_2\}}=q_{k_2}^{\{k_1,k_2\}}=1/2$.  Consequently, the
individual choice probabilities on the pairwise-indifference locus are
\begin{equation}
 p_{k_j}
 =\gamma_{\{k_j\}}+\frac{\gamma_{\{k_1,k_2\}}}{2}
 +\sum_{\varnothing\neq S\subseteq O}
 \left[
 \gamma_{S\cup\{k_1,k_2\}}g_S^C(\mathbf x)
 +\gamma_{S\cup\{k_j\}}g_S^E(\mathbf x)
 \right], \quad j=1,2.
 \label{eq:pk1-indiff-generalJ}
\end{equation}
Considering deviations from values   at total indifference, we rewrite equation in \eqref{eq:pk1-indiff-generalJ} as 
\begin{equation}
 \begin{pmatrix}
 p_{k_1}-\pi_{k_1}^*\\[1mm]
 p_{k_2}-\pi_{k_2}^*
 \end{pmatrix}
 =C^{(J)}_{k_1,k_2}
 \begin{pmatrix}
 (g_S^C-1/(|S|+2))_{\varnothing\neq S\subseteq O}\\[1mm]
 (g_S^E-1/(|S|+1))_{\varnothing\neq S\subseteq O}
 \end{pmatrix}.
 \label{eq:centered-system-generalJ}
\end{equation} 

\textbf{A}.  We start by proving equivalence of the rank condition to \eqref{eq:primitive-rank-generalJ}. Write the two types of columns of
$C^{(J)}_{k_1,k_2}$ as $ c_S=
 \gamma_{S\cup\{k_1,k_2\}}\binom{1}{1}$, $e_S=
 \binom{\gamma_{S\cup\{k_1\}}}
       {\gamma_{S\cup\{k_2\}}}$, where $\varnothing\neq S\subseteq O$.
All minors formed from two common-menu columns $c_S$, $c_T$ vanish
identically.  The minor formed from $e_S$, $e_T$ is $ \gamma_{S\cup\{k_1\}}\gamma_{T\cup\{k_2\}}
 -\gamma_{S\cup\{k_2\}}\gamma_{T\cup\{k_1\}},$ and the minor formed from $c_S$, $e_T$ is $ \gamma_{S\cup\{k_1,k_2\}}
 \left(\gamma_{T\cup\{k_2\}}-
       \gamma_{T\cup\{k_1\}}\right)$.
Thus all two-by-two minors vanish iff 
\eqref{eq:primitive-rank-generalJ} holds.

If $\operatorname{rank}\!\left(C^{(J)}_{k_1,k_2}\right)<2$, there exists a nonzero left-null vector \((a_1,a_2)\). Premultiplying
\eqref{eq:centered-system-generalJ} by this vector yields a linear equality involving only
\(p_{k_1}\) and \(p_{k_2}\) for every probability vector generated under
pairwise indifference. This proves sufficiency.

For necessity, total indifference generates $\boldsymbol\pi^*$, so any such
hyperplane can be written as
\begin{equation}
a_1(p_{k_1}-\pi_{k_1}^*)
+
a_2(p_{k_2}-\pi_{k_2}^*)
=0,
\label{eq:putative-two-coordinate-generalJ}
\end{equation}
for some $(a_1,a_2)\neq(0,0)$, for every probability vector generated
under $x_{k_1}\beta_{k_1}=x_{k_2}\beta_{k_2}$ in the
\textit{Behavioral environment}. By \eqref{eq:centered-system-generalJ}, this is
$(a_1,a_2)C^{(J)}_{k_1,k_2}u=0$ for every admissible vector of menu-level
deviations $u$.  Write $ \alpha_S=(a_1+a_2)\gamma_{S\cup\{k_1,k_2\}}$, $\eta_S=a_1\gamma_{S\cup\{k_1\}}+
         a_2\gamma_{S\cup\{k_2\}}$. Then \eqref{eq:putative-two-coordinate-generalJ} implies $\sum_{\varnothing\neq S\subseteq O}\alpha_Su_S^C
 +\sum_{\varnothing\neq S\subseteq O}\eta_Su_S^E=0$ for every attainable tied index vector. Fix $r\in O$ and $B\subseteq O\setminus\{r\}$. By the index variation condition in Part \textbf{A}, the $J-1$ index
differences range over all of $\mathbb R^{J-1}$. Hence, for every $a\in\mathbb{R}$ there exists an
attainable sequence of tied index configurations that, after a common
normalization, satisfies $x_{k_1}\beta_{k_1}=x_{k_2}\beta_{k_2}=0$,
 $x_r\beta_r=a$,  
with $x_s\beta_s\to+\infty$ for $s\in B$ and $x_s\beta_s\to-\infty$ for $s\in O\setminus(B\cup\{r\})$. 
Since a \emph{Behavioral environment}  separator must be valid for every
admissible conditional distribution of the unobservables, for the
converse it suffices to restrict attention to one admissible
$\mathbf{x}$-invariant exchangeable distribution. We do so throughout
this necessity argument.  Menus containing an element of $B$ make the probability of choosing $k_1$ converge to zero, whereas alternatives whose indices tend to $-\infty$ become irrelevant.  By dominated convergence, from the above we get 
\begin{equation}
 \underbrace{\left(\sum_{\substack{S\ni r\\ S\cap B=\varnothing}}\alpha_S\right)}_{:=A_B} G_C(a)+ \underbrace{\left(\sum_{\substack{S\ni r\\ S\cap B=\varnothing}}\eta_S\right)}_{:=E_B} G_E(a)+K_B=0 \qquad\text{for every }a\in\mathbb R,
 \label{eq:one-coordinate-completeness-generalJ}
\end{equation}
where $K_B$ does not depend on $a$,  and $G_C(a):=P\!\left(\varepsilon_{k_1}\geq\varepsilon_{k_2},\ 
             \varepsilon_{k_1}-\varepsilon_r\geq a
             \mid\mathbf x\right)$, 
$ G_E(a):=P\!\left(\varepsilon_{k_1}-\varepsilon_r\geq a
             \mid\mathbf x\right).$ Exchangeability and absolute continuity imply
 $G_C(0)=\frac13,$ $G_E(0)=\frac12$. Then 
$ \lim_{a\to-\infty}(G_C(a),G_E(a))=\left(\frac12,1\right)$, and 
 $\lim_{a\to+\infty}(G_C(a),G_E(a))=(0,0)$. 
Evaluating \eqref{eq:one-coordinate-completeness-generalJ} at these three
limits gives $K_B=0$, $A_B/3+E_B/2=0$, and $A_B/2+E_B=0$.
Hence $A_B=E_B=0$ for every $B\subseteq O\setminus\{r\}$.

As $B$ varies, these are all cumulative subset sums of the coefficients
$\{\alpha_S:S\ni r\}$ and $\{\eta_S:S\ni r\}$.  M\"obius inversion on
the subset lattice therefore gives $\alpha_S=\eta_S=0$ for every $S$
containing $r$.  Repeating the argument for each $r\in O$ yields
$\alpha_S=\eta_S=0$ for every nonempty $S$, which is precisely
$(a_1,a_2)C^{(J)}_{k_1,k_2}=0$.  A nonzero left-null vector exists only when 
$\mathrm{rank}(C^{(J)}_{k_1,k_2})<2$.  

\textbf{B.} Define
$\Gamma_1
:=
\sum_{\emptyset\neq S\subseteq O}
\left(
\gamma_{S\cup\{k_1,k_2\}}
+
\gamma_{S\cup\{k_1\}}
\right)$,
and $\Gamma_2
:=
\sum_{\emptyset\neq S\subseteq O}
\left(
\gamma_{S\cup\{k_1,k_2\}}
+
\gamma_{S\cup\{k_2\}}
\right)$. Since \(C^{(J)}_{k_1,k_2}\) has rank one, $\Gamma_2\,\mathrm{row}_1(C^{(J)}_{k_1,k_2})
=
\Gamma_1\,\mathrm{row}_2(C^{(J)}_{k_1,k_2})$. Thus, \((\Gamma_2,-\Gamma_1)\) is a left-null vector, and \eqref{eq:centered-system-generalJ} gives $H_{k_1,k_2}(p)=0$ whenever
$x_{k_1}\beta_{k_1}=x_{k_2}\beta_{k_2}$.

It is still needed to orient the two sides of the hyperplane. Suppose first that $\gamma_{S\cup\{k_1,k_2\}}>0$ for at least one nonempty \(S\subseteq O\). A positive common menu column and
rank one imply that the two rows of \(C^{(J)}_{k_1,k_2}\) are identical.
Therefore, $\Gamma_1=\Gamma_2=:\Gamma>0$, $\gamma_{S\cup\{k_1\}}
=\gamma_{S\cup\{k_2\}}$ for every nonempty $S\subseteq O$. A direct expansion gives
\begin{multline}
 \label{eq:H-common-menu-generalJ}
\frac{H_{k_1,k_2}(p)}{\Gamma}
=
\gamma_{\{k_1,k_2\}}
\left(
q_{k_1}^{\{k_1,k_2\}}
-
q_{k_2}^{\{k_1,k_2\}}
\right)
+
\sum_{\emptyset\neq S\subseteq O}
\gamma_{S\cup\{k_1,k_2\}}
\left(
q_{k_1}^{S\cup\{k_1,k_2\}}
-
q_{k_2}^{S\cup\{k_1,k_2\}}
\right)
\\
+
\sum_{\emptyset\neq S\subseteq O}
\gamma_{S\cup\{k_1\}}
\left(
q_{k_1}^{S\cup\{k_1\}}
-
q_{k_2}^{S\cup\{k_2\}}
\right).
\end{multline}
Suppose next that 
$\gamma_{S\cup\{k_1,k_2\}}=0$
for every nonempty $S\subseteq O$. Rank one implies proportionality of the two exclusive menu rows. Define
$w_S
:=
\Gamma_2\gamma_{S\cup\{k_1\}}
=
\Gamma_1\gamma_{S\cup\{k_2\}}
\geq0$. Then
\begin{equation} 
\label{eq:H-exclusive-menu-generalJ} 
H_{k_1,k_2}(p)
=
\gamma_{\{k_1,k_2\}}
\frac{\Gamma_1+\Gamma_2}{2}
\left(
q_{k_1}^{\{k_1,k_2\}}
-
q_{k_2}^{\{k_1,k_2\}}
\right)
+
\sum_{\emptyset\neq S\subseteq O}
w_S
\left(
q_{k_1}^{S\cup\{k_1\}}
-
q_{k_2}^{S\cup\{k_2\}}
\right).
\end{equation}

Part B(i). If \(x_{k_1}\beta_{k_1}>x_{k_2}\beta_{k_2}\), exchangeability implies
$q_{k_1}^{\{k_1,k_2\}}
>
q_{k_2}^{\{k_1,k_2\}}$, and, for every nonempty \(S\subseteq O\),
$q_{k_1}^{S\cup\{k_1,k_2\}}
\geq
q_{k_2}^{S\cup\{k_1,k_2\}}$, and $q_{k_1}^{S\cup\{k_1\}}
\geq
q_{k_2}^{S\cup\{k_2\}}$. 
The first inequality is strict because
the support of \(\varepsilon_{k_2}-\varepsilon_{k_1}\) is a symmetric
interval containing zero in its interior.  The second follows by
interchanging the \(k_1\)- and \(k_2\)-coordinates of the unobservables, and
the third because the two probabilities are evaluations of the same
nondecreasing choice probability function at \(x_{k_1}\beta_{k_1}\) and \(x_{k_2}\beta_{k_2}\). Since \(\gamma_{\{k_1,k_2\}}>0\), the direct $\{k_1,k_2\}$ menu term in
\eqref{eq:H-common-menu-generalJ} or
\eqref{eq:H-exclusive-menu-generalJ} is therefore strictly positive, while
all remaining terms are nonnegative. Therefore, $x_{k_1}\beta_{k_1}>(<)x_{k_2}\beta_{k_2}
\; \Rightarrow\; 
H_{k_1,k_2}(\mathbf p(\mathbf x))>(<)0$. Together with $x_{k_1}\beta_{k_1}=x_{k_2}\beta_{k_2}
\; \; \Longrightarrow\; \;
H_{k_1,k_2}(\mathbf p(\mathbf x))=0$,
established above, this proves (Bi).

\vskip 0.03in 

Part (Bii). Under \(\gamma_{\{k_1,k_2\}}=0\), the same weak inequalities and the
nonnegative coefficients in \eqref{eq:H-common-menu-generalJ} and
\eqref{eq:H-exclusive-menu-generalJ} give  
$x_{k_1}\beta_{k_1}> (<) x_{k_2}\beta_{k_2} \Rightarrow H_{k_1,k_2}(\mathbf{p}(\mathbf{x}))\geq (\leq) 0$.  Together with \(x_{k_1}\beta_{k_1}=x_{k_2}\beta_{k_2} \Rightarrow H_{k_1,k_2}(\mathbf{p}(\mathbf{x}))=0\), these
implications give $H_{k_1,k_2}(\mathbf p(\mathbf x))>(<)0 \Rightarrow x_{k_1}\beta_{k_1}>(<) x_{k_2}\beta_{k_2}$. 
The inequalities obtained above are only weak. Hence no ranking conclusion
follows from $H_{k_1,k_2}(\mathbf p(\mathbf x))=0$ without additional
support restrictions.

\vskip 0.03in 

Part (Biii). Under full support, the two weak menu level inequalities used above are
strict whenever \(x_{k_1}\beta_{k_1}> (<) x_{k_2}\beta_{k_2}\) as in each case the strict inclusion of
the corresponding choice events contains a nonempty open set of unobservables
realizations, which has positive probability under support
\(\mathbb R^J\).  Thus every positively weighted term in
\eqref{eq:H-common-menu-generalJ} or
\eqref{eq:H-exclusive-menu-generalJ} has the strict sign of
\(x_{k_1}\beta_{k_1}- x_{k_2}\beta_{k_2}\).  If \(\gamma_{\{k_1,k_2\}}>0\), the result already
follows from (Bi).  If \(\gamma_{\{k_1,k_2\}}=0\), rank one and connectedness
ensure that at least one common menu or matched exclusive menu term has
positive weight. Hence,  $x_{k_1}\beta_{k_1}> (<) x_{k_2}\beta_{k_2} \Leftrightarrow H_{k_1,k_2}(\mathbf p(\mathbf x) )>(<)0$, and, by  the equality result above,
$x_{k_1}\beta_{k_1}= x_{k_2}\beta_{k_2} \Leftrightarrow H_{k_1,k_2}(\mathbf{p}(\mathbf{x}))=0$.

\vskip 0.03in 

\textbf{C}. Suppose $\mathrm{rank}(C^{(J)}_{k_1,k_2})=0$.  Then every menu containing
one member of the pair and at least one outside alternative has zero
probability, including every menu containing both members and an outside
alternative.  Since $k_1$ and $k_2$ are connected, then necessarily $\gamma_{\{k_1,k_2\}}>0$ and $P(Y={k_1}|\mathbf x)-\pi_{k_1}^*
 =\gamma_{\{k_1,k_2\}}
   \left(q_{k_1}^{\{k_1,k_2\}}(\mathbf x)-\frac12\right)$, 
 $P(Y={k_2}|\mathbf x)-\pi_{k_2}^*
=\gamma_{\{k_1,k_2\}}
   \left(q_{k_2}^{\{k_1,k_2\}}(\mathbf x)-\frac12\right)$. By exchangeability and Assumption~\ref{assn:multi_distributionMED}, these
centered probabilities are both zero exactly when
$x_{k_1}\beta_{k_1}=x_{k_2}\beta_{k_2}$ and have the signs stated in
Part ~\textbf{C} otherwise. $\blacksquare$

\smallskip 

\noindent\textbf{Proof of Corollary~\ref{cor:linear-allpairs-generalJ}.}
By Theorem~\ref{th:linear_generalJ} B, all pairwise hyperplanes
pass through $\boldsymbol\pi^*$.

Suppose first that both rows of \(C^{(J)}_{k,\ell}\) are nonzero for every pair \(k\ne\ell\). We show that every pair belongs to some active consideration set. Otherwise, choose a pair \(k,\ell\) that never appears together. Since its rows are nonzero, there is an active menu \(S\cup\{k\}\), with \(S\ne\varnothing\) and \(k,\ell\notin S\). Row proportionality for \(C^{(J)}_{k,\ell}\) implies that \(S\cup\{\ell\}\) is also active. Choose \(h\in S\). Row proportionality for \(C^{(J)}_{k,h}\), applied to this latter menu, then implies that \((S\setminus\{h\})\cup\{k,\ell\}\) is active, which is a contradiction.

Fix alternative $1$, set $r_1=1$, and, for $k\neq1$, let $r_k>0$ be
defined so that the first row of $C^{(J)}_{k,1}$ is $r_k$ times its
second row. For
distinct $k,\ell \neq 1$, choose $h\notin\{1,k,\ell\}$. Summing the corresponding
row proportionality relations over menus containing $h$ gives
\[
\sum_{A\supseteq\{k,h\}}\gamma_A
=
r_k\sum_{A\supseteq\{1,h\}}\gamma_A,
\qquad
\sum_{A\supseteq\{\ell,h\}}\gamma_A
=
r_\ell\sum_{A\supseteq\{1,h\}}\gamma_A.
\]
These sums are strictly positive because every pair belongs to an active menu. Their ratio is \(r_k/r_\ell\). Applying the same summation to the proportional rows of \(C^{(J)}_{k,\ell}\) shows that this ratio is also their proportionality factor. By \eqref{eq:Hrank1J4}, $H_{k,\ell}(\mathbf p)=0$
$\iff$
$\frac{p_k-\pi_k^*}{r_k}
=
\frac{p_\ell-\pi_\ell^*}{r_\ell}$. Moreover, since the coefficient on $p_k-\pi_k^*$ in
\eqref{eq:Hrank1J4} is strictly positive, $H_{k,\ell}(\mathbf p)
=
c_{k,\ell} r_k
\left(
\frac{p_k-\pi_k^*}{r_k}
-
\frac{p_\ell-\pi_\ell^*}{r_\ell}
\right)$ for some $c_{k,\ell}>0$, which implies that $H_{k,\ell}(\mathbf p)>0$
 $\iff$
$\frac{p_k-\pi_k^*}{r_k}
>
\frac{p_\ell-\pi_\ell^*}{r_\ell}$.

Therefore, after a positive rescaling of the centered probability coordinates, the pairwise hyperplanes become the braid arrangement, that is, the collection of hyperplanes equating pairs of coordinates. Its regions
correspond to the $J!$ strict orderings of the rescaled coordinates (see, e.g., \citet{Stanley2007}). Since all $r_k$ are positive, a common shift of the rescaled coordinates preserves their strict ordering and can be chosen so that the corresponding centered probabilities sum to zero. Hence every braid chamber intersects the probability affine hull, and there are exactly $J!$ regions.

If instead one row of some $C^{(J)}_{k,\ell}$ is zero, suppose without
loss of generality that it is the row corresponding to $k$. Then $k$
belongs to no active non-singleton consideration set other than possibly
$\{k,\ell\}$. Connectedness therefore implies $\gamma_{\{k,\ell\}}>0$. Now take any $h\notin\{k,\ell\}$. In $C^{(J)}_{k,h}$, the menu $\{k,\ell\}$ gives a positive entry in the row corresponding to $k$. The row corresponding to $h$ cannot be zero as otherwise $h$ could belong
to no active non-singleton consideration set except possibly
$\{k,h\}$, while $\gamma_{\{k,h\}}=0$ because $k$ belongs to no active
non-singleton consideration set other than $\{k,\ell\}$. This would
contradict connectedness. Therefore,  both rows of $C^{(J)}_{k,h}$ are nonzero, and rank one implies that they are positive multiples of one another. In particular,  $\gamma_{\{\ell,h\}}>0$. The row corresponding to \(k\) has only one positive entry, associated with the menu \(\{k,\ell\}\). Proportionality therefore permits only the matching entry in the row corresponding to \(h\), associated with \(\{h,\ell\}\). Thus no other active non-singleton menu can contain \(h\). Hence, the only active non-singleton consideration sets are $\{\ell,h\}$, $ h\neq\ell$,  where $\ell$ is the common alternative, and every other alternative appears only together with $\ell$. 

For this consideration set structure, \eqref{eq:Hrank1J4} simplifies
considerably. For $h\neq\ell$,
\[
H_{\ell,h}(\mathbf p)>0
\quad\Leftrightarrow\quad
\frac{p_h-\pi_h^*}{\gamma_{\{\ell,h\}}}<0.
\]
For $h,m\neq\ell$,
\[
H_{h,m}(\mathbf p)>0
\quad\Leftrightarrow\quad
\frac{p_h-\pi_h^*}{\gamma_{\{\ell,h\}}}
>
\frac{p_m-\pi_m^*}{\gamma_{\{\ell,m\}}}.
\]
Thus all pairwise separators, with their orientations, compare the $J$
quantities
\[
0
\quad\text{and}\quad
\frac{p_h-\pi_h^*}{\gamma_{\{\ell,h\}}},
\qquad h\neq\ell.
\]
Their regions therefore correspond to the $J!$ possible strict
orderings of these quantities.

Let's restrict this arrangement to
$\relint(\widetilde{\Delta}^{B}_{J-1})$. By \eqref{pistar},
for every nonempty $S\subset\mathcal J$, $\sum_{j\in S}\pi_j^*
=
\sum_A\gamma_A\frac{|A\cap S|}{|A|}$. When $S$ is proper, connectedness guarantees an active consideration set intersecting both
$S$ and $S^c$, so the corresponding feasible set inequalities are strict
at $\boldsymbol\pi^*$, giving  $\boldsymbol\pi^*
\in
\relint\!\left(\widetilde{\Delta}^{B}_{J-1}\right)$. Since all pairwise hyperplanes pass through $\boldsymbol\pi^*$ and a
relative neighborhood of $\boldsymbol\pi^*$ is contained in the feasible
set, every one of the $J!$ regions intersects this relative interior.
These intersections are relatively open and convex.


Finally, suppose a generated probability vector lies in the strict region
corresponding to $(k_1,\ldots,k_J)$. Then $H_{k_r,k_s}(\mathbf p(\mathbf x))>0$ 
for all $r<s$, so Theorem~\ref{th:linear_generalJ} (Bi)-(Bii) implies $x_{k_1}\beta_{k_1}>
x_{k_2}\beta_{k_2}>
\cdots>
x_{k_J}\beta_{k_J}$. 

Conversely, the strict ordering $x_{k_1}\beta_{k_1}>x_{k_2}\beta_{k_2}>\cdots>x_{k_J}\beta_{k_J}$ implies, by Theorem~\ref{th:linear_generalJ} (Bi)--(Bii), that $H_{k_r,k_s}(\mathbf p(\mathbf x))\geq0$ for all $r<s$. Choose any $\mathbf q\in\mathcal C_\kappa$. Since $\mathbf p(\mathbf x)\in\widetilde{\Delta}^{B}_{J-1}$, $\mathbf q\in\relint(\widetilde{\Delta}^{B}_{J-1})$, and $\widetilde{\Delta}^{B}_{J-1}$ is convex, for every $t\in(0,1]$, we have $(1-t)\mathbf p(\mathbf x)+t\mathbf q
\in
\relint(\widetilde{\Delta}^{B}_{J-1})$. By linearity of $H_{k_r,k_s}$, we have $H_{k_r,k_s}
\bigl((1-t)\mathbf p(\mathbf x)+t\mathbf q\bigr)>0$
 for all $r<s$. Thus, 
$(1-t)\mathbf p(\mathbf x)+t\mathbf q\in\mathcal C_\kappa$ for every $t>0$, and these vectors converge to
$\mathbf p(\mathbf x)$ as $t\downarrow0$. Therefore, 
$\mathbf p(\mathbf x)$ lies in the closure of $\mathcal C_\kappa$.

If $\gamma_{\{k,\ell\}}>0$ for every distinct pair $k,\ell$,
Theorem~\ref{th:linear_generalJ} (Bi) makes all these inequalities
strict under a strict ranking and gives 
$H_{k,\ell}(\mathbf p(\mathbf x))=0$
$\iff$
$x_k\beta_k=x_\ell\beta_\ell$. The same conclusions follow, regardless of the values of
$\gamma_{\{k,\ell\}}$, under the full support restriction in
Theorem~\ref{th:linear_generalJ} (Biii). Hence, under either
additional condition, the attainable trace of each open cell coincides
with the corresponding strict ranking region.
$\blacksquare$

\begin{example} \label{ex:star}
Let \(J=4\), and suppose the only active consideration sets are $\{1,4\}$, $\{2,4\}$, $\{3,4\}$, 
each with probability \(1/3\). The consideration structure is connected, and \(C^{(4)}_{k,\ell}\) has rank one for every distinct pair \(k,\ell\). Let the unobservables be i.i.d. uniform on \([-1,1]\).

Consider the deterministic index vector $(x_1 \beta_1,x_2 \beta_2,x_3 \beta_3,x_4 \beta_4)=(4,3,-3,0)$.
Alternative 1 always wins from \(\{1,4\}\), alternative 2 always wins from \(\{2,4\}\), and alternative 4 always wins from \(\{3,4\}\). Hence, the choice probability vector is 
$\mathbf p=\left(\frac13,\frac13,0,\frac13\right)$, whereas $\pi^*=\left(\frac16,\frac16,\frac16,\frac12\right)$. 
For the pair \((1,2)\), equation \eqref{eq:Hrank1J4} becomes $H_{1,2}(\mathbf p) = \frac13\left[ \left(p_1-\frac16\right) - \left(p_2-\frac16\right) \right]=0$,
although \(x_1 \beta_1>x_2 \beta_2\). Moreover, the index vectors \((3,4,-3,0)\) and \((7/2,7/2,-3,0)\) generate the same probability vector. Thus, this vector is compatible with either strict ranking of alternatives 1 and 2, as well as their indifference. \(\blacksquare\)
\end{example}


\noindent\textbf{Proof of Theorem \ref{th:linear_other_coordinates}.} Fix $k_1 \neq k_2$ and denote $c_{k_1}=1$, $c_{k_2}=-1$. The left-hand side of
\eqref{eq:other-coordinate-hyperplane} once we have substituted a specific realization of the  probability vector and $P(Y=h|\mathbf x)-\pi_h^*=\sum_{A:h\in A}\gamma_A
 \left(q_h^A(\mathbf x)-\frac1{|A|}\right)$ is 
\begin{align}
 L(\mathbf p(\mathbf x)): =\sum_{A\in\mathcal A}\gamma_A
 \sum_{h\in A}c_h
 \left(q_h^A(\mathbf x)-\frac1{|A|}\right).
 \label{eq:L-menu-decomposition}
\end{align}
Let's evaluate the contribution to \eqref{eq:L-menu-decomposition} of each type
of menu.

First, let $A\subseteq\mathcal J\setminus\{k_1,k_2\}$.  By condition (i), all alternatives in an active such menu have a common coefficient,
say $c_A$.  Hence its contribution  $\gamma_Ac_A\sum_{h\in A}
 \left(q_h^A(\mathbf x)-\frac1{|A|}\right)=0$,
Singleton menus also contribute zero to \eqref{eq:L-menu-decomposition}. 

Second, consider a menu $A=S\cup\{k_1,k_2\}$ with nonempty
$S\subseteq\mathcal J\setminus\{k_1,k_2\}$.  Condition (iii) gives
$c_s=0$ for every $s\in S$, so its contribution is $ \gamma_{S\cup\{k_1,k_2\}}
 \left(q_{k_1}^{S\cup\{k_1,k_2\}}(\mathbf x)
       -q_{k_2}^{S\cup\{k_1,k_2\}}(\mathbf x)\right)$.
The same formula, with $S=\varnothing$, applies to the menu
$\{k_1,k_2\}$.

Third, fix a nonempty outside set $S$ and consider the matched exclusive menus
$S\cup\{k_1\}$ and $S\cup\{k_2\}$. If $\gamma_{S\cup\{k_1\}}+\gamma_{S\cup\{k_2\}}=0$, the pair contributes nothing.  Otherwise,
condition (ii) assigns every $s\in S$ the common coefficient $ c_S=\frac{\gamma_{S\cup\{k_1\}}-\gamma_{S\cup\{k_2\}}}{\gamma_{S\cup\{k_1\}}+\gamma_{S\cup\{k_2\}}}$.
Using
$\sum_{s\in S}(q_s^{S\cup\{k_1\}}-1/(|S|+1))
=-(q_{k_1}^{S\cup\{k_1\}}-1/(|S|+1))$, the contribution of the first menu is
{\small\begin{align*}
 \gamma_{S\cup\{k_1\}}\left[
 q_{k_1}^{S\cup\{k_1\}}-\frac{1}{|S|+1}
 +c_S\sum_{s\in S}
   \left(q_s^{S\cup\{k_1\}}-\frac{1}
   {|S|+1} \right)
 \right]
 &=\gamma_{S\cup\{k_1\}}(1-c_S)
   \left(q_{k_1}^{S\cup\{k_1\}}-\frac{1}{|S|+1}\right).
\end{align*}}
Similarly, the contribution of $S\cup\{k_2\}$ is
$ -\gamma_{S\cup\{k_2\}}(1+c_S)
 \left(q_{k_2}^{S\cup\{k_2\}}-\frac{1}{|S|+1}\right)$.
Taking into account that $\gamma_{S\cup\{k_1\}}(1-c_S)=\gamma_{S\cup\{k_2\}}(1+c_S)
 =\frac{2\gamma_{S\cup\{k_1\}}\gamma_{S\cup\{k_2\}}}{\gamma_{S\cup\{k_1\}}+\gamma_{S\cup\{k_2\}}}$,
the combined contribution of the two matched menus is
\begin{equation*}
 \frac{2\gamma_{S\cup\{k_1\}}\gamma_{S\cup\{k_2\}}}{\gamma_{S\cup\{k_1\}}+\gamma_{S\cup\{k_2\}}}
 \left(q_{k_1}^{S\cup\{k_1\}}(\mathbf x)
       -q_{k_2}^{S\cup\{k_2\}}(\mathbf x)\right).
\end{equation*}

Combining the  three cases yields the explicit decomposition
\begin{multline}
 L(\mathbf x)
 =\gamma_{\{k_1,k_2\}}
 \left(q_{k_1}^{\{k_1,k_2\}}(\mathbf x)-q_{k_2}^{\{k_1,k_2\}}(\mathbf x)\right)\\
 +\sum_{\varnothing\neq S\subseteq\mathcal J\setminus\{k_1,k_2\}}
 \gamma_{S\cup\{k_1,k_2\}}
 \left(q_{k_1}^{S\cup\{k_1,k_2\}}(\mathbf x)-
       q_{k_2}^{S\cup\{k_1,k_2\}}(\mathbf x)\right)\\
 +\sum_{\substack{\varnothing\neq S\subseteq
                   \mathcal J\setminus\{k_1,k_2\}:\\
                   \gamma_{S\cup\{k_1\}}+
                   \gamma_{S\cup\{k_2\}}>0}}
 \frac{2\gamma_{S\cup\{k_1\}}\gamma_{S\cup\{k_2\}}}
      {\gamma_{S\cup\{k_1\}}+\gamma_{S\cup\{k_2\}}}
 \left(q_{k_1}^{S\cup\{k_1\}}(\mathbf x)-
       q_{k_2}^{S\cup\{k_2\}}(\mathbf x)\right).
 \label{eq:full-L-decomposition-other-coordinates}
\end{multline}

By exchangeability, every difference in
\eqref{eq:full-L-decomposition-other-coordinates} has the weak sign of $x_{k_1}\beta_{k_1}-x_{k_2}\beta_{k_2}$. The  difference $q_{k_1}^{\{k_1,k_2\}}(\mathbf x)-q_{k_2}^{\{k_1,k_2\}}(\mathbf x)$ is strict whenever the two deterministic indices differ. Under full support, every  difference with a positive weight in
\eqref{eq:full-L-decomposition-other-coordinates} is strict whenever the
two deterministic indices differ.

Note that all weights in
\eqref{eq:full-L-decomposition-other-coordinates} are nonnegative. Suppose they were all zero. Then no active menu could contain both \(k_1\) and \(k_2\), and \(\gamma_{S\cup\{k_1\}}\gamma_{S\cup\{k_2\}}=0\) for every nonempty \(S\subseteq\mathcal J\setminus\{k_1,k_2\}\). Condition (ii) would therefore make every coefficient in an active menu containing \(k_1\) equal to \(1\), and every coefficient in an active menu containing \(k_2\) equal to \(-1\). Condition (i) makes the coefficients equal within every other active menu. Along a path of overlapping active menus connecting \(k_1\) to \(k_2\), these equalities would imply \(c_{k_1}=c_{k_2}\), contradicting \(c_{k_1}=1\) and \(c_{k_2}=-1\). Thus at least one weight is strictly positive.

Hence, 
\[
x_{k_1}\beta_{k_1}=x_{k_2}\beta_{k_2}
\;\Longrightarrow\;
L(\mathbf x)=0,
\qquad
x_{k_1}\beta_{k_1}>(<)x_{k_2}\beta_{k_2}
\;\Longrightarrow\;
L(\mathbf x)\geq(\leq)0.
\]
If $\gamma_{\{k_1,k_2\}}>0$, the  term $\gamma_{\{k_1,k_2\}}
 \left(q_{k_1}^{\{k_1,k_2\}}(\mathbf x)-q_{k_2}^{\{k_1,k_2\}}(\mathbf x)\right)$ makes the latter
inequalities strict. Under full support, they are strict whenever the
indices differ even if $\gamma_{\{k_1,k_2\}}=0$. This proves the stated
sign characterizations. Finally,
\eqref{eq:other-coordinate-hyperplane} defines a hyperplane of dimension $J-2$ in the simplex. $\blacksquare$  

\medskip
\noindent\textbf{Proof of Corollary \ref{cor:linear-allpairs-other-coordinates}.} Apply Theorem~\ref{th:linear_other_coordinates} to each adjacent pair. Transitivity gives the complete ordering, and distinct strict orderings are mutually exclusive. The final claim follows from the exact sign cases of the theorem. $\blacksquare$

\vskip 0.05in

\textbf{Proof of Theorem \ref{th:cs_heteroskedasticity}.}
A. Under the common location-scale representation, common location terms cancel and $U_j\geq U_k
\;\Longleftrightarrow\;
\widetilde\varepsilon_k-\widetilde\varepsilon_j
\leq
\frac{x_j\beta_j-x_k\beta_k}{s(\mathbf x)}.$

Thus the choice probability vector depends on $(\mathbf x,\boldsymbol\beta)$ only through $d=d^{\boldsymbol\beta}(\mathbf x)$, which defines $\Psi$ on $\mathcal E^o_{\mathrm{diff}}$. Connectivity and interior support imply that every alternative is chosen with positive probability on this domain, so $\Psi(d)\in\relint(\Delta_{J-1})$, and absolute continuity gives continuity.

To prove injectivity, take \(d\neq\widetilde d\), with \(d_1=\widetilde d_1=0\), and define
$T=\{j:d_j>\widetilde d_j\}$.
Interchanging the two vectors if necessary, \(T\) is nonempty. Its complement is also nonempty as \(1\notin T\). Moreover,
$d_j-d_k>\widetilde d_j-\widetilde d_k$
for every $j\in T$, $k\notin T$. Thus, in every consideration set, the probability of choosing an alternative in \(T\) weakly increases when the indices change from \(\widetilde d\) to \(d\). 
Connectivity implies that some active consideration set \(A\) contains alternatives both in and outside \(T\). Choose \(\eta>0\) sufficiently small, with \(\eta<\min_{j\in T}(d_j-\widetilde d_j).\) Since \(d\in E_{\mathrm{diff}}^o\), the support of the realized utility differences
$\bigl(d_j+\widetilde\varepsilon_j-\widetilde\varepsilon_1\bigr)_{j=1}^J$ 
contains a relative neighborhood of zero. Hence, for sufficiently small \(\eta\), there is a positive probability event on which these differences are sufficiently close to \(\eta\) for \(j\in T\) and to zero for \(j\notin T\). On this event, the alternative chosen from \(A\) under \(d\) belongs to \(T\). Replacing \(d\) by \(\widetilde d\) decreases each utility in \(T\) by more than \(\eta\), while weakly increasing every utility outside \(T\). On a sufficiently small such neighborhood, the alternative chosen under \(\widetilde d\) therefore lies outside \(T\).
Consequently, the probability of choosing from \(T\) strictly increases for menu \(A\) and weakly increases for every other menu. Since \(\gamma_A>0\), we have 
$\sum_{j\in T}\Psi_j(d)>
\sum_{j\in T}\Psi_j(\widetilde d)$.
Thus \(\Psi(d)\neq\Psi(\widetilde d)\), proving injectivity.

We must also exclude identical probabilities generated by an index vector outside \(E_{\mathrm{diff}}^o\). Fix \(d\in E_{\mathrm{diff}}^o\) and any distinct normalized vector \(\widetilde d\). A sufficiently short segment from \(d\) toward \(\widetilde d\) remains inside \(E_{\mathrm{diff}}^o\). The preceding argument gives a strict change along this segment in the aggregate choice probability of a nonempty proper set of alternatives. Along the rest of the segment toward \(\widetilde d\), their indices continue to move relative to the remaining alternatives in the same direction, so this aggregate probability changes weakly in that direction. Thus \(\widetilde d\) cannot generate the same probability vector as \(d\).
 
 Since the projection of $\mathcal E^o_{\mathrm{diff}}$ onto the last $J-1$ coordinates is open in $\mathbb R^{J-1}$, invariance of domain implies that $\Psi$ is a homeomorphism onto its image, and thus a topological embedding. Its inverse therefore recovers $d$ and all pairwise signs $d_j-d_k$.

B. For any $C_1, C_2\subseteq\mathcal E^o_{\mathrm{diff}}$, injectivity gives
$\Psi(C_1)\cap\Psi(C_2)=\Psi(C_1\cap C_2)$. Taking $C_1=\mathcal D_{\boldsymbol\beta}$ and $C_2=\mathcal E^o_{\mathrm{diff}}\cap\{d:d_j=d_k\}$ proves the displayed identity, and the cases $>$ and $<$ are identical. The set $C_2$ and its image do not depend on $\boldsymbol\beta$, whereas $\mathcal D_{\boldsymbol\beta}$ may. This proves statements on the common split of $\Psi(E^o_{\mathrm{diff}})$ by $\{d:d_j=d_k\}$ and attainable trace. 

More generally, injectivity gives $ \Psi(C_1)\cap\cdots\cap\Psi(C_m)
 =
 \Psi(C_1\cap\cdots\cap C_m)$ for any finite
collection $C_1,\ldots,C_m\subseteq\mathcal E^o_{\mathrm{diff}}$. Fix a permutation $(k_1,\ldots,k_J)$.  The deterministic inequalities
$d_{k_1}>d_{k_2}>\cdots>d_{k_J}$ are equivalent to the $J-1$ adjacent
inequalities
$d_{k_r}>d_{k_{r+1}}$, $r=1,\ldots,J-1$.  Applying the above identity with a finite intersection gives $\Psi\!\left(
\mathcal D_{\boldsymbol\beta}\cap
\{d:d_{k_1}>d_{k_2}>\cdots>d_{k_J
}\}
\right)=
\Psi(\mathcal D_{\boldsymbol\beta})
\cap
\bigcap_{r=1}^{J-1}
\Psi\!\left(
\mathcal E^o_{\mathrm{diff}}\cap
\{d:d_{k_r}>d_{k_{r+1}}\}
\right)$.  Injectivity also guarantees that the images $\Psi\!\left(\mathcal E^o_{\mathrm{diff}}\cap \{d:d_{k_1}>d_{k_2}>\cdots>d_{k_J}\} \right)$ are pairwise disjoint across distinct permutations $(k_1,\ldots,k_J)$.  

C. $\mathcal E^o_{\mathrm{diff}}$ is relatively open and convex in $\{d\in\mathbb R^J:d_1=0\}$. Moreover, exchangeability and convexity of the support imply that $0\in\mathcal E^o_{\mathrm{diff}}$. Hence, for every $j\neq k$, we have that  $\mathcal E^o_{\mathrm{diff}}\cap\{d:d_j=d_k\}$ is a nonempty relatively open subset of a $(J-2)$-dimensional hyperplane and therefore a $(J-2)$-dimensional topological manifold. Its complement in $\mathcal E^o_{\mathrm{diff}}$ has exactly two connected components, $\mathcal E^o_{\mathrm{diff}}\cap\{d:d_j>d_k\}$ and $\mathcal E^o_{\mathrm{diff}}\cap\{d:d_j<d_k\}$, which are nonempty and convex. Since $\Psi$ is a homeomorphism onto its image, the same manifold and separation properties hold for their images under $\Psi$. Finally, relative openness at $0$ implies that every strict order region in the index space intersects $\mathcal E^o_{\mathrm{diff}}$ arbitrarily close to $0$. Each such intersection is convex and, hence, connected. Thus, all $J!$ strict order regions in the probability space are nonempty and connected, and continuity of $\Psi$ implies that $\Psi(0)$ lies in the closure of each. $\blacksquare$

\vskip 0.05in 

\noindent\textbf{Proof of Theorem~\ref{th:nonexch_specialcase}.} A. The location-scale representation again makes the choice probability vector a function $\Psi(d)$ of $d=d^{\boldsymbol\beta}(\mathbf x)$ alone. Interior support implies that every alternative has positive winning probability on $\mathcal E^o_{\mathrm{diff}}$, and absolute continuity gives continuity. The injectivity argument in the proof of Theorem \ref{th:cs_heteroskedasticity} applies with the full choice set \(\mathcal J\) as the only consideration set. The positive probability event constructed there follows from \(d\in E_{\mathrm{diff}}^o\) and does not require exchangeability. Hence, \(\Psi\) is injective. As in the proof of Theorem \ref{th:cs_heteroskedasticity} , invariance of domain makes \(\Psi\) a homeomorphism onto its image.

At complete indifference, the common marginal c.d.f. give
$P\!\left(\widetilde\varepsilon_j\geq\widetilde\varepsilon_k
\text{ for all }k\right)
=
\int_0^1
\frac{\partial\widetilde C}{\partial u_j}(v,\ldots,v)\,dv
=
\pi_j^{*,M}.$ Whenever $0\in\mathcal E^o_{\mathrm{diff}}$, the left-hand side is $\Psi_j(0)$. 

B. Analogous to the proof of Theorem \ref{th:cs_heteroskedasticity}.

C. $\mathcal E^o_{\mathrm{diff}}$ is relatively open and convex in $\{d\in\mathbb R^J:d_1=0\}$. For any $j\neq k$, identical marginal distributions imply $P(\widetilde\varepsilon_j>\widetilde\varepsilon_k)>0$ and $P(\widetilde\varepsilon_j<\widetilde\varepsilon_k)>0$. Hence,  $\mathcal E^o_{\mathrm{diff}}$ contains points on both sides of $d_j=d_k$. By convexity,  $\mathcal E^o_{\mathrm{diff}}\cap\{d:d_j=d_k\}$  is  a nonempty relatively open subset of a $(J-2)$-dimensional hyperplane (and hence a $(J-2)$-dimensional topological manifold). Its complement  in $\mathcal E^o_{\mathrm{diff}}$ has exactly two connected components, $\mathcal E^o_{\mathrm{diff}}\cap\{d:d_j>d_k\}$ and $\mathcal E^o_{\mathrm{diff}}\cap\{d:d_j<d_k\}$,  which are convex. Since $\Psi$ is a homeomorphism onto its image, the same manifold and separation properties hold for their images under $\Psi$. Finally, if $0\in\mathcal E^o_{\mathrm{diff}}$, relative openness implies that every strict order region intersects $\mathcal E^o_{\mathrm{diff}}$ arbitrarily close to $0$. Each such intersection is convex and hence connected. Thus all $J!$ strict order regions in the probability space are nonempty and connected, and continuity of $\Psi$ implies that $\Psi(0)=\boldsymbol\pi^{*,M}$ lies in the closure of each. \hfill$\blacksquare$



\section{Non-exchangeable bivariate copulas and the quantile threshold}

First, we show that any quantile index $\tau\in(0,1)$ can be generated by a suitable absolutely continuous joint distribution of the unobservables $(\varepsilon_1,\varepsilon_0)$ satisfying the support conditions of Theorem~\ref{prop:copula}. The construction below corresponds to the special case of Theorem \ref{prop:copula} in which the conditional marginal distributions of the two unobservables are the same, so that departures from the benchmark ordering probability are driven entirely by the copula.

\noindent\textbf{Proof.} First, we are going to use absolutely continuous checkerboard copulas. 
Let \(p=1-\tau\). First suppose \(p=m/n\) with
\(m\in\{1,\ldots,n-1\}\). Partition \([0,1]\) into intervals
\(I_i=(i/n,(i+1)/n]\), \(i=0,\ldots,n-1\), and set \(k=n-m\). Define
\(
c(u,v)
=
n\sum_{i=0}^{n-1}
\mathbf 1\{u\in I_i,\ v\in I_{i+k\,(\mathrm{mod}\,n)}\}.
\)
This is a nonnegative density with uniform margins, hence defines an
absolutely continuous copula. Since exactly \(m\) of the \(n\) occupied
squares lie above the diagonal,
\(
P_C(V\geq U)=m/n=p.
\) For general \(p\in(0,1)\), choose rationals \(q_0<p<q_1\). By the
previous step, there are absolutely continuous copulas \(C_0,C_1\) with
\(P_{C_j}(V\geq U)=q_j\), \(j=0,1\). With
$\lambda=\frac{p-q_0}{q_1-q_0}$, $C=(1-\lambda)C_0+\lambda C_1$, the copula \(C\) is absolutely continuous and satisfies
\(
P_C(V\geq U)=(1-\lambda)q_0+\lambda q_1=p=1-\tau\). 

To ensure the support conditions of Theorem~\ref{prop:copula}, let
$K\subset\mathbb R^2$ be any convex support with nonempty interior
that is invariant to interchange of its coordinates. Choose an
absolutely continuous exchangeable distribution $G$ with support $K$
and density positive on its interior. By convexity and invariance of $K$ under interchange of coordinates, 
$K$ contains a square $I\times I$ centered at a point on the diagonal.

Choose $0<\delta<2\min\{p,1-p\}$ and let
$q=\frac{p-\delta/2}{1-\delta}\in(0,1)$. Let $H_q$ be the checkerboard distribution constructed above, rescaled
to $I\times I$, and define
$F_\delta=\delta G+(1-\delta)H_q$. Then $F_\delta$ is absolutely continuous, has identical marginals and
support $K$, while
$P_{F_\delta}(\varepsilon_1\geq\varepsilon_0)
=\frac{\delta}{2}+(1-\delta)q=p=1-\tau$.
Since $K$ contains $I\times I$, the support of
$\varepsilon_1-\varepsilon_0$ contains a neighborhood of zero.
Thus every $\tau\in(0,1)$ is attainable under all the maintained
support conditions of Theorem~\ref{prop:copula}. $\blacksquare$

\section{Distributional route: Additional details for Section \ref{sec:distributionalmulti}}

This appendix records additional details behind the discussion in Section
\ref{sec:distributionalmulti}.  

First, let's  make it explicit why a fixed copula and
common conditional marginals stabilize the image of total indifference but do
not, by themselves, produce a common partial-indifference boundary.  

Consider the image $(\pi_{1}^*(\mathbf{x}), \ldots, \pi_{J}^*(\mathbf{x}))^{\top}$ in  $\Delta_{J-1}$ of an $\mathbf{x}$ that corresponds to the case of complete indifference \eqref{totalindiff}: $\pi_{k}^*(\mathbf{x})  = \int_{-\infty}^{+\infty} \frac{\partial C_{\mathbf{x}}}{\partial u_{k}} (F_{1,\mathbf{x}}(\varepsilon_{k}), \ldots,F_{J,\mathbf{x}}(\varepsilon_{k})) dF_{k,\mathbf{x}}(\varepsilon_{k})$, $k=1,\ldots, J.$   We already see a difference with the behavioral route of consideration sets. Here, even at the level of an \textit{empirical specification}
(that is, for fixed $(M,\boldsymbol{\beta},\mathbb P_{\mathbf x})$),
collecting such images across
$\mathbf{x}\in\operatorname{supp}(\mathbb P_{\mathbf x})$
may trace out a set whose affine dimension  varies from zero (when this point does not depend on $\mathbf{x}$) to $J-1$ (when the convex hull of such points has a relative interior in $\Delta_{J-1}$).
Passing from the empirical specification to the corresponding
fixed-$\boldsymbol{\beta}$ structural model enlarges the covariate
domain from $\operatorname{supp}(\mathbb P_{\mathbf x})$ to
$\mathbb R^{M_e}$. Hence the resulting collection of image points
can only expand, so its affine dimension can
increase (up to $J-1$) or remain unchanged. An analogous weak expansion conclusion applies when passing from a
fixed-$\boldsymbol{\beta}$ structural model $(M,\boldsymbol{\beta})$
to the corresponding fixed-$\boldsymbol{\beta}$ behavioral
specification $(B,\boldsymbol{\beta})$, since the latter ranges over
the admissible conditional distributions in $\mathcal P$. At the heart of the broadness of such conclusions is, once again, 
the fact that the properties of the behavior of the joint distribution of unobservables conditional on $\mathbf x$ with respect to $\mathbf x$ are unrestricted. To have a hope of  getting an analogue of a ``quantile model'' at the \textit{Structural environment} level, a natural first constraint  is constancy in \(x\) of the  total indifference image \(\pi^*(x)\). 

Suppose that (i) the copula $C_{\mathbf{x}}$ does not depend on ${\mathbf{x}}$ and (ii) all unobservable $\varepsilon_j |\mathbf{x}$ have exactly the same marginal distributions: $F_{k_j,\mathbf x}(e) = F_{k_p,\mathbf x}(e)=F_{\mathbf x}(e)$ for all $e \in \mathbb{R}$ for any $k_j$, $k_p$. Then $\pi^*(\mathbf{x})$ does not depend on $\mathbf{x}$ as 
\begin{align*}\pi_{k}^{*}(\mathbf{x})  & = \int_{-\infty}^{+\infty} \frac{\partial C}{\partial u_{k}} (F_{\mathbf{x}}(\varepsilon_{k}), \ldots,F_{\mathbf{x}}(\varepsilon_{k})) dF_{\mathbf{x}}(\varepsilon_{k})   =\int_{0}^{1} \frac{\partial C}{\partial u_{k}} (v, \ldots,v) dv :=\pi^{*,M}_k, \end{align*}
$k=1,\ldots, J,$ does not depend on $\mathbf{x}$. Under conditions (i) and (ii), the notation of this image point by
$\pi_k^{*,M}$ emphasizes that it is specific to the
\textit{Structural environment}. Since $C$ need not be exchangeable,
we should not expect $\pi_j^{*,M}=1/J$ for all $j$. Moreover, the point
$\boldsymbol{\pi}^{*,M}$ need not be common across the \emph{Structural
environments} admitted by a given \emph{Behavioral environment}. Variation in the total indifference image alone, however, does not establish ranking ambiguity.

Constancy of $\pi^{*,M}_k$ is not sufficient for a common
partial indifference boundary.  To see this more explicitly, continue with restrictions (i) and (ii) above (constancy of the copula and the same conditional marginal c.d.f.) and take $J=3$ and
$x_1\beta_1=x_2\beta_2$, and write $a=x_1\beta_1-x_3\beta_3$.  Define
\begin{align*}
 \phi_{j3,\mathbf x}(a)
 :={}
 \int
 \frac{\partial C}{\partial u_j}
 \left(F_{\mathbf x}(e),F_{\mathbf x}(e),
 F_{\mathbf x}(e+a)\right)
 \,dF_{\mathbf x}(e),\quad j=1,2.
\end{align*} 
The partial indifference image can be written parametrically as
\begin{equation}
\label{eq:appendixE-union-curves}
\begin{aligned}
I^{M,\beta}_{1,2}\cap\operatorname{relint}(\Delta^{M,\beta}_2)
={}&
\Bigl\{
\bigl(\phi_{13,x}(a),\phi_{23,x}(a),
1-\phi_{13,x}(a)-\phi_{23,x}(a)\bigr)^\top:
\\
&\qquad x\in\mathbb R^{\widetilde M},\
x_1\beta_1=x_2\beta_2,\
a=x_1\beta_1-x_3\beta_3
\Bigr\}
\cap\operatorname{relint}(\Delta^{M,\beta}_2).
\end{aligned}
\end{equation}
Although the total indifference point is common, the functions \(\phi_{13,x}\) and \(\phi_{23,x}\) may vary with the shape of \(F_x\). A fixed copula and common conditional marginals therefore do not by themselves guarantee a common partial indifference curve. 

To give an example of this, consider the trivariate copula $C(u_1,u_2,u_3)
=
u_1u_2u_3
[1+\frac12(1-u_1)(1-u_3)].$ Its density is  $c(u_1,u_2,u_3)
=
1+\frac12(1-2u_1)(1-2u_3)$, which lies between \(1/2\) and \(3/2\). Hence \(C\) is an absolutely continuous copula with full support on the unit cube. 
Let \(U=(U_1,U_2,U_3)\) have this copula and uniform marginals. Consider two unobservable distributions: 
$\varepsilon_j=U_j$, $j=1,2,3$, and  $\varepsilon_j=U_j^2$,
 $j=1,2,3$. Both have the same copula \(C\), identical marginals across alternatives, and an absolutely continuous joint distribution with convex support \([0,1]^3\).
Under the tied deterministic indices \(x_1\beta_1=x_2\beta_2=0\), \(x_3\beta_3=-a\), write the first two choice probabilities as \(p_1(a)\) and \(p_2(a)\). 
At complete indifference, both distributions give $\bigl(p_1(0),p_2(0)\bigr)
=
\left(\frac{13}{40},\frac7{20}\right)$. Their right derivatives, however, differ: $p_1'(0+)=7/12$ and $p_2'(0+)\
=1/2$ in the first case whereas $p_1'(0+)=7/12$ and $p_2'(0+)=13/24$ in the second case, respectively, giving $6/7$ and $13/14$ slopes in the $(p_1,p_2)$-plane. Thus, the two partial indifference curves cannot be different parametrizations of the same local curve.

One way to obtain such a common partial indifference curve is to suppose, in addition, that $\boldsymbol\varepsilon$ is independent of
$\mathbf x$.  The functions above no longer depend on $\mathbf x$ and  the indifference locus is the curve $ \mathcal I_{1,2}^{M,\boldsymbol{\beta}}\cap\relint(\Delta^{M,\boldsymbol{\beta}}_2)
 = \left\{(\phi_{13}(a),\phi_{23}(a),
 1-\phi_{13}(a)-\phi_{23}(a))^{\top}:a\right\} \cap\relint(\Delta^{M,\boldsymbol{\beta}}_2),$
where $a$ ranges over the values of
$x_1\beta_1-x_3\beta_3$ attainable subject to
$x_1\beta_1=x_2\beta_2$. For the separation argument, consider index configurations satisfying \(d_\beta(x)\in E_{\mathrm{diff}}^o\), as in Theorem \ref{th:nonexch_specialcase}. Suppose \(x_1\beta_1>x_2\beta_2\), and write \(a=x_1\beta_1-x_3\beta_3\). Raising the second index to \(x_1\beta_1\) produces the tied configuration \((x_1\beta_1,x_1\beta_1,x_3\beta_3)\), with choice probabilities \(\phi_{13}(a)\) and \(\phi_{23}(a)\). This change weakly decreases the probability of choosing alternative 1 and weakly increases that of choosing alternative 2.
These changes are strict. Membership in \(E_{\mathrm{diff}}^o\) implies that the support of the realized utility differences contains a neighborhood of complete utility indifference. Consequently, with positive probability, alternative 1 initially beats alternative 2 by a positive gap smaller than \(x_1\beta_1-x_2\beta_2\), and both beat alternative 3. Raising the second index then changes the chosen alternative from 1 to 2. Hence, $P(Y=1| x)>\phi_{13}(a)$, $P(Y=2| x)<\phi_{23}(a)$.

 Thus, when \(x_1\beta_1>x_2\beta_2\), the probability vector lies strictly southeast of a point on the parametric indifference curve in the \((p_1,p_2)\)-plane. When \(x_1\beta_1<x_2\beta_2\), it lies strictly northwest of a point on that curve. These two sets are disjoint because both coordinates of the curve are nondecreasing in \(a\). Hence, the two strict ranking images generated by configurations in \(E_{\mathrm{diff}}^o\) are disjoint. The comparison uses the full structural curve with the attainable indifference locus being its trace over the index configurations admitted by the model.
Partial exchangeability can make one of structural boundary curves linear even when the
full unobservables distribution is not exchangeable. If \(C(u_1,u_2,u_3)\) is invariant to interchanging \(u_1\) and \(u_2\), then \(\phi_{13}(a)=\phi_{23}(a)\). For configurations satisfying \(d_\beta(x)\in E_{\mathrm{diff}}^o\), we have 
$p_1=p_2 \iff x_1\beta_1=x_2\beta_2$, $p_1>p_2 \iff x_1\beta_1>x_2\beta_2$. The other two pairwise boundaries may remain nonlinear if alternative 3 is not
exchangeable with alternatives 1 and 2.  Hence linearity of one pairwise
comparison does not imply an exchangeable multinomial model or a fully linear
partition.

\vskip 0.05in

\noindent \textit{Dependence on $\mathbf x$.}
The curves in Figure~\ref{fig:nonexch_gaussian} also give a direct construction
of the overlap described in the main text.  Let each $x_j$ be one-dimensional
and set $\beta_1=\beta_2=\beta_3=1/\sqrt{3}$.  Use Case 1 distribution  when
$x_1>x_2$, Case 3  when $x_1<x_2$, and Case 2 when $x_1=x_2$.  Each fixed distribution generates a smooth separating
curve.  The curves differ already at total indifference. For a centered Gaussian vector with unit variances,
\[
\pi_j^{*,M}
=\frac14+\frac{1}{2\pi}
\arcsin\!\left(
\frac{1-\rho_{jk}-\rho_{j\ell}+\rho_{k\ell}}
{2\sqrt{(1-\rho_{jk})(1-\rho_{j\ell})}}
\right),
\qquad \{j,k,\ell\}=\{1,2,3\}.
\]This gives \((\pi_1^{*,M},\pi_2^{*,M})\approx(0.374820,0.338060)\) in Case 1 and \((0.420579,0.315551)\) in Case 3. The Case 1 point has a smaller first coordinate but a larger second coordinate. Since both coordinates increase strictly along each curve, Case 1 lies strictly above Case 3 on a nonempty interval. The strip between Case 1 and Case 3 curves is on the $x_1\beta_1>x_2\beta_2$ side
of Case 1 curve and on the $x_1\beta_1<x_2\beta_2$ side of Case 3 curve.  The same argument as in Section \ref{sec:cs_structural}, now using Theorem \ref{th:nonexch_specialcase} and full consideration, shows that each fixed Gaussian probability map is onto \(\relint(\Delta_2)\). Thus, probabilities strictly between the two curves are attainable under both distributions, with the opposite rankings indicated above. This construction gives an explicit instance of ranking ambiguity. The following proposition records the more general point that once the
admissible distributional class is sufficiently rich, failure of symmetry
at pairwise index indifference can itself generate such ambiguity, even
for a fixed $\boldsymbol{\beta}$. 

\begin{prop}[Failure of distribution-robust ranking recovery]
\label{prop:distributional_no_uniform_recovery}
Fix $\boldsymbol{\beta}\in\mathcal B_0$ and suppose the joint index
differences range over $\mathbb R^{J-1}$. Suppose the
\emph{Behavioral environment} admits an $\mathbf{x}$-independent,
absolutely continuous distribution
$P_{\boldsymbol{\varepsilon}}$ with common marginals and full support
on $\mathbb R^J$, as well as the distribution obtained from it by
relabeling alternatives $k_1$ and $k_2$. If, at some index
configuration satisfying \eqref{partialindiff}, the first distribution
generates $P(Y=k_1|\mathbf{x})\neq P(Y=k_2|\mathbf{x})$, then $\mathcal{P}^{+,B,\boldsymbol{\beta}}_{k_1,k_2}
\cap
\mathcal{P}^{+,B,\boldsymbol{\beta}}_{k_2,k_1}
\cap
\relint\!\left(\Delta^{B,\boldsymbol{\beta}}_{J-1}\right)
\neq\varnothing.$ Hence, fixed-$\boldsymbol{\beta}$ ranking recovery fails at the
\emph{Behavioral environment} level.
\end{prop}
\textbf{Proof of Proposition
\ref{prop:distributional_no_uniform_recovery}.}
Without loss of generality, suppose that at the stated pairwise-indifferent
configuration
$P(Y=k_1|\mathbf{x})<P(Y=k_2|\mathbf{x})$.
Holding the other indices fixed, increase $x_{k_1}\beta_{k_1}$.
The difference
$P(Y=k_1|\mathbf{x})-P(Y=k_2|\mathbf{x})$ is continuous,
starts negative, and converges to $1$. Hence, at some configuration with
$x_{k_1}\beta_{k_1}>x_{k_2}\beta_{k_2}$, we have $P(Y=k_1|\mathbf{x})=P(Y=k_2|\mathbf{x})$. Relabeling $k_1$ and $k_2$ in both the unobservables distribution and
the index configuration leaves this probability vector unchanged but
reverses the deterministic ranking. Full joint index difference
variation makes both configurations attainable for the same
$\boldsymbol{\beta}$.

It remains only to note that the common probability vector is in the relative interior. Under the fixed full support distribution, the map from normalized index differences to choice probabilities is continuous and injective, that is, if two index vectors differ, the choice probability of the set of alternatives whose relative indices increase changes strictly. By invariance of domain, its image is open in $\relint(\Delta_{J-1})$. Hence the constructed probability vector lies in $\relint(\Delta^{B,\boldsymbol{\beta}}_{J-1})$. \hfill $\blacksquare$

\end{document}